\documentclass[11pt]{article}
\usepackage{geometry}                
\usepackage[parfill]{parskip}    
\usepackage[pdftex]{graphicx,color}
\usepackage{amssymb,amsmath,amsthm}
\usepackage{epstopdf}
\usepackage{wrapfig}
\usepackage{amssymb,amsmath,amsthm}
\numberwithin{equation}{section}

\usepackage{epstopdf}
\usepackage{wrapfig}
\usepackage{graphicx}
\usepackage{arydshln}
\usepackage{setspace}

\usepackage{mathtools}
\mathtoolsset{showonlyrefs}

\usepackage{blkarray, bigstrut}

\usepackage{marvosym}
\usepackage{tikz}
\usepackage{verbatim}

\usepackage{tikz}

\usepackage{tikz}
\usetikzlibrary{arrows}
\usetikzlibrary{decorations.pathmorphing}
\usetikzlibrary{decorations.markings}
\usetikzlibrary{patterns}
\usetikzlibrary{automata}
\usetikzlibrary{positioning}
\usepackage{tikz-cd}
\tikzset{->-/.style={decoration={
			markings,
			mark=at position #1 with {\arrow{latex}}},postaction={decorate}}}

\tikzset{-<-/.style={decoration={
			markings,
			mark=at position #1 with {\arrowreversed{latex}}},postaction={decorate}}}

\usetikzlibrary{shapes.misc}\tikzset{cross/.style={cross out, draw, 
		minimum size=2*(#1-\pgflinewidth), 
		inner sep=0pt, outer sep=0pt}}

\usepackage{caption}
\usepackage{subcaption}

\usetikzlibrary{quotes,angles}

\usepackage[pdftex,bookmarks,colorlinks,breaklinks]{hyperref}
\definecolor{dullmagenta}{rgb}{0.4,0,0.4}   
\definecolor{darkblue}{rgb}{0,0,0.4}
\hypersetup{linkcolor=red,citecolor=blue,filecolor=dullmagenta,urlcolor=darkblue}

\def\red#1{\textcolor[rgb]{0.9, 0, 0}{#1} }
\def\blue#1{\textcolor[rgb]{0,0,1.0}{#1}}
\def\dblue#1{\textcolor[rgb]{0,0,0.7}{#1}}

\newcommand{\ii}{{\rm i}}
\newcommand{\ee}{{\rm e}}
\newcommand{\dd}{{\rm d}}
\newcommand{\CC}{{\mathbb C}}
\newcommand{\RR}{{\mathbb R}}

\def\re{\mathop{\rm Re}}
\def\im{\mathop{\rm Im}}

\newtheorem{thm}{Theorem}[section]
\newtheorem{prop}[thm]{Proposition}

\newtheorem{lemma}[thm]{Lemma}

\newtheorem{remark}{Remark}

\renewcommand{\descriptionlabel}[1]%
         {\dblue{#1:}\\}

\title{Local Statistics in Normal Matrix Models with Merging Singularity}

\author{Torben Kr\"uger\thanks{Torben Kr\"uger was supported by the Novo Nordisk Fonden Project Grant 0064428 and the VILLUM FONDEN via Young Investigator Award (Grant No. 29369).}, Seung-Yeop Lee, Meng Yang\thanks{Meng Yang was partially supported by the Novo Nordisk Fonden Project Grant 0064428, the VILLUM FONDEN via Young Investigator Award (Grant No. 29369), the Research Grant RCXMA23007, and the Start-up funding YJKY230037 at Great Bay University..}}

\date{}                                           

\begin{document}
\maketitle

\begin{abstract}
We study the normal matrix model, also known as the two-dimensional one-component plasma at a specific temperature,   with merging singularity. As the number $n$ of particles tends to infinity
we obtain
the limiting local correlation kernel at the singularity, which is related to the parametrix of the Painlev\'e~II equation. The two main tools are Riemann-Hilbert
problems and the generalized Christoffel-Darboux identity. The 
correlation kernel exhibits a novel anisotropic scaling behavior, where the corresponding
spacing scale of particles is $n^{-1/3}$ in the direction of merging and $n^{-1/2}$ in the perpendicular direction. 
In the vicinity at different distances to the merging singularity we also observe Ginibre bulk and edge statistics, as well as the sine-kernel and the universality class corresponding to the elliptic ensemble in the weak non-Hermiticity regime for the local correlation function.

\end{abstract}

\section{Introduction and Main Results}
The normal matrix model with external potential $Q:\CC\to\RR$ is an ensemble of $n$ particles, $\{z_j\}_{j=1}^n\subset \CC$, distributed according to the probability density
\begin{equation}\label{def distribution}
\rho_n(z_1,\dots,z_n)=\frac{1}{Z_{n}}\prod_{i<j}|z_i-z_j|^2\ee^{-N\sum_{j=1}^nQ(z_j)},    \end{equation} where $Z_{n}$ is the normalization constant and $N$ is a positive real parameter. The name, normal matrix model, reflects that \eqref{def distribution} is the joint density of complex eigenvalues of random normal matrices (see e.g. \cite{Chau 1998, Chau 1992}).  For  example, the choice $Q(z) =|z|^2$ yields the complex Ginibre ensemble, i.e., the eigenvalue distribution of an $n \times n$ random matrix with independent  complex standard normal-distributed entries \cite{Ginibre 1965}.  Another interpretation of \eqref{def distribution} is as the Gibbs measure of the two-dimensional one-component plasma (OCP) at a specific  temperature. The OCP is also called  Coulomb gas or $\log$-gas. It consists of $n$ particles in the potential $Q$ subject to a logarithmic pair-interaction and plays a role in the fractional quantum Hall effect \cite{Arovas 1984, Laughlin 1983}, the theory of superfluid-superconducting films \cite{Minnhagen 1987} and two-dimensional turbulence \cite{Frohlich 1982,Marchioro 1994,Onsager 1949}.   For an extensive overview of the connection between  eigenvalues of random matrices and $\log$-gases we refer to \cite{Forrester 2010,byun 2025}.

One of the fundamental questions in the interacting particle system is to understand the correlations among the particles. This correlation is characterized by the $k$-point correlation function 
$$R_{n,N}^{k}(z_1,\dots,z_k):=\frac{n!}{(n-k)!}\int_{\CC^{n-k}}\rho_n(z_1,\dots,z_n)\prod_{j=k+1}^n\dd A(z_j),$$ where $\dd A$ is the Lebesgue area measure of the complex plane.  This function encodes the marginal probability for $k$ particles.

The $k$-point correlation function can be  written as 
\begin{equation}\label{eq determinantal}
 R_{n,N}^{k}(z_1,\dots,z_k)=\det\big[{\bf K}_{n,N}(z_i,z_j)\big]_{i,j=1}^k 
\end{equation}
(see e.g. \cite{Deift 1999}), where the correlation kernel ${\bf K}_{n,N}(z,{\zeta})$  is given by
\begin{equation}\label{def kernel}
    {\bf K}_{n,N}(z,{\zeta}):=\ee^{-\frac{N}{2}Q(z)-\frac{N}{2}Q(\zeta)}\sum_{k=0}^{n-1}\frac{1}{h_k}p_{k,N}(z)\overline{p_{k,N}(\zeta)},\end{equation}
in terms of the monic orthogonal polynomial with degree $n$, subject to the orthogonality conditions
\begin{equation}\label{eq ortho4}
\int_\CC p_{n,N}(z)\,\overline{p_{m,N}(z)}\ee^{-NQ(z)}\,\dd
A(z)=h_n\delta_{nm},\quad n,m\geq 0,\end{equation} where $h_n$ is the positive norming constant and $\delta_{nm}$ is the Kronecker delta.
These orthogonal polynomials satisfy 
$$p_{n,N}(z)={\mathbb E}\prod_{j=1}^n(z-z_j),$$
where the expectation is taken over the probability distribution in \eqref{def distribution}.

Equation \eqref{eq determinantal} implies that the correlation kernel ${\bf K}_{n,N}$ determines the statistical behavior of the particles.
In the scaling limit of $n\to\infty$, while $n/N$ is fixed to a constant, the correlation kernel ${\bf K}_{n,N}$, with the proper scaling of coordinates, often converges to a universal function. 
It is expected that this universal function only depends on the symmetry class of the model and the local behaviour of the limiting density $\lim_{n,N\to\infty} R^1_{n,N}$, i.e. on whether the density vanishes and if so on its vanishing order. 
  Below, we describe several known results about the asymptotic behaviors of the correlation kernels.

In contrast to the $2$-dimensional setting decribed above, in  unitary ensembles, all the eigenvalues are confined on the real axis, and the joint distribution of eigenvalues is proportional to 
\begin{equation*}
\prod_{i<j}(x_i-x_j)^2\ee^{-N\sum_{j=1}^nQ(x_j)}    \prod_{k=1}^n\dd x_k,
\end{equation*} 
with $x_i \in \RR$. 
If $Q$ is real analytic and grows sufficiently fast at infinity, the eigenvalues are confined to a union of intervals \cite{Deift 1998}.  It is known that the local scaling limits of the correlation kernel are universal in the sense that they only depend on the vanishing order of the limiting density $\sigma(x):=\lim_{n\to\infty}\frac{1}{n}R_{n,N}^1(x)$. By \cite{DKMVZ 19992, Bleher 1999, Pastur 1997}, in the bulk of the spectrum, the universal correlation kernel at the microscopic scale is the Sine kernel;  at the edge of the spectrum, the limiting density of the eigenvalues typically vanishes as a square root, and the universal correlation kernel at the microscopic scale is the Airy kernel \cite{Pastur 2003}. In the critical case studied in \cite{Bleher 2003, Claeys 2006}, where the spectral density vanishes quadratically in the interior of the asymptotic spectrum,
the universal limiting kernel can be written in terms of the Hastings-McLeod solution of Painlev\'e II equation,  a.k.a. the Painlev\'e II kernel. 

In a matrix ensemble with an external source, different universal limiting kernels emerge at points of vanishing spectral density within the interior of its support. In \cite{Tracy 2006,Hikami 1998,Hikami 19982,Erdos 2020}, when the spectral density vanishes like a cubic root, the universal limiting kernel can be written in terms of the Pearcey integrals, a.k.a. the Pearcey kernel \cite{Pearcey 1946}.  

In normal matrix models where eigenvalues are complex valued, as $N$ tends to infinity proportional to the number of particles $n$, the limiting density $\sigma(z):=\lim_{n,N\to\infty}R^1_{n,N}(z)$ is supported on a compact set ${\cal S}\subset \CC$, which we call the droplet \cite{HM 2013}.

In the bulk of the droplet, i.e., at $z^*\in{\rm Int}({\cal S})$ where $\sigma(z^*)
\neq 0$, the universal scaling limit of the correlation kernel is given \cite{Ameur 2011} by
\begin{equation}\label{ginibre kernel}
\lim_{n\to\infty}\frac{1}{n\Delta Q(z^*)}{\bf K}_{n,n}\bigg(z^*+\frac{\nu}{\sqrt{n\Delta Q(z^*)}},z^*+\frac{\eta}{\sqrt{n\Delta Q(z^*)}}\bigg)=G(\nu,\eta):=\ee^{\nu\overline \eta-\frac{|\nu|^2}{2}-\frac{|\eta|^2}{2}}.
\end{equation} 
Here $G(\nu,\eta)$ is also known as the Ginibre kernel \cite{Ginibre 1965}. Recently in \cite{Hedenmalm 2017} the universality has also been shown at the boundary $z^*\in\partial{\cal S}$, for some general class of potentials $Q$ where the boundary of the droplet ${\cal S}$ is a smooth Jordan curve.
\begin{equation}\label{Faddeeva kernel}
\lim_{n\to\infty}\frac{1}{n\Delta Q(z^*)}{\bf K}_{n,n}\bigg(z^*+\frac{\ee^{\ii \theta}\nu}{\sqrt{n\Delta Q(z^*)}},z^*+\frac{\ee^{\ii\theta}\eta}{\sqrt{n\Delta Q(z^*)}}\bigg)=G(\nu,\eta)\frac{1}{2}{\rm erfc}\Big(\frac{\nu+\overline \eta}{\sqrt 2}\Big).\end{equation} 
Here the angle $\theta$ is that of the outer normal direction to the boundary of the droplet at $z^*$, and ${\rm erfc}$ denotes the complementary error function. The kernel on the right hand side is also known as the Faddeeva kernel \cite{Faddeyeva 1961}. Similar edge universality results have also been established for matrices with i.i.d. non-Gaussian entries in \cite{Erdos 2021,Campbell 2024}.

For the two-dimensional OCP at inverse temperatures different from the one corresponding to the normal matrix model, universality results analogous to \eqref{ginibre kernel} and \eqref{Faddeeva kernel} are not known. However, central limit theorems for mesoscopic linear statistics have recently been established in \cite{Leble 2018} and \cite{Bauerschmidt 2017}.

In this paper we obtain the limiting correlation kernel around the point in the droplet boundary where the merging singularity occurs, see Figure \ref{Figure 1}. 
The universal behavior at such critical points have been open problems in normal matrix models while, in unitary ensembles, where the full hierarchy of critical behaviors  have been investigated \cite{Claeys 2010,Claeys 2011,Claeys 2007}.  It was expected and conjectured by physicists that  similar  behavior will emerge in the normal matrix ensemble \cite{Teodorescu 2005,Bettelheim 2005,Lee 2006}.  The current paper confirms this conjecture and the appearance of the Painlev\'e~II kernel at the merging criticality in the microscopic scale of  $\sim n^{-1/3}$.
A remarkable feature of the particle spacing at the criticality is its anisotropic nature. While the particles are spaced at a distance $\sim n^{-1/3}$ in the direction of merging, they are concentrated on the much smaller scale $\sim n^{-1/2}$ in the perpendicular direction.

We also observe a rich scaling behaviour for the particle density. This density vanishes quadratically in the  direction of merging at the critical point $b_c$, giving rise to the  Painlev\'e~II kernel. However, the characteristic  scale on which the density varies from its value at $b_c$ to its value in the bulk is $n^{-1/4}$, i.e. larger than the scale, $n^{-1/3}$, of inter-particle distance at the critical point. This is in sharp contrast to the regular edge behaviour \eqref{Faddeeva kernel}, where the particle density varies on the same  scale on which the particles are spaced. The separation of characteristic lengths for the density and particle spacing causes the appearance of additional universal particle correlations in a mesoscopic distance from the merging singularity. Letting $l$ be the distance from the merging singularity, Sine-kernel statistic emerges at $n^{-1/3}\ll l \ll n^{-1/4}$.    For $l = n^{-1/4}$ the local universality class coincides with the one observed in the weak non-Hermiticity regime of the elliptic ensemble \cite{Fyodorov 1982, Akemann 2018}. Beyond this distance, i.e. for $l \gg n^{-1/4}$, the familiar bulk and edge behaviours from \eqref{ginibre kernel} and \eqref{Faddeeva kernel} are observed.

In this paper we consider the  external potential,
\begin{equation}\label{def Q}
   Q(z)=|z|^2+2c\log\frac{1}{|z-a|},\quad a>0,\quad c>0. 
\end{equation}
This model was introduced in \cite{Ba 2015} to study the strong asymptotics of the orthogonal polynomials $p_{n,N}(z)$. For the case when $c$ is an integer, the corresponding orthogonal polynomials were also studied in \cite{Akemann 2003}.  Below we review some useful facts from \cite{Ba 2015}.

Defining $t:=n/N$, the droplet ${\cal S}$ depends on the parameters $a, c$, and $t$. When $t=t_c$ where $t_c:=a(a+2\sqrt c)$ the droplet undergoes a topological transition from genus 0 to genus 1 with the merging singularity at $z=b_c$ where $b_c:=a+\sqrt c$; see Figure \ref{Figure 1}.


\begin{figure}
\begin{center}
\includegraphics[width=0.4\textwidth]{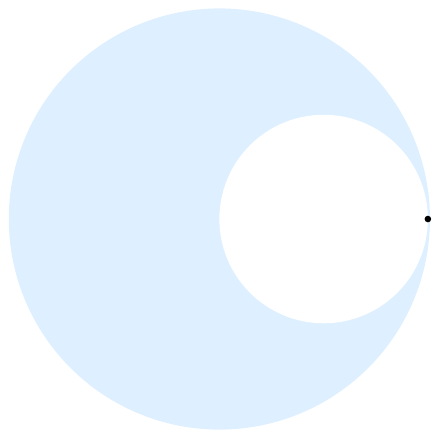}
\end{center}
\caption{The droplet ${\cal S}$ (shaded region) and the merging point $b_c$ (black dot) when $t=t_c$ for $a=1$ and $c=1$.  The center of the inner (white) disk is $a=1$.}\label{Figure 1}
\end{figure}

To state our main result, we  recall the Painlev\'e II Riemann-Hilbert problem \cite{Claeys 2006,Ba 2015} for $\widetilde\Psi:\CC\to~\CC^{2\times 2}$:
\begin{equation}\label{rhp phi}
\begin{cases}
 \widetilde\Psi_+(\xi;s)= \widetilde\Psi_-(\xi;s)\begin{pmatrix}
1&0\\1&1
\end{pmatrix},&\quad \xi\in \gamma_+,\\
\widetilde\Psi_+(\xi;s)=  \widetilde\Psi_-(\xi;s)\begin{pmatrix}
1&-1\\0&1
\end{pmatrix},&\quad \xi\in \gamma_-,\\
\widetilde\Psi(\xi;s)=\Big(I+\frac{\Pi_1(s)}{2\ii\xi}+{\cal O}\Big(\frac{1}{\xi^2}\Big)\Big)\ee^{-\ii\left(\frac{4}{3}\xi^3+s\xi\right)\sigma_3},&\quad |\xi|\to\infty,\\
\widetilde\Psi(\xi;s)\mbox{\quad is holomorphic}, &\mbox{otherwise},
\end{cases}
 \end{equation}
 where $s\in \RR$ is a parameter, $\sigma_3=\begin{pmatrix}
1&0\\0&-1
\end{pmatrix}$ is the third Pauli matrix, $\gamma_\pm$ are the piecewise straight contours shown in Figure \ref{Figure 2}, and $\Pi_1(s)$ is given in terms of $q(s)$ by
\begin{equation}\label{def Pi}
\Pi_1(s)= \begin{pmatrix}
r(s)&q(s)\\-q(s)&-r(s)
\end{pmatrix},\quad  r(s)=q'(s)^2-sq(s)^2-q(s)^4.  
\end{equation} 
The unique solution to this Riemann-Hilbert problem provides $q(s)$, which is the Hastings-McLeod solution of the Painlev\'e II equation $q''=sq+2q^3$; see \cite{Bleher 2003,Claeys 2006,Flaschka 1980} for more details. 

\begin{figure} 
\centering
\begin{tikzpicture}[scale=0.85]
\draw[thick,postaction={decorate, decoration={markings, mark = at position 0.5 with {\arrow{>}}}} ] (0,0)  -- (4,2) ;
\draw[thick,postaction={decorate, decoration={markings, mark = at position 0.5 with {\arrow{>}}}} ] (-4,2) -- (0,0);
\draw[thick,postaction={decorate, decoration={markings, mark = at position 0.5 with {\arrow{>}}}} ] (0,0) -- (4,-2) ;
\draw[thick,postaction={decorate, decoration={markings, mark = at position 0.5 with {\arrow{>}}}} ] (-4,-2) -- (0,0);
\foreach \Point/\PointLabel in {(0,0)/0}
\draw[fill=black] \Point circle (0.05) node[below] {$\PointLabel$};

 \draw[thick]
    (4,-2) coordinate (a) node[right] {}
    -- (0,0) coordinate (b) node[below] {}
    -- (4,2) coordinate (c) node[above right] {}
    pic["$\frac{\pi}{3}$", draw=black,  angle eccentricity=1.4, angle radius=0.5cm]
    {angle=a--b--c};

\foreach \Point/\PointLabel in {(3.8,1.8)/\gamma_+,(-4,1.8)/\gamma_+,(-4,-2.1)/\gamma_-,(3.8,-2.1)/\gamma_-}
\draw[fill=black]  
 \Point node[below right] {$\PointLabel$};
 \end{tikzpicture}
 \caption{The jump contours of $\widetilde\Psi(\xi;s)$. }\label{Figure 2}
 \end{figure}
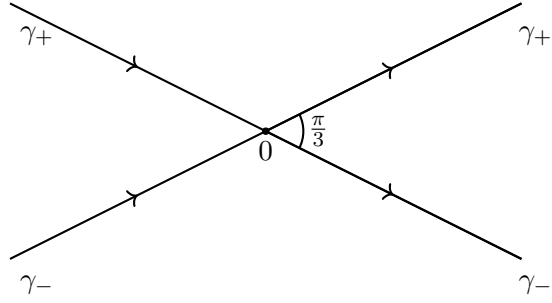


Let us define an analytic continuation, $\Psi(\xi;s)$, of the matrix function $\widetilde\Psi(\xi;s)$ on the positive imaginary axis by
\begin{equation}\label{def of psi0}
\begin{aligned}
\Psi(\xi;s):=\begin{cases}
\widetilde\Psi(\xi;s),  & \frac{\pi}{6}<\arg \xi<\frac{5\pi}{6},\vspace{0.1cm}\\
\widetilde\Psi(\xi;s)\begin{pmatrix}
0&-1\\1&1
\end{pmatrix},  &\frac{7\pi}{6}<\arg \xi<\frac{11\pi}{6},\\
\widetilde\Psi(\xi;s)\begin{pmatrix}
1&0\\1&1
\end{pmatrix},  &{\mbox{otherwise}}.
\end{cases}
\end{aligned}\end{equation}

\begin{thm}\label{main theorem}
Let $t_c:=a(a+2\sqrt c)$ and $b_c:=a+\sqrt c$. Furthermore, let $n,N\to\infty$ such that $t:=n/N$ satisfies $t=t_c+{\cal O}(N^{-2/3})$. Let \begin{equation}\label{def gammacl}
s:=\frac{\gamma_cN^{2/3}}{2b_c}(t-t_c), \quad \gamma_c:=\frac{2b_c^{1/3}c^{1/6}}{a^{1/3}}.
\end{equation} 
For any $0<\delta<1/6$, the correlation kernel in \eqref{def kernel} exhibits the asymptotic behaviour 
\begin{equation}\label{eq limiting kernel}
\frac{C_N(y,y')\gamma_c}{N^{5/6}}{\bf K}_{n,N}\bigg(b_c+\frac{x}{\sqrt N}+\frac{\ii \gamma_c y}{N^{1/3}},b_c+\frac{x'}{\sqrt N}+\frac{\ii \gamma_c y'}{N^{1/3}}\bigg)={\bf K}_s(x,y,x',y')+{\cal O}\Big(\frac{1}{N^{1/6-\delta}}\Big),    \end{equation}
uniformly over a region where $(N^{-\delta}x,y,N^{-\delta}x',y')\in \RR^4$ is bounded.

Here $C_N(y,y')$ is defined by
      \begin{equation*}
   C_N(y,y'):=
   \ee^{-\ii (\varphi_N(y)-\varphi_N(y'))}\,, \qquad \varphi_N(w):=a\gamma_cN^{2/3} w+s w+4w^3/3.
      \end{equation*}
The limiting kernel is given by  
\begin{equation}\label{main results3}
 {\bf K}_s(x,y,x',y'):=\begin{cases}\displaystyle
 \frac{\ee^{-(x^2+(x')^2)}}{\sqrt{\pi/2}}\frac{\Psi_{21}(y;s)\Psi_{11}(y';s)-\Psi_{11}(y;s)\Psi_{21}(y';s)}{2\pi\ii(y-y')},\quad y\neq y',\\\displaystyle
 \frac{\ee^{-(x^2+(x')^2)}}{\sqrt{\pi/2}}\frac{\Psi_{21}'(y;s){\Psi_{11}(y;s)}-\Psi_{11}'(y;s){\Psi_{21}(y;s)}}{2\pi\ii},\quad y=y'.
 \end{cases} 
    \end{equation}
     Here $\Psi_{11}$ and $\Psi_{21}$ are corresponding entries of $\Psi$ in \eqref{def of psi0}.  The symbol `` $'$ '' in $\Psi'_{21}$ and $\Psi'_{11}$  stands for the derivative with respect to $y$. 
\end{thm}

We note that $C_N(y,y')$ drops out when computing the $k$-point correlation function \eqref{eq determinantal}, hence it is irrelevant to the statistical behavior. 

This is the first observation of the limiting kernel at criticality in the normal random matrix ensemble.  The zooming scales are anisotropic: $N^{1/3}$ in the direction of the merging (represented by the coordinates $y$ and $y'$) and $N^{1/2}$ in the direction perpendicular to the merging (represented by the coordinate $x$ and $x'$).     If we instead use the zooming scale of $N^{1/3}$ for both directions, the Gaussian factor $\ee^{-x^2}$ in $x$ (and similarly for $x'$) will create the $\delta$-function in the large $N$ limit, and the limiting process will be concentrated on a 1-dimensional space represented by $y$ and $y'$, producing the same limiting process as in the 1-matrix model studied in \cite{Bleher 2003, Claeys 2006}.  This supports the conjecture made by physicists \cite{Lee 2006} that the normal matrix model has the similar universal critical behaviors as the ones observed in the unitary matrix model.



\begin{figure}
\centering
    \includegraphics[width=\textwidth]{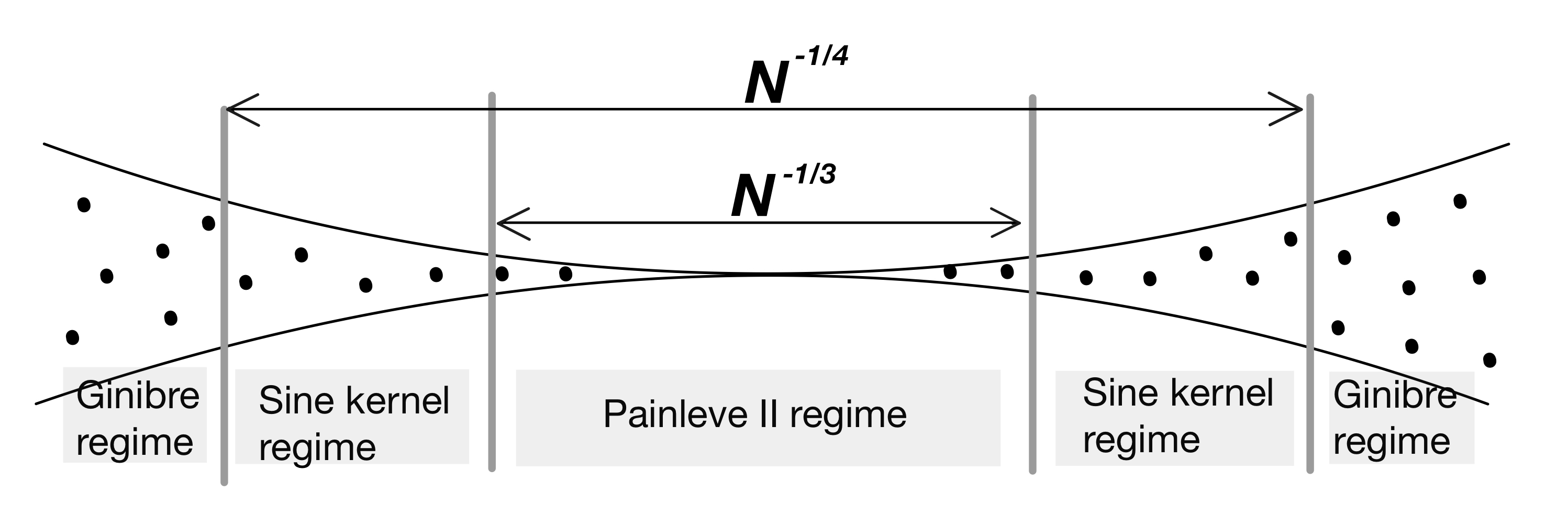} 
\caption{Various scaling behaviors appear depending on the distance from the merging point. This figure is only qualitative.  For example, in Figure \ref{Figure 1}, the two boundaries of the global droplet bend towards the same direction.  In this figure the droplet is rotated by $90^{\circ}$.  See Remark \ref{remark2}.}\label{Fig 4}
\end{figure}

The following theorem provides a detailed description of ${\bf K}_{n,N}$ at the merging singularity on mesoscopic scales, that interpolates between the singularity and the bulk regime. 

Below we introduce the parameter $\tau$ to represent the scaled distance, $\sim N^{\tau}$, from the singularity.  

\begin{thm}\label{main theorem3}
Let $t_c:=a(a+2\sqrt c)$ and $b_c:=a+\sqrt c$. Let $n,N\to\infty$ such that $t=n/N$ and $t=t_c+{\cal O}(N^{-2/3})$. 

If $\frac{1}{6}<\tau<\frac{1}{4}$, the following limit is uniform over $(\nu,\eta)$ in a compact subset of $\CC^2$ and $(X,Y)$ in a compact subset of $\RR^2\setminus\{Y=0\}$, 
\begin{equation}\label{eq main thm3}
\begin{aligned}
&\quad\lim_{n,N\to\infty}\frac{C_{N,\tau}(X\sqrt N/N^{2\tau},Y,\nu,\eta)}{n}{\bf K}_{n,N}\bigg(b_c+\frac{X}{\sqrt {2}N^{2\tau}} +\frac{\ii(2cb_c^2)^{\frac{1}{4}}Y}{N^{\tau}}+\frac{\nu}{\sqrt N},b_c+\frac{X}{\sqrt {2}N^{2\tau}} +\frac{\ii(2cb_c^2)^{\frac{1}{4}}Y}{N^{\tau}}+\frac{\eta}{\sqrt N}\bigg)\\
&=\begin{cases}\displaystyle
0, &X+b_cY^2<0 \,\, \mbox{or}\,\,  X+\sqrt cY^2>0,\\\displaystyle
    \frac{G(\nu,\eta)}{\pi t_c},&X+\sqrt cY^2<0<X+b_cY^2, \\\displaystyle
 \frac{G(\nu,\eta)}{2\pi t_c}{\rm erfc}\Big(\frac{\nu+\overline \eta}{\sqrt 2}\Big),&X+b_cY^2=0 \,\, \mbox{or}\,\, X+\sqrt cY^2=0,   
\end{cases}    
\end{aligned}
  \end{equation}
where $G(\nu,\eta)$ is the Ginibre kernel defined in \eqref{ginibre kernel},
$C_{N,\tau}(X,Y,\nu,\eta)$ is defined by
    \begin{equation}\label{def CNXY}
    \begin{aligned}
C_{N,\tau}(X,Y,\nu,\eta):&=\ee^{\frac{ (\overline\nu-\nu+\eta-\overline\eta)( X+\sqrt 2\sqrt N b_c)}{2\sqrt 2}}\ee^{\frac{\nu+\overline\nu-\eta-\overline\eta}{2}\frac{\ii(2cb_c^2)^{\frac{1}{4}} Y\sqrt N}{N^{2\tau}} }\frac{(\sqrt c+u(\nu))^{\frac{Nc}{2}}}{(\sqrt c+\overline {u(\nu)})^{\frac{Nc}{2}}}\frac{(\sqrt c+  \overline {v(\eta)})^{\frac{Nc}{2}}}{(\sqrt c+ v(\eta))^{\frac{Nc}{2}}},    \\
u(\nu)&:=\frac{\ii(2cb_c^2)^{\frac{1}{4}}Y}{N^\tau}+\frac{(X+\sqrt2\nu)}{\sqrt {2N}} ,\quad v(\eta):=\frac{\ii(2cb_c^2)^{\frac{1}{4}}Y}{N^\tau}+\frac{ (X+\sqrt2\eta)}{\sqrt {2N}},
    \end{aligned}
      \end{equation}
and {\rm erfc} is the complementary error function given by
$${\rm erfc}(z):=\frac{2}{\sqrt\pi}\int_z^\infty\ee^{-t^2}\dd t.$$
      
If $\frac{1}{4}<\tau<\frac{3}{10}$, the following limit is uniform over $(Y,x,y,x',y')$ in a compact subset of $\RR^5\setminus\{Y=0\}$,
 \begin{equation}\label{eq main thm32}
 \begin{aligned}
&\quad \lim_{n,N\to\infty}\frac{\widehat C_N(Y,x,y,x',y')N^{2\tau}}{n\sqrt N}{\bf K}_{n,N}\bigg(b_c+\frac{\ii(2cb_c^2)^{\frac{1}{4}}Y}{N^{\tau}}+\frac{\ii \sqrt 2\pi y}{N^{1-2\tau}aY^2}+\frac{x}{\sqrt {N}}-\frac{t_cY^2}{2\sqrt {2}aN^{2\tau}}+\frac{s}{\sqrt 2 \gamma_c N^{2/3}},\\
&\qquad \qquad\qquad \qquad\qquad \qquad \qquad \qquad\qquad b_c+\frac{\ii(2cb_c^2)^{\frac{1}{4}}Y}{N^{\tau}}+\frac{\ii \sqrt 2\pi y'}{N^{1-2\tau}aY^2}+\frac{x'}{\sqrt {N}}-\frac{t_cY^2}{2a\sqrt {2}N^{2\tau}}+\frac{s}{\sqrt 2 \gamma_c N^{2/3}}\bigg)\\
&=\frac{ aY^2\ee^{-x^2-(x')^2} }{ \pi^{3/2}t_c}\frac{\sin \big(\pi(y-y')\big) }{\pi(y-y')} ,    
 \end{aligned}
\end{equation} 
where $\widehat C_N(Y,x,y,x',y')$ is defined in terms of $C_{N,\tau}$ \eqref{def CNXY} by
      \begin{equation*}
\widehat C_{N}(Y,x,y,x',y'):=C_{N,\tau}\bigg(-\frac{b_c+\sqrt c}{2}\frac{Y^2\sqrt N}{N^{2\tau}}+\frac{s}{\sqrt 2 \gamma_c N^{1/6}},Y,x+\frac{\ii \sqrt 2\pi y}{aY^2}\frac{N^{2\tau}}{\sqrt N},x'+\frac{\ii \sqrt 2\pi y'}{aY^2}\frac{N^{2\tau}}{\sqrt N}\bigg).
      \end{equation*}

      For the  case when $\tau=\frac{1}{4}$, the following limit is uniform over  $(\nu,\eta)$ in a compact subset of $\CC^2$ and $(X,Y)$ in a compact subset of $\RR^2\setminus\{Y=0\}$,
\begin{equation}\label{thm1.32}
\begin{aligned}
&\quad \lim_{n,N\to\infty}\frac{C_{N,1/4}(X,Y,\nu,\eta)}{n}{\bf K}_{n,N}\bigg(b_c+\frac{X}{\sqrt {2N}} +\frac{\ii(2cb_c^2)^{\frac{1}{4}}Y}{N^{\tau}}+\frac{\nu}{\sqrt N},b_c+\frac{X}{\sqrt {2N}} +\frac{\ii(2cb_c^2)^{\frac{1}{4}}Y}{N^{\tau}}+\frac{\eta}{\sqrt N}\bigg)\\
&=\frac{G(\nu,\eta) }{2 \pi t_c}\bigg({\rm erfc}\Big(X+\sqrt c Y^2+\frac{\nu+\overline{\eta}}{\sqrt 2}\Big)-{\rm erfc}\Big(X+b_c Y^2+\frac{\nu+\overline{\eta}}{\sqrt 2}\Big)\bigg),
\end{aligned}
\end{equation} 
where 
$C_{N,1/4}(X,Y,\nu,\eta)$ is defined in \eqref{def CNXY} and $G(\nu,\eta)$ is defined in \eqref{ginibre kernel}.
\end{thm} 

Again, the prefactors $C_{N, \tau}$ drop out when determining the correlations functions and, thus, do not influence the particle statistics.

\begin{remark} {\bf  Gaps in $\tau$ range:} \label{Remark 1}
Theorem \ref{main theorem3} does not address the ranges of $\tau$, where $0<\tau\leq 1/6$ and  $3/10 \leq\tau < 1/3$. Let us remark on these ranges.

In \eqref{eq main thm3} the two parabolas, $X=-b_cY^2$ and $X=-\sqrt cY^2$, are used to define the cases.  They are the quadratic approximations of the boundary of the droplet near the critical merging point.  For the case of $0<\tau\leq \frac{1}{6}$, which corresponds to moving away from the merging point, we expect the behaviour from \eqref{eq main thm3} to remain valid if one uses the better approximation of the droplet boundaries and an accompanying modification of the prefactor $C_{N, \tau}$.  This means that, instead of using the parabolas to divide the cases in \eqref{eq main thm3}, we will need a better approximation of the circles describing the droplet boundary in terms of the higher order terms as $X=-b_cY^2 + A_4 Y^4+ A_6 Y^6+\dots$ where $A_4$, $A_6,\dots$ are some (scaling) coefficients.    

For the other range of $3/10 \leq\tau < 1/3$, we expect that the limiting behaviour established for $1/4<\tau<3/10$ in \eqref{eq main thm32}  remains valid. To prove this claim we need higher order corrections of orthogonal polynomial to approach deeper (i.e. larger $\tau$) into the critical region.  The reason for our belief is based on the following heuristic check. If one assumes that the asymptotic behavior at $\tau=1/3$ in Theorem 1.1 holds for smaller values of $\tau$, one can in fact obtains  the asymptotic behavior in \eqref{eq main thm32} by using the known asymptotics of Painlev\'e II transcendents.  To be more concrete, one can substitute $x$ and $y$ in \eqref{eq limiting kernel} respectively by 
$$ x-\frac{b_c+\sqrt c}{2\sqrt 2}\frac{Y^2\sqrt N}{N^{2\tau}}+\frac{s}{\sqrt 2 \gamma_c N^{1/6}} \quad \text{and} \quad \frac{(2cb_c^2)^{\frac{1}{4}}Y N^{1/3}}{\gamma_c N^\tau}+\frac{\sqrt 2\pi y}{\gamma_c aY^2}\frac{N^{2\tau} }{N^{2/3}},$$ 
and similarly for $x'$ and $y'$, to eventually obtain the asymptotic behavior of \eqref{eq main thm32}.  Note that this substitution is not covered by the theorem because $y$ in \eqref{eq limiting kernel} is assumed to be bounded.   However such matching suggests that the result that we obtained for $1/4<\tau<3/10$, being heuristically reproduced from the result at $\tau=1/3$, may hold for all the intermediate region $3/10\leq \tau <1/3$.
\end{remark}

\begin{remark}\label{remark2}{\bf Rich Scaling Behaviour Near the Merging Singularity:}
In contrast to the 1 dimensional Coulomb gas, the merging singularity here exhibits an additional characteristic mesoscopic scale $N^{-1/4}$, greater than the microscopic length scale $N^{-1/3}$ of particle spacing at the singularity. 
On this mesoscopic length scale  the particle density has a non-trivial profile and drops from its bulk value down to zero. Such rich scaling behaviour is absent at regular edge points of normal matrix models.
 For at the distance $d\sim N^{-\tau}$ from the merging point $b_c$
 with $\frac{1}{4}<\tau<\frac{3}{10}$ the local correlation of the particles is described by the sine kernel \eqref{eq main thm32} as in the GUE. See Figure \ref{Fig 4} for the regions where the sine kernel appears. In this regime the  inter-particle distance in the $y$-direction scales as $\sim N^{2\tau-1}$, because the particle density is of size $d^2 \sim N^{-2\tau}$ and the particle system is effectively one dimensional, due to the Gaussian decay in the $x$-direction.  One can also see why the transition from Sine-kernel to the Ginibre regime occurs at $\tau=1/4$.  At $\tau=1/4$ the width of the support of the equilibrium measure becomes of order $\sim N^{-1/2}$, matching the inter-particle distance of the Ginibre ensemble. For $\tau < 1/4$ this width is  thick enough to accommodate the two-dimensional  particles system.
\end{remark}

For $\tau=1/4$, one can obtain the limiting density by setting $\nu=\eta=0$ in \eqref{thm1.32},
\begin{equation}\label{main results4}
 \lim_{n,N\to\infty}\frac{1}{n}R_{n,N}^1\bigg(b_c+\frac{X}{\sqrt {2 N}} +\frac{\ii(2cb_c^2)^{\frac{1}{4}}Y}{N^{1/4}}\bigg)=
\frac{1}{2 \pi t_c}\big({\rm erfc}\big( X+\sqrt c Y^2\big)-{\rm erfc}\big(X+b_c Y^2\big)\big).
\end{equation} See Figure \ref{Fig 3} for the limiting asymptotic behavior of $\frac{1}{n}R_{n,N}^1$  in \eqref{main results4}.

\begin{remark}
The limiting kernel for $\tau=\frac{1}{4}$ \eqref{thm1.32} in Theorem \ref{main theorem3} with $X=-\frac{(b_c+\sqrt c)Y^2}{2}$ also appears in the weak non-Hermiticity regime of elliptic ensembles \cite{Fyodorov 1982, Akemann 2018} and a  class of normal, almost-Hermitian random matrix ensembles \cite{Ameur 2020, Ameur 2023}.
\end{remark}

\begin{figure}
\centering
  \includegraphics[width=0.15\textwidth]{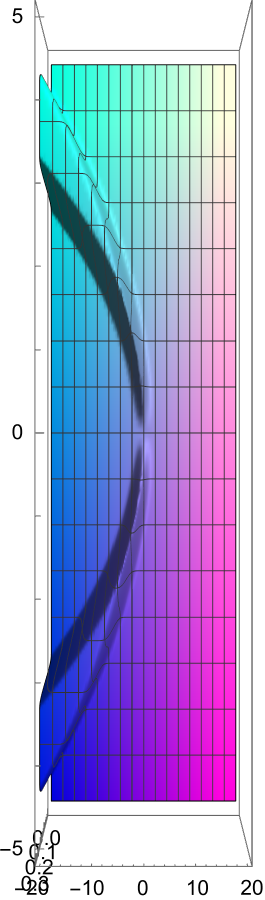} \qquad\qquad\qquad\qquad\qquad\includegraphics[width=0.4\textwidth]{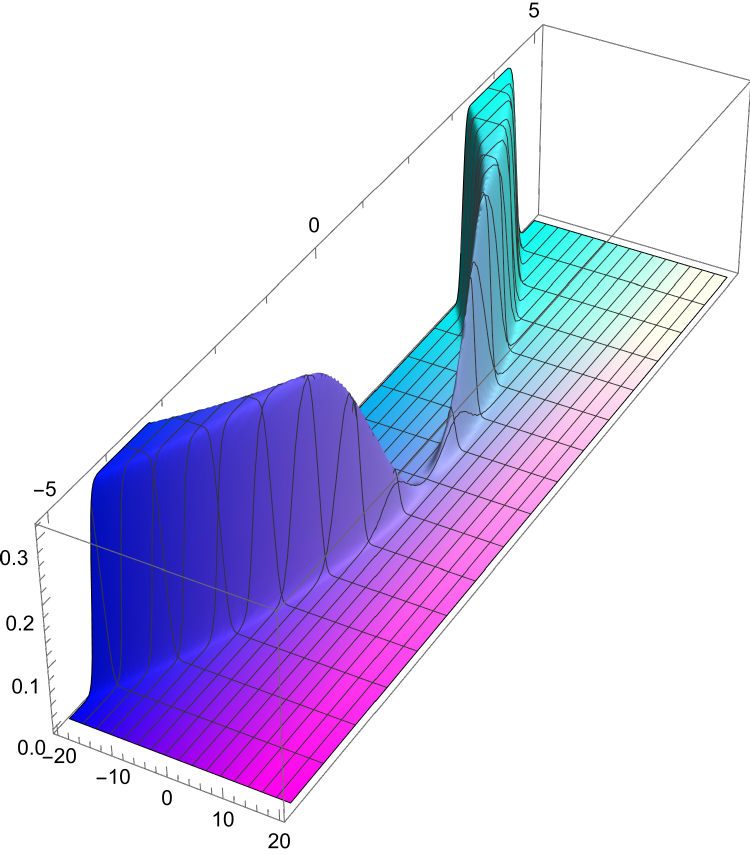}
\caption{The plot of the leading term in the right hand side of \eqref{main results4}  for $a=c=1$. Here $-20\leq X \leq 20$, $-5\leq Y \leq 5$.}\label{Fig 3}
\end{figure}

{\bf Structure of the rest of the paper:} In Section \ref{section 2}, we obtain the asymptotics of the orthogonal polynomials and the norming constants up to the first sub-leading term by using Riemann-Hilbert problems. For the proofs, the main tools are the generalized version of the Christoffel-Darboux identity (the C-D identity) and the Riemann-Hilbert problems. In sections~\ref{section 3.1} and  \ref{section 3}, we prove Theorem~\ref{main theorem}, and in Section~\ref{section 4}, we prove Theorem~\ref{main theorem3}. In Appendix~\ref{appendix a}, we obtain some relations in the Painlev\'e II Riemann-Hilbert problem. In Appendix~\ref{appendix b}, we state the proof of the Theorem~\ref{thm 21}. In Appendix~\ref{appendix c}, we state the  proof of Theorem~\ref{Thm 41}.

\section{The C-D Identity, Orthogonal Polynomials and Norming Constants}\label{section 2}
In this section, we state the generalized version of the Christoffel-Darboux identity (C-D identity), and the fine asymptotics of the orthogonal polynomials and the norming constants by using Riemann-Hilbert problems. In Subsection \ref{section 2.1} we introduce the C-D identity from \cite{byun 2021}, list the necessary facts from \cite{Ba 2015} and define the (complex) logarithmic potential, i.e. the $g$-function; In Subsection \ref{section 2.2} we determine the asymptotics of the  orthogonal polynomials $p_n$ up to the first sub-leading term. The leading order has been computed through the corresponding Riemann-Hilbert problem and its analysis in \cite{Ba 2015}. Here we make use of the ``partial Schlesinger transform'' developed in \cite{Bertola 2008} to refine this analysis and derive the asymptotic behavior of $p_n$ in the subleaing orders, which is necessary to compute the local correlation kernel in Theorems~\ref{main theorem} and \ref{main theorem3}; In Subsection \ref{section 2.3} we derive several recurrence relations through which we obtain the norming constant $h_n$ and the relations among $p_n$, $p_{n-1}$ and $p_{n+1}$.

\subsection{The C-D identity and the $g$-function}\label{section 2.1}

We recall that the monic orthogonal polynomial $p_n(z)=p_{n,N}(z)$ of degree $n$ is defined by \eqref{eq ortho4}.
Let 
\begin{equation}\label{def psi}
    \psi_n(z):=(z-a)^{Nc}p_n(z),\quad n=0,1,\dots,
\end{equation}
and write
\begin{equation}\label{def prekernel}
{\cal K}_{n}(z,{\zeta}):=\ee^{-Nz\overline\zeta}\sum_{k=0}^{n-1}\frac{1}{h_k}\psi_k(z)\overline{\psi_k(\zeta)},    
\end{equation}
where ${\cal K}_{n}(z,{\zeta})$, sometimes called the {\it pre-kernel}, is purely analytic in $z$ and anti-analytic in $\zeta$.

Using the definition of the correlation kernel in \eqref{def kernel} and the explicit expression of $Q$ in \eqref{def Q}, the correlation kernel can be written as
\begin{equation}\label{relation of kernels}
{\bf K}_{n}(z,{\zeta})={\bf K}_{n,N}(z,{\zeta})
=\frac{\ee^{Nz\overline\zeta}}{\ee^{\frac{N}{2}|z|^2+\frac{N}{2}|\zeta|^2}}\frac{|z-a|^{Nc}|\zeta-a|^{Nc}}{(z-a)^{Nc}(\overline\zeta-a)^{Nc}}{\cal K}_{n}(z,{\zeta}).
    \end{equation}

    \begin{thm}\label{thm 21}
Suppose that $a \not=0$. Then we have the following form of the Christoffel-Darboux identity:
\begin{align} \label{CDI_v3}
	\begin{split}
 \overline{\partial}_\zeta {\cal K}_{n}(z,\zeta)
		&=e^{ -N z \bar{\zeta} }  \frac{1}{ \frac{n+Nc}{N}h_{n-1}-h_{n}  }	\overline{\partial}_\zeta \overline{ \psi_{n}(\zeta) }    \Big( \psi_n(z)-z \psi_{n-1}(z) \Big)
		\\
		&\qquad -e^{ -N z \bar{\zeta} } \frac{p_{n+1}(a)}{p_n(a)} \frac{N\,h_n/h_{n-1} }{ \frac{n+Nc+1}{N} h_n-h_{n+1}   } \overline{ \psi_{n-1}(\zeta) }  \Big(  \psi_{n+1}(z)-z \psi_n(z)  \Big).
	\end{split}
\end{align}
\end{thm}
We note that $h_n$ is the norming constant defined by \eqref{eq ortho4}.
Theorem \ref{thm 21} has been proved in \cite{byun 2021} and also used in \cite{byun 2022}.  For the convenience of the readers, we put the proof in Appendix~\ref{appendix b}. For the radially symmetric case when $a=0$, we have
$p_j(z)=z^j$ and $h_j= \frac{\Gamma(j+Nc+1)}{ N^{j+Nc+1} }$. As the correlation kernel can be computed directly, we exclude such case.

From Theorem \ref{thm 21}, to obtain the asymptotic behavior of the correlation kernel, we need the asymptotic behavior of $p_n$ in the scaling limit
\begin{equation}
    n,N\to\infty, \quad \frac{n}{N}=t=t_c+{\cal O}\Big(\frac{1}{N^{2/3}}\Big). 
\end{equation}
In particular, we will need the asymptotic behavior of $p_n$ up to the first sub-leading term in large $N$ expansion.

The critical parameter value $t_c$ corresponds to the droplet with the merging singularity, as shown in Figure~\ref{Figure 1}.  We will need the evolution of the droplet before ($t <t_c$) and after ($t >t_c$) the merging singularity or, more accurately, the (complex) logarithmic potential $g(z)$ generated by the limiting measure of the Coulomb particles. 

Let us introduce some notations and results from \cite{Ba 2015}.  We define the function $V(z)$ by 
$$V(z):=az-c\log(z-a)+(c+t)\log z,$$
and the function $\phi(z)$ by 
\begin{equation}\label{def phi}
 \phi(z):=az-(t+c)\log(z)+c\log(z - a)+t \ell,  
\end{equation} where the constant $\ell$ will be defined below.
Let $b$ and $\beta$ be two critical points of $\phi(z)$ such that
\begin{equation}
b:=\frac{a^2+t+\sqrt{(t-a^2)^2-4a^2c}}{2a},\quad \beta:=\frac{a^2+t-\sqrt{(t-a^2)^2-4a^2c}}{2a},    \end{equation} where we take the principal branch of the roots. We set the constant $\ell$ by
\begin{equation}\label{def ell}
  \ell:=\frac{1}{t}\big((t+c)\log\beta-c\log(\beta-a)-a\beta \big) 
\end{equation}
such that  $\phi(\beta)=0$.  Note that $\ell$ is real for $t\geq t_c$ and $\ell$ is complex for $t<t_c$.
We also have
\begin{equation}\label{def b and beta}
b=b_c+\frac{c^{1/4}}{\sqrt a}\sqrt{t-t_c}+{\cal O}(t-t_c),\quad \beta=b_c-\frac{c^{1/4}}{\sqrt a}\sqrt{t-t_c}+{\cal O}(t-t_c).    \end{equation}
These $b$ and $\beta$ are complex for $t<t_c$ and we choose them such that ${\rm Im}(\beta)>0>{\rm Im}(b)$ in this case. See Figure \ref{figue bbeta} for the illustration of $b$ and $\beta$. 

We define the curve ${\cal B}$ such that ${\cal B}$ satisfies the following conditions: ${\cal B}$ is a simple closed curve enclosing $0$ and $a$ with $\beta\in{\cal B}$ and , on the curve ${\cal B}$, the following equality is satisfied, 
\begin{equation}
    {\rm Re}\Big(\log z + \frac{c}{t}\log\Big(\frac{z}{z-a}\Big)-\frac{a z}{t}\Big) ={\rm Re}(\ell).
\end{equation}
More details about the curve is shown in \cite[Section~2.3]{Ba 2015}. See Figure \ref{figue scurve} for the plot of ${\cal B}$. 

We define the $g$-function by
\begin{equation}\label{def gfunction}
g(z):=\begin{cases}\displaystyle
\frac{1}{2t} \big(V(z)-\phi(z)+t\ell\big)=\log z+\frac{c}{t}\log \left(\frac{z}{z-a}\right),\quad &z\in{\rm Ext}({\cal B}),\\ \displaystyle
\frac{1}{2t} \big(V(z)+\phi(z)+t\ell\big)=
 \frac{az}{t}+\ell,\quad &z\in{\rm Int}({\cal B}),
\end{cases}    
\end{equation}  where ${\rm Ext}({\cal B})$ and ${\rm Int}({\cal B})$ stands for the exterior (including $\infty$) and the interior of the Jordan curve ${\cal B}$.

\begin{figure}
\centering
  \includegraphics[width=0.48\textwidth]{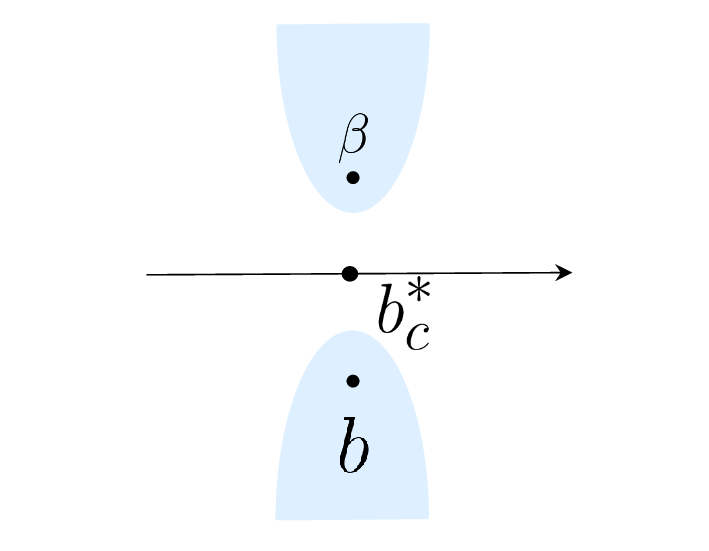} \includegraphics[width=0.48\textwidth]{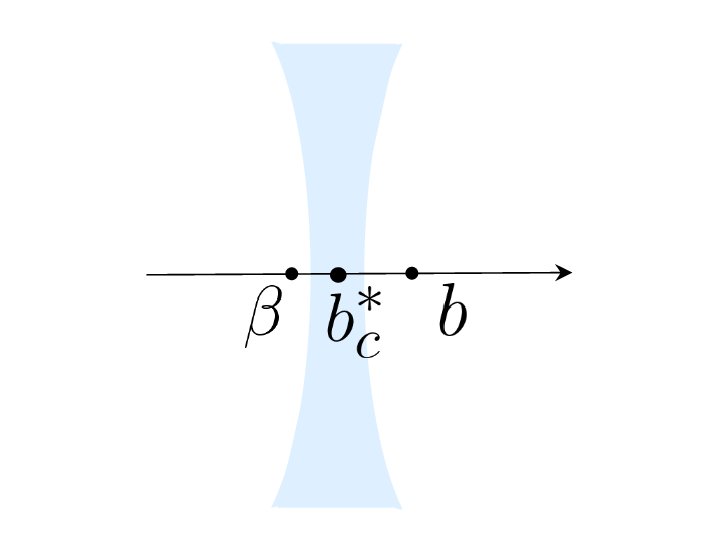}
\caption{The positions of $\beta, b$ and $b_c^*$ in $z$-coordinate for $t<t_c$(left) and $t>t_c$(right). The shaded parts represent the regions of the droplet.}\label{figue bbeta}
\end{figure}

\begin{figure}
\begin{center}
\includegraphics[width=0.52\textwidth]{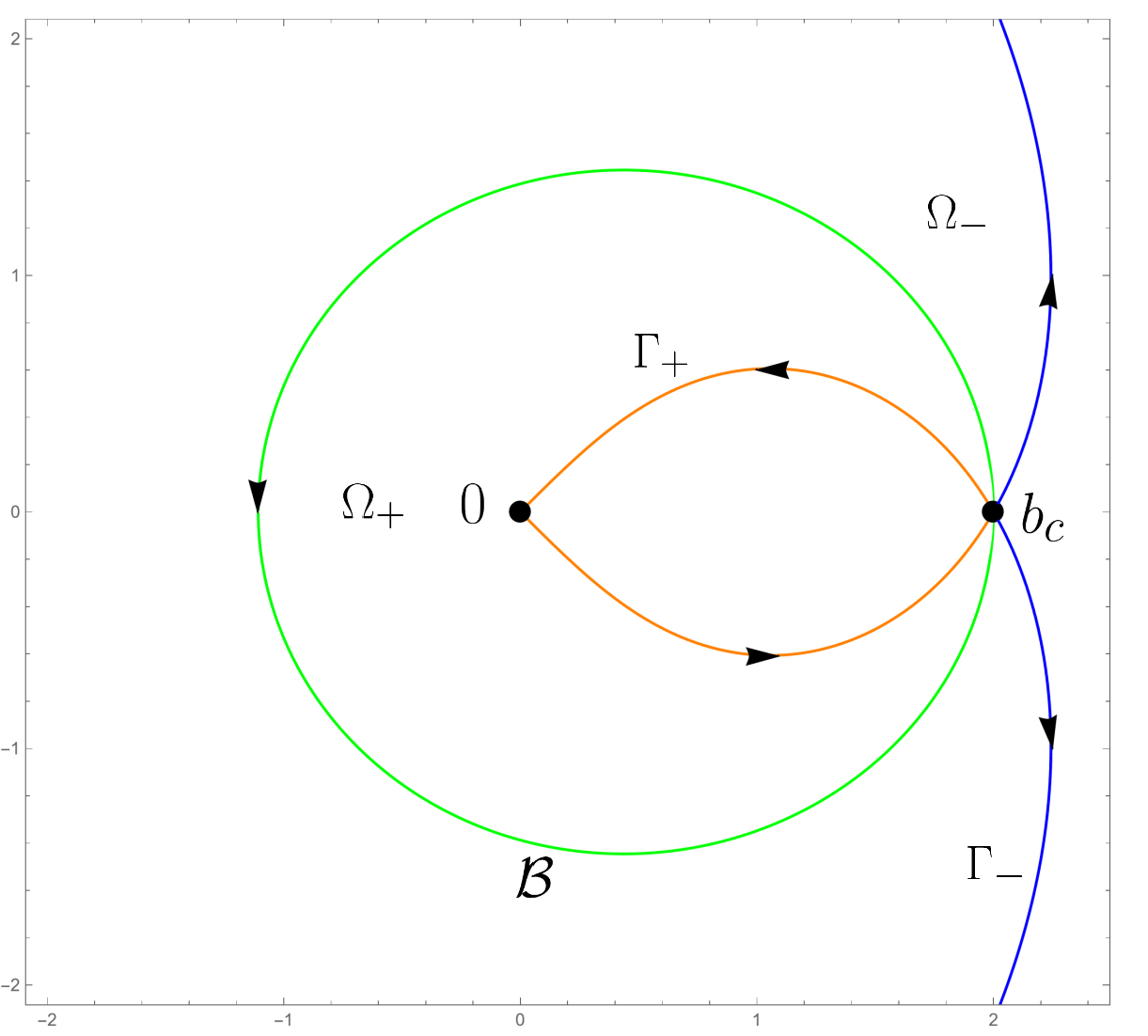}
\end{center}
\caption{When $t=t_c$, the figure describes the green curve ${\cal B}$, the orange line $\Gamma_+$, the blue line $\Gamma_-$, and the black dots $b_c$ and the origin. All the curves are assigned a direction as shown. All the figures are numerically plotted for $a=c=1$. }\label{figue scurve}
\end{figure}

\begin{remark} For readers with some knowledge of \cite{Ba 2015}, we note that our definition of $g(z)$ matches the definition for $t\geq t_c$, the post-critical regime in the reference, while for $t<t_c$, our definition does not match the definition of the pre-critical regime in the reference.  Since we are only interested in the vicinity of the criticality, i.e. $t-t_c ={\cal O}(N^{-2/3})$, such deviation does not affect our analysis. \end{remark}

\subsection{Fine asymptotics of $p_n$}\label{section 2.2}

Let $\Gamma$ be a simple closed curve enclosing $0$ and $a$, we define the matrix function 
\begin{equation}\label{def Yz}
Y(z):=Y_n(z)=\begin{pmatrix}\displaystyle
p_{n}(z)&\displaystyle\frac{1}{2\pi\ii}\int_\Gamma\frac{p_{n}(w)\omega_{n,N}(w)}{w-z}\dd w\\ q_{n}(z)&\displaystyle\frac{1}{2\pi\ii}\int_\Gamma\frac{q_{n}(w)\omega_{n,N}(w)}{w-z}\dd w
\end{pmatrix},    
\end{equation}
where \begin{equation}\label{def omega}
\omega_{n,N}(z):=\ee^{-NV(z)}=\frac{(z-a)^{Nc}\ee^{-Naz}}{z^{Nc+n}},    
\end{equation} and $q_{n}(z):=q_{n,N}(z)$ is the unique polynomial of degree $n-1$ such that
$$\frac{1}{2\pi\ii}\int_\Gamma\frac{q_{n}(w)\omega_{n,N}(w)}{w-z}\dd w=\frac{1}{z^n}\left(1+{\cal O}\Big(\frac{1}{z}\Big)\right).$$
Then $Y$ satisfies the following Riemann-Hilbert problem:
\begin{equation}\label{rhp for y}
\begin{cases}
 Y_+(z)=  Y_-(z)\begin{pmatrix}
1&\omega_{n,N}(z)\\0&1
\end{pmatrix},&\quad z\in \Gamma,\\
Y(z)=\left(I+{\cal O}(z^{-1})\right)z^{n\sigma_3},&\quad z\to\infty,\\
Y(z)\mbox{\quad is holomorphic},& \mbox{otherwise}.\end{cases}
 \end{equation}

From now on, let $\Gamma$ exactly match ${\cal B}$. We choose $\Gamma_+$ to be the steepest descent path from $\beta$ inside ${\rm Int}({\cal B})$ such that ${\rm Re}(\phi(z))<0$ on $\Gamma_+$, and $\Gamma_-$ to be the steepest descent path from $\beta$ inside ${\rm Ext}({\cal B})$ such that ${\rm Re}(\phi(z))<0$ on $\Gamma_-$. The domains $\Omega_\pm$ are defined by the open sets enclosed by ${\cal B}$ and $\Gamma_\pm$ respectively. See Figure \ref{figue scurve}.

Let us define the matrix $A(z)$ by 
\begin{align}\label{def A}
    A(z):=\begin{cases}
    \ee^{-\frac{tN\ell}{2}\sigma_3}Y(z)\ee^{-tN(g(z)-\frac{\ell}{2})\sigma_3},& z\in\CC\setminus\Omega_+\cup\Omega_-,\\
     \ee^{-\frac{tN\ell}{2}\sigma_3}Y(z)\begin{pmatrix}
1&0\\-1/\omega_{n,N}(z)&1
\end{pmatrix}\ee^{-tN(g(z)-\frac{\ell}{2})\sigma_3},&z\in\Omega_+,\\
      \ee^{-\frac{tN\ell}{2}\sigma_3}Y(z)\begin{pmatrix}
1&0\\1/\omega_{n,N}(z)&1
\end{pmatrix}\ee^{-tN(g(z)-\frac{\ell}{2})\sigma_3},& z\in\Omega_-.
    \end{cases}
\end{align}
Then $A$ satisfies the  Riemann-Hilbert problem
\begin{equation}\label{jump for A}
\begin{cases}
 A_+(z)=  A_-(z)\begin{pmatrix}
1&0\\\ee^{N\phi(z)}&1
\end{pmatrix},&\quad z\in \Gamma_\pm,\\
 A_+(z)=  A_-(z)\begin{pmatrix}
0&1\\-1&0
\end{pmatrix},&\quad z\in {\cal B},\\
A(z)=I+{\cal O}(z^{-1}),&\quad z\to\infty,\\
A(z)\mbox{\quad is holomorphic},& \mbox{otherwise}.\end{cases}
 \end{equation}
Since ${\rm Re}(\phi(z))<0$ along $\Gamma_\pm$, the jump of $A$ converges, as $N\to\infty$, to the jump of $\Phi$ that we define by the  Riemann-Hilbert problem
 \begin{equation}
\begin{cases}
 \Phi_+(z)= \Phi_-(z)\begin{pmatrix}
0&1\\-1&0
\end{pmatrix},&\quad z\in {\cal B},\\
\Phi(z)=I+{\cal O}(z^{-1}),&\quad z\to\infty,\\
\Phi(z)\mbox{\quad is holomorphic},& \mbox{otherwise}.\end{cases}
 \end{equation} 
 A solution to the above Riemann-Hilbert problem is given by
 \begin{align}\label{eq Phi}
    \Phi(z):=\begin{cases}
    I,& z\in {\rm Ext}({\cal B}),\\
   \begin{pmatrix}
0&1\\-1&0
\end{pmatrix} ,&z\in{\rm Int}({\cal B}).
    \end{cases}
\end{align}

 When $z$ is close to the points $b$ and $\beta$, the convergence of the jump matrix of $A$ to the jump matrix of $\Phi$ is getting worse. Therefore we need the local parametrices  around $b$ and $\beta$ that satisfies the exact jump conditions of $A$ in \eqref{jump for A}.

There exists a fixed disk, $D_c$, centered at $b_c$ such that there exists the univalent map $\xi: D_c\to \CC$ satisfying
\begin{equation}\label{def xixi}
\xi(\beta)=-\xi(b)\quad \mbox{and}\quad -\frac{8\ii}{3}\big(\xi(z)-\xi(\beta)\big)^2\big(\xi(z)+2\xi(\beta)\big)=
 N\phi(z),\quad z\in D_c.
   \end{equation} 
Here, the branch of the solution to the cubic equation is taken such that $\xi(\beta) >0 $ for $t<t_c$ and  $-\ii\xi(\beta) >0 $ for $t>t_c$, see Figure \ref{zeta betab}. The existence of such $\xi$ is shown in \cite[Section 6]{Ba 2015}. Note that $\xi$ maps $\Gamma_\pm$ into $\gamma_\pm$ and ${\cal B}$ into $\RR$. See Figure \ref{Figure 2} for the contours $\gamma_\pm$. 

\begin{figure}
\centering
  \includegraphics[width=0.40\textwidth]{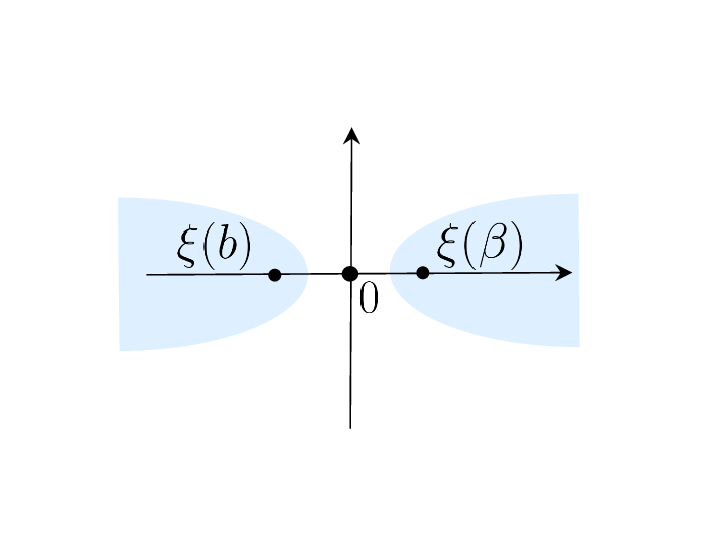} \includegraphics[width=0.40\textwidth]{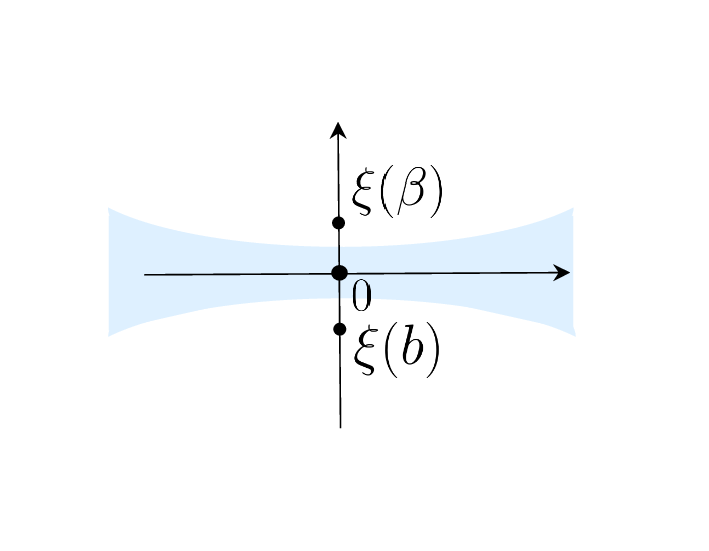}
\caption{The positions of $\xi(\beta)$, $\xi(b)$ and $\xi(b_c^*)=0$ for $t<t_c$ (left) and $t>t_c$ (right). The shaded parts represent the image of the droplet under $\xi$.}\label{zeta betab}  
\end{figure}

Let us define $\hat s$ such that \eqref{def xixi} can be written as
\begin{equation}\label{def xi}
-2\ii\left(\frac{4}{3}\xi(z)^3+\hat s \xi(z)+\frac{8}{3}\xi(\beta)^3\right)=
 N\phi(z),\quad z\in D_c.
\end{equation}
The parameter $\hat s$ is given by 
\begin{equation}\label{def zetabeta}
\hat s:=-4\xi(\beta)^2=s(1+{\cal O}(t-t_c)),\quad \text{where}\quad s:=\frac{\gamma_cN^{2/3}}{2b_c}(t-t_c).
\end{equation} 
The parameter $\gamma_c$ is defined in \eqref{def gammacl}. Note that $\hat s = s(1+{\cal O}(N^{-2/3}))$ for bounded $s$.

Let $b_c^*$ be the preimage of the origin under $\xi$,
\begin{equation}\label{def bcstar}
b_c^*:=\xi^{-1}(0).    
\end{equation}
By the definition of $b_c^*$ above and \eqref{def xixi}  we have 
\begin{equation}\label{def b star0}
 -\ii\frac{16}{3}\xi(\beta)^3=N\phi(b_c^*),\quad    \phi(b)=2\phi(b_c^*).
\end{equation}
 It follows that $$a^2+t-2ab_c^*-(t+c)\big(\log(t+c)-2\log b_c^*\big)+c\big(\log c-2\log(b_c^*-a)\big)=0.$$ Hence, we have
\begin{equation}\label{def b star}
b_c^*=b_c+\frac{s}{2\gamma_cN^{2/3}}+{\cal O}\Big(\frac{1}{N^{4/3}}\Big).           \end{equation} See Figure \ref{figue bbeta} for the illustration of $b_c^*$.

Using \eqref{def zetabeta}, we have
\begin{equation}\label{zeta beta}
   \ii\frac{16}{3}\xi(\beta)^3=\begin{cases}\displaystyle
   \frac{2\hat s^{3/2}}{3},&\hat s>0,\\\displaystyle
   \frac{2(-\hat s)^{3/2}}{3}\ii,& \hat s\leq 0.
   \end{cases}
\end{equation}

Let us define 
$$r_1:=\frac{\xi'(b_c^*)}{\ii N^{1/3}},\quad r_2:=\frac{\xi''(b_c^*)}{\ii N^{1/3}},\quad r_3:=\frac{\xi'''(b_c^*)}{\ii N^{1/3}},$$
so that we have
\begin{equation}\label{zeta expanssion}
\frac{\xi(z)}{\ii N^{1/3}}=r_1(z-b_c^*)+\frac{r_2}{2}(z-b_c^*)^2 +\frac{r_3}{6}(z-b_c^*)^3+{\cal O}(z-b_c^*)^4,\quad \mbox{as $z\to b_c^*$.}
\end{equation}

By the definition of $\xi$ in \eqref{def xixi}, the facts that
\begin{equation}
\begin{aligned}
\xi'(b_c^*)&=\frac{N}{-2\ii \hat s}\phi(b_c^*)',\quad \xi''(b_c^*)=\frac{N}{-2\ii \hat s}\phi(b_c^*)'',
\quad \xi'''(b_c^*)=\frac{N}{-2\ii \hat s}\phi(b_c^*)'''-\frac{8}{\hat s}\big(\xi'(b_c^*)\big)^3,
\end{aligned}\end{equation}  and \eqref{def phi}, we have 
\begin{equation}\label{eq-gamma123}
r_1=-\frac{1}{\gamma_c}+{\cal O}\Big(\frac{1}{N^{2/3}}\Big),\quad r_2=\frac{t_c} {2ab_c\sqrt c \gamma_c}+{\cal O}\Big(\frac{1}{N^{2/3}}\Big), \quad
r_3=-\frac{ (9a^2+12a\sqrt c+16c) }{8b_c^2c\gamma_c}+{\cal O}\Big(\frac{1}{N^{2/3}}\Big).
    \end{equation}

Let us define the following Riemann-Hilbert problem for the matrix function $\Pi(\xi;\hat s)$,
\begin{equation}
\begin{cases}
\Pi_+(\xi;\hat s)= \Pi_-(\xi;\hat s),&\quad \xi\in \RR,\\
\Pi_+(\xi;\hat s)= \Pi_-(\xi;\hat s)\begin{pmatrix}
1&0\\\ee^{2\ii\left(\frac{4}{3}\xi^3+\hat s\xi\right)}&1
\end{pmatrix},&\quad \xi\in \gamma_+,\\
 \Pi_+(\xi;\hat s)= \Pi_-(\xi;\hat s)\begin{pmatrix}
1&-\ee^{-2\ii\left(\frac{4}{3}\xi^3+\hat s\xi\right)}\\0&1
\end{pmatrix},&\quad \xi\in \gamma_-,\\
\Pi(\xi;\hat s)=I+{\cal O}(\xi^{-1}) ,&\quad  \xi\to \infty.
\end{cases}
 \end{equation} 
 
The solution of the above Riemann-Hilbert problem can be written as 
\begin{equation}\label{pi to psi}
 \Pi(\xi;\hat s)=\widetilde\Psi(\xi;\hat s)\ee^{\ii\left(\frac{4}{3}\xi^3+\hat s\xi\right)\sigma_3},   
\end{equation} where $\widetilde\Psi(\xi;\hat s)$ is the solution of the Riemann-Hilbert problem in \eqref{rhp phi}.
 
 Let us define  ${\cal P}(z)$ by
\begin{equation}\label{def p}
 {\cal P}(z)=   \begin{cases}
    \ee^{-\ii\frac{8}{3}\xi(\beta)^3\sigma_3}H(z)\Pi(\xi(z);\hat s)\ee^{\ii\frac{8}{3}\xi(\beta)^3\sigma_3},&\quad z\in D_c\cap{\rm Int}(\cal B), \\
    \begin{pmatrix}
0&1\\-1&0
\end{pmatrix}\ee^{-\ii\frac{8}{3}\xi(\beta)^3\sigma_3}H(z)\Pi(\xi(z);\hat s)\ee^{\ii\frac{8}{3}\xi(\beta)^3\sigma_3}\begin{pmatrix}
0&-1\\1&0
\end{pmatrix},&\quad z\in D_c\cap{\rm Ext}(\cal B).
    \end{cases}
\end{equation} 
Here $H(z)$ is a holomorphic matrix function that will be determined in Proposition \ref{proposition sh}. 
 One can check that $\Phi(z){\cal P}(z)$ satisfies the jump conditions \eqref{jump for A} of $A$ in $D_c$.

In order to obtain the asymptotics of $p_n$ up to the first sub-leading term, we need the first two sub-leading terms of $\Pi(\xi;\hat s)$.
Let $\Pi(\xi;\hat s)$ be given by
\begin{equation}\label{def pi1pi2}
\Pi(\xi;\hat s)=I+\frac{\Pi_1(\hat s)}{2\ii\xi}+\frac{\Pi_2(\hat s)}{\xi^2}+{\cal O}\left(\frac{1}{\xi^3}\right),\end{equation} where
\begin{equation*}
    \Pi_1(\hat s)= \begin{pmatrix}
r(\hat s)&q(\hat s)\\-q(\hat s)&-r(\hat s)
\end{pmatrix},\quad \Pi_2(\hat s)= \begin{pmatrix}
p_{11}(\hat s)&p_{12}(\hat s)\\p_{21}(\hat s)&p_{22}(\hat s)
\end{pmatrix}.
\end{equation*} The series above is in integer powers of $\xi$ because the monodromy of $\Pi(\xi;\hat s)$ converges to the identity exponentially fast as $z\to\infty$.
Here we have
\begin{equation}\label{relation up}
\begin{aligned}
r(\hat s)&=q'(\hat s)^2-\hat sq(\hat s)^2-q(\hat s)^4,\\ p_{12}(\hat s)&=p_{21}(\hat s)=\frac{q(\hat s)r(\hat s)+q'(\hat s)}{4}, \\
 p_{11}(\hat s)&=p_{22}(\hat s)=\frac{q(\hat s)^2-r(\hat s)^2}{8}.
\end{aligned}    
\end{equation}

The relations in \eqref{relation up} are derived in Appendix~\ref{appendix a}.

The following proposition derives a rational matrix function $S(z)$ with the only pole at $b_c^*$ and a holomorphic matrix function $H(z)$ such that the modified global parametrix $\Phi(z)S(z)$ matches with the local parametrix $\Phi(z){\cal P}(z)$ along $\partial D_c$. This procedure of improving the local parametrix is known as the ``partial Schlesinger transform'' \cite{Bertola 2008}. The construction of $H_1(z)$ and $S(z)$ with a simple pole and the constant matrix $S_{11}(\hat s)$ was also described in \cite{Ba 2015}. Here we construct $H_1(z)$, $H_2(z)$ and $S(z)$ with a simple pole and double poles such that the modified global parametrix better matches with the local parametrix along $\partial D_c$.

\begin{prop}\label{proposition sh}
Let $\Pi_1(\hat s)$ and $\Pi_2(\hat s)$ be given in \eqref{def pi1pi2}.   We assume  $t-t_c = {\cal O}(N^{-2/3})$ and, therefore, $\hat s={\cal O}(1)$ as $N$ grows to infinity.  Let $S(z)$ be a rational matrix function  with the only pole at $b_c^*$ given by
\begin{equation}\label{def SS}
S(z):=\ee^{-\ii\frac{8}{3}\xi(\beta)^3\sigma_3}\Big(I+\frac{S_{11}(\hat s)+S_{21}(\hat s)}{z-b_c^*}+\frac{S_{22}(\hat s)}{(z-b_c^*)^2}\Big)\ee^{\ii\frac{8}{3}\xi(\beta)^3\sigma_3},   
\end{equation} where \begin{align}\label{def s11}
 S_{11}(\hat s)&=-\frac{\Pi_1(\hat s)}{2r_1N^{1/3}},  \\ \label{def s21}
     S_{21}(\hat s)&=\frac{r_2}{r_1^3N^{2/3}}\begin{pmatrix}
0&p_{12}(\hat s)\\p_{21}(\hat s)&0
\end{pmatrix},\\\label{def s22}
S_{22}(\hat s)&=-\frac{\Pi_2(\hat s)}{r_1^2N^{2/3}}.
   \end{align}
Let $H(z)$ be the matrix function given by \begin{equation}
H(z):=H_2(z)H_1(z),
\end{equation} where
\begin{align}\label{def h1}
H_1(z)&:=I-\frac{\Pi_1(\hat s)}{2\ii\xi(z)}+\frac{S_{11}(\hat s)}{z-b_c^*},\\
\label{def h2}
H_2(z)&:=I+\frac{ S_{21}(\hat s)}{z-b_c^*}+\frac{ S_{22}(\hat s)}{(z-b_c^*)^2}-
\bigg(\frac{\Pi_2(\hat s)}{\xi(z)^2}+\Big(H_1(z)-I\Big)\frac{\Pi_1(\hat s)}{2\ii\xi(z)}\bigg),
\end{align}
such that $H(z)$ is holomorphic at $b_c^*$.
Then there exists a fixed disk $D_c$ centered at $b_c$ such that
\begin{align}\label{eq hcpi3}
    \ee^{-\ii\frac{8}{3}\xi(\beta)^3\sigma_3}H(z)\Pi(\xi(z);\hat s)\ee^{\ii\frac{8}{3}\xi(\beta)^3\sigma_3}&=S(z)\bigg(I+{\cal O}\Big(\frac{1}{N}\Big)\bigg),\quad \mbox{as $N\to\infty$},
\end{align} uniformly on $\partial D_c$. Let us further denote $H(z)$  by \begin{equation}\label{def HH}
H(z)=I+\begin{pmatrix}
h_{11}(z)&h_{12}(z)\\
h_{21}(z)&h_{22}(z)
\end{pmatrix}.
\end{equation}
We have, uniformly over $z\in D_c$,
\begin{equation}\label{eq h11}
\begin{aligned}
h_{11}(z)&=-\frac{r_2r(\hat s)}{4r_1^2N^{1/3}}+\frac{(3r_2^2-2r_1r_3)r(\hat s)(z-b_c^*)}{24r_1^3N^{1/3}}\\
&\quad+\frac{2(9r_2^2-4r_1r_3)p_{11}(\hat s)-(3r_2^2-r_1r_3)(q(\hat s)^2-r(\hat s)^2)}{24r_1^4N^{2/3}}+{\cal O}\Big(\frac{1}{N},\frac{(z-b_c^*)^2}{N^{1/3}}\Big),\\
h_{12}(z)&=-\frac{r_2q(\hat s)}{4r_1^2N^{1/3}}+\frac{(3r_2^2-2r_1r_3)q(\hat s)(z-b_c^*)}{24r_1^3N^{1/3}}+\frac{(9r_2^2-4r_1r_3)p_{12}(\hat s)}{12r_1^4N^{2/3}}+{\cal O}\Big(\frac{1}{N},\frac{(z-b_c^*)^2}{N^{1/3}}\Big),\\
h_{21}(z)&=\frac{r_2q(\hat s)}{4r_1^2N^{1/3}}-\frac{(3r_2^2-2r_1r_3)q(\hat s)(z-b_c^*)}{24r_1^3N^{1/3}}+\frac{(9r_2^2-4r_1r_3)p_{21}(\hat s)}{12r_1^4N^{2/3}}+{\cal O}\Big(\frac{1}{N},\frac{(z-b_c^*)^2}{N^{1/3}}\Big),\\
h_{22}(z)&=\frac{r_2r(\hat s)}{4r_1^2N^{1/3}}-\frac{(3r_2^2-2r_1r_3)r(\hat s)(z-b_c^*)}{24r_1^3N^{1/3}}\\
&\quad+\frac{2(9r_2^2-4r_1r_3)p_{22}(\hat s)-(3r_2^2-r_1r_3)(q(\hat s)^2-r(\hat s)^2)}{24r_1^4N^{2/3}} +{\cal O}\Big(\frac{1}{N},\frac{(z-b_c^*)^2}{N^{1/3}}\Big),\end{aligned}    
\end{equation}
where the big $\cal O$ notation with multiple arguments is defined
by  ${\cal O}(A, B) ={\cal O}({\rm max}(|A|,|B|))$.

\end{prop}
\begin{proof}
In the proof below, we will use the following bounds on $\partial D_c$ whenever necessary:
\begin{equation}
    S_{11}(\hat s)={\cal O}(N^{-1/3}),\quad S_{21}(\hat s)={\cal O}(N^{-2/3}),\quad S_{22}(\hat s)={\cal O}(N^{-2/3}),\quad  \Pi_j(\hat s) = {\cal O}(1),\quad  \xi ={\cal O}(N^{1/3}).
\end{equation}
All the bounds in the proof should be understood in the limit of large $N$.
When $z \in \partial D_c$ and therefore $1/\xi = O(N^{-1/3})$.  We sometimes use the notation ${\cal O}(1/\xi^3)$ to indicate the origin of the bound.

Firstly, by the definition of $H_1(z)$ in \eqref{def h1} with $S_{11}(\hat s)$ in \eqref{def s11}, one can see that $H_1(z)$ is holomorphic at $b_c^*$. Moreover, using  $\Pi(\xi;\hat s)$ with $\xi=\xi(z)$ in \eqref{def pi1pi2} and  $H_1(z)$ in \eqref{def h1}, we have

\begin{equation}\label{eq h1s1}
\begin{aligned}
H_1(z)\Pi(\xi;\hat s)&=I+\frac{S_{11}(\hat s)}{z-b_c^*}+\bigg(\frac{\Pi_2(\hat s)}{\xi^2}+\big(H_1(z)-I\big)\frac{\Pi_1(\hat s)}{2\ii\xi}\bigg)+{\cal O}\left(\frac{S_{11}(\hat s)}{(z-b_c^*)\xi^2},\frac{1}{\xi^3}\right).
\end{aligned}
\end{equation}
Since
\begin{equation}\label{eq aprrixpi2}
\frac{\Pi_2(\hat s)}{\xi^2}+\big(H_1(z)-I\big)\frac{\Pi_1(\hat s)}{2\ii\xi}\sim {\cal O}\Big(\frac{1}{N^{2/3}}\Big)   \end{equation} as $N\to \infty$ and $z\in\partial D_c$, we have
\begin{equation}\label{require holo1}
H_1(z)\Pi(\xi;\hat s)=I+\frac{S_{11}(\hat s)}{z-b_c^*}+{\cal O}\left(\frac{1}{N^{2/3}}\right)    
\end{equation} as $N\to \infty$ and $z\in\partial D_c$.

Secondly,  by the definition of $H_2(z)$ in \eqref{def h2} with  $\Pi_1(\hat s)$, $\Pi_2(\hat s)$ in \eqref{def pi1pi2}, $S_{21}(\hat s)$ in \eqref{def s21}, and $S_{22}(\hat s)$ in \eqref{def s22}, one can see that $H_2(z)$ is holomorphic at $b_c^*$. It follows that $H(z)$ is holomorphic at $b_c^*$.   Moreover, using $\Pi(\xi(z);\hat s)$ in \eqref{def pi1pi2} , $H_1(z)$ in \eqref{def h1}, and $H_2(z)$ in \eqref{def h2}, we have
\begin{equation}\label{eq hc21}
\begin{aligned}
H_2(z)H_1(z)\Pi(\xi;\hat s)&=
H_2(z)\bigg(I+\frac{S_{11}(\hat s)}{z-b_c^*}+\frac{\Pi_2(\hat s)}{\xi^2}+\big(H_1(z)-I\big)\frac{\Pi_1(\hat s)}{2\ii\xi}+{\cal O}\big(\frac{1}{\xi^3}\big)\bigg)\\
&=I+\frac{S_{11}(\hat s)+S_{21}(\hat s)}{z-b_c^*}+\frac{S_{22}(\hat s)}{(z-b_c^*)^2}+\bigg(\frac{ S_{21}(\hat s)}{z-b_c^*}+\frac{ S_{22}(\hat s)}{(z-b_c^*)^2}\bigg)\frac{S_{11}(\hat s)}{z-b_c^*}\\
&\qquad -\bigg(\frac{\Pi_2(\hat s)}{\xi^2}+\Big(H_1(z)-I\Big)\frac{\Pi_1(\hat s)}{2\ii\xi}\bigg)\frac{S_{11}(\hat s)}{z-b_c^*}+{\cal O}\left(\frac{1}{\xi^3}\right),
    \end{aligned}
\end{equation} where the 1st equation is obtained by \eqref{eq h1s1}, the 2nd equation is obtained by the definition of $H_2(z)$ in \eqref{def h2}. As $N\to \infty$ and $z\in\partial D_c$, using \eqref{eq aprrixpi2}, \eqref{eq hc21} and the definitions of $S_{11}(\hat s)$, $S_{21}(\hat s)$ and $S_{22}(\hat s)$ in \eqref{def s11}, \eqref{def s21}, \eqref{def s22}, respectively, we have
\begin{equation*}
\bigg(\frac{ S_{21}(\hat s)}{z-b_c^*}+\frac{ S_{22}(\hat s)}{(z-b_c^*)^2}\bigg)\frac{S_{11}(\hat s)}{z-b_c^*}-\bigg(\frac{\Pi_2(\hat s)}{\xi^2}+\Big(H_1(z)-I\Big)\frac{\Pi_1(\hat s)}{2\ii\xi}\bigg)\frac{S_{11}(\hat s)}{z-b_c^*}+{\cal O}\left(\frac{1}{\xi^3}\right)={\cal O}\Big(\frac{1}{N}\Big).
\end{equation*}
It follows that
\begin{equation*}
\begin{aligned}
H_2(z)H_1(z)\Pi(\xi;\hat s)
&=I+\frac{S_{11}(\hat s)+S_{21}(\hat s)}{z-b_c^*}+\frac{S_{22}(\hat s)}{(z-b_c^*)^2}+{\cal O}\Big(\frac{1}{N}\Big)\\
&=\bigg(I+\frac{S_{11}(\hat s)+S_{21}(\hat s)}{z-b_c^*}+\frac{S_{22}(\hat s)}{(z-b_c^*)^2}\bigg)\Big(I+{\cal O}\Big(\frac{1}{N}\Big)\Big).
    \end{aligned}\end{equation*} Therefore, \eqref{eq hcpi3} holds. Moreover,
when $z\in D_c$, using the definitions of $H_1(z)$ in \eqref{def h1} and $H_2(z)$ in \eqref{def h2}, \eqref{eq h11} holds.
\end{proof}

Let us define $A^\infty(z)$ by
\begin{align}\label{eq Ainfty}
    A^\infty(z):=\begin{cases}
    \Phi(z)S(z),& z\in {\rm Int}({\cal B})\setminus D_c,\\
    \Phi(z)\begin{pmatrix}
0&1\\-1&0
\end{pmatrix}S(z)\begin{pmatrix}
0&-1\\1&0
\end{pmatrix} ,&z\in{\rm Ext}({\cal B})\setminus D_c,\\
      \Phi(z){\cal P}(z),& z\in D_c.
    \end{cases}
\end{align}
This will be the strong asymptotics of $A$ and we define the error matrix by $${\cal E}(z)=A^\infty(z)A^{-1}(z).$$ 

When $z\in {\rm Int}({\cal B}) \cap \partial D_c$, we have
\begin{equation}
    \begin{aligned}
    {\cal E}_+(z)({\cal E}_-(z))^{-1}&=A_+^\infty(z)(A_-^\infty(z))^{-1}=\Phi(z){\cal P}(z)S(z)^{-1}\Phi(z)^{-1}\\
    &=\begin{pmatrix}
0&1\\-1&0
\end{pmatrix}\ee^{-\ii\frac{8}{3}\xi(\beta)^3\sigma_3}H(z)\Pi(\xi(z);\hat s)\ee^{\ii\frac{8}{3}\xi(\beta)^3\sigma_3}S(z)^{-1}\begin{pmatrix}
0&-1\\1&0
\end{pmatrix}\\
&=\begin{pmatrix}
0&1\\-1&0
\end{pmatrix}S(z)\bigg(I+{\cal O}\Big(\frac{1}{N}\Big)\bigg)S(z)^{-1}\begin{pmatrix}
0&-1\\1&0
\end{pmatrix}\\
&=I+{\cal O}\Big(\frac{1}{N}\Big).
    \end{aligned}
\end{equation} Here the 2nd equality is obtained by the definition of $A^\infty(z)$ in \eqref{eq Ainfty}, the 3rd equality is obtained by \eqref{def p} and the 4th equality is obtained by \eqref{eq hcpi3}.

By a similar computation, the same error bound holds for $z\in {\rm Ext}({\cal B}) \cap \partial D_c$. One can check that the error is exponentially small in $N$ away from $\partial D_c$. By the small norm theorem \cite{DKMVZ 1999, Deift 1999}, we obtain that 
\begin{equation}\label{small norm}
A(z)=\bigg(I+{\cal O}\Big(\frac{1}{N}\Big)\bigg)A^\infty(z).
\end{equation}

\begin{prop}\label{thm ortho} Let $\xi$ and $\ell$ be given in \eqref{def xixi} and \eqref{def ell}.
Let $H(z)$  be defined in \eqref{def HH}. When $z\in D_c$, we have the following asymptotics,
\begin{equation}\label{pn strong}
    \begin{aligned}
    p_n(z)
  &=\ee^{\ii\frac{8}{3}\xi(\beta)^3}\ee^{Ntg(z)}\Big(\ee^{\frac{N}{2}\phi(z)}[H\Psi]_{21}^z+{\cal O}\big(N^{-1}\big)\Big),\\
  q_n(z)&=-\ee^{-\ii\frac{8}{3}\xi(\beta)^3}\ee^{Nt(g(z)-\ell)}\Big(\ee^{\frac{N}{2}\phi(z)}[H\Psi]_{11}^z+{\cal O}\big(N^{-1}\big)\Big),
    \end{aligned}
\end{equation} where
$$\begin{aligned}
[H\Psi]_{21}^z:=\big[H(z)\Psi(\xi(z);\hat s)\big]_{21}&=(1+h_{22}(z))\Psi_{21}(\xi(z);\hat s)+h_{21}(z)\Psi_{11}(\xi(z);\hat s),\\
[H\Psi]_{11}^z:=\big[H(z)\Psi(\xi(z);\hat s)\big]_{11}&=(1+h_{11}(z))\Psi_{11}(\xi(z);\hat s)+h_{12}(z)\Psi_{21}(\xi(z);\hat s).
\end{aligned}$$ 
Here $[H\Psi]_{jk}^z$ is the corresponding entry of $H\Psi$ and  the error bounds are uniform over $D_c$.
\end{prop}

\begin{proof}
When $z\in D_c \cap  {\rm Int}({\cal B})\setminus \Omega_+  $,  we have
\begin{equation}\label{computation intb}
    \begin{aligned}
    Y(z)&=\ee^{\frac{tN\ell}{2}\sigma_3}A(z)\ee^{\frac{tN\left(2g(z)-\ell\right)}{2}\sigma_3}=\ee^{\frac{tN\ell}{2}\sigma_3}\bigg(I+{\cal O}\Big(\frac{1}{N}\Big)\bigg)A^\infty(z)\ee^{\frac{tN\left(2g(z)-\ell\right)}{2}\sigma_3}\\
    &=\ee^{\frac{tN\ell}{2}\sigma_3}\bigg(I+{\cal O}\Big(\frac{1}{N}\Big)\bigg)\Phi(z){\cal P}(z)\ee^{\frac{tN\left(2g(z)-\ell\right)}{2}\sigma_3}\\
    &=\ee^{\frac{tN\ell}{2}\sigma_3}\bigg(I+{\cal O}\Big(\frac{1}{N}\Big)\bigg)\begin{pmatrix}
0&1\\-1&0
\end{pmatrix}\ee^{-\ii\frac{8}{3}\xi(\beta)^3\sigma_3}H(z)\Pi(\xi;\hat s)\ee^{\ii\frac{8}{3}\xi(\beta)^3\sigma_3} \ee^{\frac{tN\left(2g(z)-\ell\right)}{2}\sigma_3},
    \end{aligned}
\end{equation}
where the 1st equality is obtained by \eqref{def A},  the 2nd equality is obtained by \eqref{small norm}, the 3rd equality is obtained by \eqref{eq Ainfty}, the 4th equality is obtained by the definition of $\Phi(z)$ in \eqref{eq Phi} and the definition of ${\cal P}$ in \eqref{def p}. It follows that 
\begin{equation}
    p_n(z)=[Y(z)]_{11}=\ee^{\ii\frac{16}{3}\xi(\beta)^3} \ee^{Ntg(z)}\Big(\big[H(z)\Pi(\xi;\hat s)\big]_{21}+{\cal O}\big({N^{-1}}\big)\Big),
\end{equation}
where $[H\Pi]_{21}$ is the $(2,1)$-entry of $H\Pi$. By using the fact in \eqref{pi to psi}, the above identity can be further written as
\begin{equation}\label{ortho1 for later}
    \begin{aligned}
    p_n(z)
    &=\ee^{\ii\frac{16}{3}\xi(\beta)^3} \ee^{Ntg(z)}\bigg(\ee^{\ii\big(\frac{4}{3}\xi(z)^3+\hat s\xi(z)\big)}\big[H(z)\widetilde\Psi(\xi(z);\hat s)\big]_{21}+{\cal O}\Big(\frac{1}{N}\Big)\bigg)\\
     &=\ee^{\ii\frac{16}{3}\xi(\beta)^3} \ee^{Ntg(z)}\bigg(\ee^{\ii\big(\frac{4}{3}\xi(z)^3+\hat s\xi(z)\big)}\big[H(z)\Psi(\xi(z);\hat s)\big]_{21}+{\cal O}\Big(\frac{1}{N}\Big)\bigg)\\
  &=\ee^{\ii\frac{8}{3}\xi(\beta)^3}\ee^{Ntg(z)}\Big(\ee^{\frac{N}{2}\phi(z)}[H\Psi]_{21}^z+{\cal O}\big({N^{-1}}\big)\Big).
    \end{aligned}
\end{equation}  Here the 2nd equality is obtained by the definition of $\Psi$ in \eqref{def of psi0},  the 3rd equality is obtained by the definition of $\xi$ in \eqref{def xi} and \eqref{def xixi}. By the definition of $H$ in \eqref{def HH}, $$[H\Psi]_{21}^z=(1+h_{22}(z))\Psi_{21}(\xi(z);\hat s)+h_{21}(z)\Psi_{11}(\xi(z);\hat s).$$

Similarly, we have
\begin{equation}\label{ortho2 for later}
    \begin{aligned}
    q_n(z)&=[Y(z)]_{21}=-\ee^{Nt(g(z)-\ell)}\bigg(\ee^{\ii\big(\frac{4}{3}\xi(z)^3+\hat s\xi(z)\big)} [H\Psi]_{11}^z+{\cal O}\Big(\frac{1}{N}\Big)\bigg)\\
  &=-\ee^{-\ii\frac{8}{3}\xi(\beta)^3}\ee^{Nt(g(z)-\ell)}\Big(\ee^{\frac{N}{2}\phi(z)}[H\Psi]_{11}^z+{\cal O}\big({N^{-1}}\big)\Big),
    \end{aligned}
\end{equation} 
where $$[H\Psi]_{11}^z=(1+h_{11}(z))\Psi_{11}(\xi(z);\hat s)+h_{12}(z)\Psi_{21}(\xi(z);\hat s).$$ 

Following the same arguments as above  when $z\in D_c \cap \Omega_+  $ and using \eqref{def A},  we have
\begin{equation*}
    \begin{aligned}
   & Y(z)=\ee^{\frac{tN\ell}{2}\sigma_3}A(z)\ee^{\frac{tN\left(2g(z)-\ell\right)}{2}\sigma_3}\begin{pmatrix}
1&0\\1/w_{n,N}(z)&1
\end{pmatrix}\\
    &=\ee^{\frac{tN\ell}{2}\sigma_3}\bigg(I+{\cal O}\Big(\frac{1}{N}\Big)\bigg)A^\infty(z)\ee^{\frac{tN\left(2g(z)-\ell\right)}{2}\sigma_3}\begin{pmatrix}
1&0\\1/w_{n,N}(z)&1
\end{pmatrix}\\
    &=\ee^{\frac{tN\ell}{2}\sigma_3}\bigg(I+{\cal O}\Big(\frac{1}{N}\Big)\bigg)\Phi(z){\cal P}(z)\ee^{\frac{tN\left(2g(z)-\ell\right)}{2}\sigma_3}\begin{pmatrix}
1&0\\1/w_{n,N}(z)&1
\end{pmatrix}\\
    &=\ee^{\frac{tN\ell}{2}\sigma_3}\bigg(I+{\cal O}\Big(\frac{1}{N}\Big)\bigg)\begin{pmatrix}
0&1\\-1&0
\end{pmatrix}\ee^{-\ii\frac{8}{3}\zeta_\beta^3\sigma_3}H(z)\Pi(\xi(z),\hat s)\ee^{\ii\frac{8}{3}\zeta_\beta^3\sigma_3}\ee^{\frac{tN\left(2g(z)-\ell\right)}{2}\sigma_3}\begin{pmatrix}
1&0\\1/w_{n,N}(z)&1
\end{pmatrix}\\
 &=\ee^{\frac{tN\ell}{2}\sigma_3}\bigg(I+{\cal O}\Big(\frac{1}{N}\Big)\bigg)\begin{pmatrix}
0&1\\-1&0
\end{pmatrix}\ee^{-\ii\frac{8}{3}\zeta_\beta^3\sigma_3}H(z)\widetilde\Psi(\xi,\hat s)\ee^{\ii\left(\frac{4}{3}\xi^3+\hat s\xi\right)\sigma_3}\ee^{\ii\frac{8}{3}\xi(\beta)^3\sigma_3} \ee^{\frac{tN\left(2g(z)-\ell\right)}{2}\sigma_3} \begin{pmatrix}
1&0\\1/w_{n,N}(z)&1
\end{pmatrix}
\\
 &=\ee^{\frac{tN\ell}{2}\sigma_3}\bigg(I+{\cal O}\Big(\frac{1}{N}\Big)\bigg)\begin{pmatrix}
0&1\\-1&0
\end{pmatrix}\ee^{-\ii\frac{8}{3}\zeta_\beta^3\sigma_3}H(z)\widetilde\Psi(\xi,\hat s)\begin{pmatrix}
1&0\\1&1
\end{pmatrix}\ee^{\ii\left(\frac{4}{3}\xi^3+\hat s\xi\right)\sigma_3}\ee^{\ii\frac{8}{3}\xi(\beta)^3\sigma_3} \ee^{\frac{tN\left(2g(z)-\ell\right)}{2}\sigma_3}\\
 &=\ee^{\frac{tN\ell}{2}\sigma_3}\bigg(I+{\cal O}\Big(\frac{1}{N}\Big)\bigg)\begin{pmatrix}
0&1\\-1&0
\end{pmatrix}\ee^{-\ii\frac{8}{3}\zeta_\beta^3\sigma_3}H(z)\Psi(\xi,\hat s)\ee^{\ii\left(\frac{4}{3}\xi^3+\hat s\xi\right)\sigma_3}\ee^{\ii\frac{8}{3}\xi(\beta)^3\sigma_3} \ee^{\frac{tN\left(2g(z)-\ell\right)}{2}\sigma_3}
.
    \end{aligned}
\end{equation*}
where the 1st equality is obtained by \eqref{def A},  the 2nd equality is obtained by \eqref{small norm}, the 3rd equality is obtained by \eqref{eq Ainfty}, the 4th equality is obtained by the definition of $\Phi(z)$ in \eqref{eq Phi} and the definition of ${\cal P}$ in \eqref{def p}, the 5th equality is obtained by the fact in \eqref{pi to psi},  the 6th equality is obtained by the definition of $\omega_{n,N}(z)$ in \eqref{def omega}, the definition of $g$-function in \eqref{def gfunction}, and the fact that
$$\ee^{\ii\left(\frac{4}{3}\xi^3+\hat s\xi\right)\sigma_3}\ee^{\ii\frac{8}{3}\xi(\beta)^3\sigma_3} \ee^{\frac{tN\left(2g(z)-\ell\right)}{2}\sigma_3}\begin{pmatrix}
1&0\\1/w_{n,N}(z)&1
\end{pmatrix}=\begin{pmatrix}
1&0\\1&1
\end{pmatrix}\ee^{\ii\left(\frac{4}{3}\xi^3+\hat s\xi\right)\sigma_3}\ee^{\ii\frac{8}{3}\xi(\beta)^3\sigma_3} \ee^{\frac{tN\left(2g(z)-\ell\right)}{2}\sigma_3},$$ 
and the 7th equality is obtained by the definition of $\Psi$ in \eqref{def of psi0}.

Similar to the computations in \eqref{ortho1 for later} and \eqref{ortho2 for later}, we  see that \eqref{pn strong} holds for $z\in D_c\cap \Omega_+ $.

The analogous calculation can be done for $z\in D_c \cap  {\rm Ext}({\cal B})$.  \end{proof}

\begin{prop}\label{thm ortho1} 
When $z$ is away from $D_c$ and $z\in {\rm Int}({\cal B})$, we have the  asymptotics,
\begin{equation}\label{pn inside b}
    \begin{aligned}
p_n(z)&=\ee^{\ii\frac{16}{3}\xi(\beta)^3}\ee^{Naz+tN\ell}\bigg(\frac{q(\hat s)}{2r_1N^{1/3}(z-b_c^*)}+\frac{r_2p_{21}(\hat s)}{r_1^3N^{2/3}(z-b_c^*)}-\frac{p_{21}(\hat s)}{r_1^2N^{2/3}(z-b_c^*)^2}+{\cal O}\Big(\frac{1}{N}\Big)\bigg),\\
  q_n(z)&=-\ee^{Naz}\bigg(1-\frac{r(\hat s)}{2r_1N^{1/3}(z-b_c^*)}-\frac{p_{11}(\hat s)}{r_1^2N^{2/3}(z-b_c^*)^2}+{\cal O}\Big(\frac{1}{N}\Big)\bigg), \\
  [Y(z)]_{12}&=\ee^{-Naz}\bigg(1+\frac{r(\hat s)}{2r_1N^{1/3}(z-b_c^*)}-\frac{p_{22}(\hat s)}{r_1^2N^{2/3}(z-b_c^*)^2}+{\cal O}\Big(\frac{1}{N}\Big)\bigg).
    \end{aligned}
\end{equation}  Here the error bounds are uniform over a compact subset of ${\rm Int}({\cal B})\setminus {D}_c$.

When $z$ is away from $D_c$ and $z\in {\rm Ext}({\cal B})$, we have the  asymptotics,
\begin{equation}\label{pn outside b}
\begin{aligned}
p_n(z)&=z^n\left(\frac{z}{z-a}\right)^{Nc}\bigg(1+\frac{r(\hat s)}{2r_1N^{1/3}(z-b_c^*)}-\frac{p_{22}(\hat s)}{r_1^2N^{2/3}(z-b_c^*)^2}+{\cal O}\Big(\frac{1}{N}\Big)\bigg),\\
q_n(z)&=-\frac{z^n}{\ee^{\ii\frac{16}{3}\xi(\beta)^3}\ee^{tN\ell}}\left(\frac{z}{z-a}\right)^{Nc}\bigg(-\frac{q(\hat s)}{2r_1N^{1/3}(z-b_c^*)}+\frac{r_2p_{12}(\hat s)}{r_1^3N^{2/3}(z-b_c^*)}-\frac{p_{12}(\hat s)}{r_1^2N^{2/3}(z-b_c^*)^2}+{\cal O}\Big(\frac{1}{N}\Big)\bigg),\\
[Y(z)]_{12}&=-\frac{\ee^{\ii\frac{16}{3}\xi(\beta)^3}\ee^{tN\ell}}{z^n}\left(\frac{z-a}{z}\right)^{Nc}\bigg(\frac{q(\hat s)}{2r_1N^{1/3}(z-b_c^*)}+\frac{\gamma_2p_{21}(\hat s)}{r_1^3N^{2/3}(z-b_c^*)}-\frac{p_{21}(\hat s)}{r_1^2N^{2/3}(z-b_c^*)^2}+{\cal O}\Big(\frac{1}{N}\Big)\bigg).
\end{aligned}      
\end{equation} Here the error bounds are uniform over a compact subset of ${\rm Ext}({\cal B})\setminus {D}_c$.

Let ${D}_\tau:=\{z: |z-b_c^*|\leq\frac{1}{N^\tau}\}$ with $0\leq \tau<\frac{1}{3}$. When $z\in D_c\setminus D_\tau$, we have the  asymptotics, 
\begin{equation}\label{pn omegaplus}
\begin{aligned}
p_n(z)&=\ee^{Naz+tN\ell}\ee^{\ii\frac{16}{3}\xi(\beta)^3}\bigg(\frac{q(\hat s)}{2r_1N^{1/3}(z-b_c^*)}+\frac{r_2p_{21}(\hat s)}{r_1^3N^{2/3}(z-b_c^*)}-\frac{p_{21}(\hat s)}{r_1^2N^{2/3}(z-b_c^*)^2}+{\cal O}\Big(\frac{1}{N^{1-3\tau}}\Big)\bigg)\\
&\quad +\ee^{Naz+tN\ell}\ee^{-N\phi(z)}\bigg(1+\frac{r(\hat s)}{2r_1N^{1/3}(z-b_c^*)}-\frac{p_{22}(\hat s)}{r_1^2N^{2/3}(z-b_c^*)^2}+{\cal O}\Big(\frac{1}{N^{1-3\tau}}\Big)\bigg),\\
q_n(z) &=-\ee^{Naz}\bigg(1-\frac{r(\hat s)}{2r_1N^{1/3}(z-b_c^*)}-\frac{p_{11}(\hat s)}{r_1^2N^{2/3}(z-b_c^*)^2}+{\cal O}\Big(\frac{1}{N^{1-3\tau}}\Big)\bigg)\\
&\quad+\ee^{-\ii\frac{16}{3}\xi(\beta)^3}\ee^{Naz}\ee^{-N\phi(z)}\bigg(\frac{q(\hat s)}{2r_1N^{1/3}(z-b_c^*)}-\frac{r_2p_{12}(\hat s)}{r_1^3N^{2/3}(z-b_c^*)}+\frac{p_{12}(\hat s)}{r_1^2N^{2/3}(z-b_c^*)^2}+{\cal O}\Big(\frac{1}{N^{1-3\tau}}\Big)\bigg).
\end{aligned}        
\end{equation} Here the error bounds are uniform over $D_c\setminus D_\tau$.
\end{prop}

\begin{proof}
When $z\in {\rm Int}({\cal B})\setminus D_c$, using \eqref{def A},  we have
\begin{equation}\label{comp inBnotd}
    \begin{aligned}
   & Y(z)=\ee^{\frac{tN\ell}{2}\sigma_3}\bigg(I+{\cal O}\Big(\frac{1}{N}\Big)\bigg)A^\infty(z)\ee^{\frac{tN\left(2g(z)-\ell\right)}{2}\sigma_3}\\
   &=\ee^{\frac{tN\ell}{2}\sigma_3}\bigg(I+{\cal O}\Big(\frac{1}{N}\Big)\bigg)\Phi(z)S(z)\ee^{\frac{tN\left(2g(z)-\ell\right)}{2}\sigma_3}\\
     &=\ee^{\frac{Nt\ell}{2}\sigma_3}\bigg(I+{\cal O}\Big(\frac{1}{N}\Big)\bigg)\begin{pmatrix}
0&1\\-1&0
\end{pmatrix}\ee^{-\ii\frac{8}{3}\xi(\beta)^3\sigma_3}\bigg(I+\frac{S_{11}(\hat s)+S_{21}(\hat s)}{z-b_c^*}+\frac{S_{22}(\hat s)}{(z-b_c^*)^2}\bigg)\ee^{\ii\frac{8}{3}\xi(\beta)^3\sigma_3} \ee^{\frac{Nt\left(2g(z)-\ell\right)}{2}\sigma_3}.
    \end{aligned}
\end{equation}
where the 1st equality is obtained by \eqref{small norm}, the 2nd equality is obtained by \eqref{eq Ainfty}, the 3rd equality is obtained by the definition of $\Phi(z)$ in \eqref{eq Phi} and the definition of $S$ in \eqref{def SS}. 

Using the definitions of $S_{11}$ in \eqref{def s11}, $S_{21}$ in \eqref{def s21}, and $S_{22}$ in \eqref{def s22}, we have \eqref{pn inside b}. 

When $z$ is away from $D_c$, similar computations can be done for $z$ in ${\rm Ext}({\cal B})\setminus\Omega_-$ and $\Omega_\pm$. 

When $z\in D_c\setminus{D}_\tau$, using $\xi$ in \eqref{zeta expanssion}, we have $\xi ={\cal O}(N^{1/3-\tau})$. For $z\in \Omega_+  $  we have
\begin{equation}\label{comp in annulus}
    \begin{aligned}
    &Y(z)=\ee^{\frac{tN\ell}{2}\sigma_3}A(z)\ee^{\frac{tN\left(2g(z)-\ell\right)}{2}\sigma_3}\begin{pmatrix}
1&0\\1/w_{n,N}(z)&1
\end{pmatrix}\\
&=\ee^{\frac{tN\ell}{2}\sigma_3}\bigg(I+{\cal O}\Big(\frac{1}{N}\Big)\bigg)A^\infty(z)\ee^{\frac{tN\left(2g(z)-\ell\right)}{2}\sigma_3}\begin{pmatrix}
1&0\\1/w_{n,N}(z)&1
\end{pmatrix}\\
    &=\ee^{\frac{tN\ell}{2}\sigma_3}\bigg(I+{\cal O}\Big(\frac{1}{N}\Big)\bigg)\Phi(z){\cal P}(z)\ee^{\frac{tN\left(2g(z)-\ell\right)}{2}\sigma_3}\begin{pmatrix}
1&0\\1/w_{n,N}(z)&1
\end{pmatrix}\\
&=\ee^{\frac{tN\ell}{2}\sigma_3}\bigg(I+{\cal O}\Big(\frac{1}{N}\Big)\bigg)\begin{pmatrix}
0&1\\-1&0
\end{pmatrix}{\cal P}(z)\ee^{\frac{tN\left(2g(z)-\ell\right)}{2}\sigma_3}\begin{pmatrix}
1&0\\1/w_{n,N}(z)&1
\end{pmatrix}\\
    &=\ee^{\frac{tN\ell}{2}\sigma_3}\Big(I+{\cal O}\big(\frac{1}{N}\big)\Big)\begin{pmatrix}
0&1\\-1&0
\end{pmatrix}\ee^{-\ii\frac{8}{3}\xi(\beta)^3\sigma_3}H(z)\Pi(\xi;x)\ee^{\ii\frac{8}{3}\xi(\beta)^3\sigma_3} \ee^{\frac{tN\left(2g(z)-\ell\right)}{2}\sigma_3}\begin{pmatrix}
1&0\\1/w_{n,N}(z)&1
\end{pmatrix}\\
    &=\ee^{\frac{tN\ell}{2}\sigma_3}\Big(I+{\cal O}\big(\frac{1}{N}\big)\Big)\begin{pmatrix}
0&1\\-1&0
\end{pmatrix}\ee^{-\ii\frac{8}{3}\xi(\beta)^3\sigma_3} \bigg(I+\frac{S_{11}(\hat s)+S_{21}(\hat s)}{z-b_c^*}+\frac{S_{22}(\hat s)}{(z-b_c^*)^2}+{\cal O}\Big(\frac{1}{N^{1-3\tau}}\Big)\bigg)\\
&\quad \times\ee^{\ii\frac{8}{3}\xi(\beta)^3\sigma_3} \ee^{\frac{tN\left(2g(z)-\ell\right)}{2}\sigma_3}\begin{pmatrix}
1&0\\1/w_{n,N}(z)&1
\end{pmatrix},
    \end{aligned}
\end{equation}
where the last equality is obtained by $H(z)\Pi(\xi;x)$  in \eqref{eq hc21} and the following fact
\begin{equation*}
\bigg(\frac{ S_{21}(\hat s)}{z-b_c^*}+\frac{ S_{22}(\hat s)}{(z-b_c^*)^2}\bigg)\frac{S_{11}(\hat s)}{z-b_c^*}-\bigg(\frac{\Pi_2(\hat s)}{\xi^2}+\Big(H_1(z)-I\Big)\frac{\Pi_1(\hat s)}{2\ii\xi}\bigg)\frac{S_{11}(\hat s)}{z-b_c^*}+{\cal O}\left(\frac{1}{\xi^3}\right)={\cal O}\Big(\frac{1}{N^{1-3\tau}}\Big).
\end{equation*}  Using the definitions of $S_{11}$ in \eqref{def s11}, $S_{21}$ in \eqref{def s21}, and $S_{22}$ in \eqref{def s22}, and the definition of $g$-function in \eqref{def gfunction}, we have \eqref{pn omegaplus}. 

A similar computation can be done for $z$ in other regions of $D_c\setminus D_\tau.$ 

In particular, when $z\in D_c\setminus D_\tau$ and $z\in {\rm Int}({\cal B})$, $\ee^{N\phi(z)}$ is exponentially small in $N$, \eqref{pn omegaplus} matches the corresponding expressions in \eqref{pn inside b} by simply setting $\tau=0$ in the error bound ${\cal O}(N^{3\tau-1})$.
\end{proof}

\subsection{Recurrence relations}\label{section 2.3}

By Theorem \ref{thm 21}, we will also need the asymptotic behavior of $p_{n-1}$ and $p_{n+1}$ to obtain asymptotic behavior of the correlation kernel. In order to obtain $p_{n-1}$ and $p_{n+1}$ we need the following recurrence relations. Similar recurrence relations have been obtained in \cite[Appendix B]{Lee 2017} for generating numerical plots of the zeros of planar orthogonal polynomials with the external potential $Q(z)=|z|^2+\frac{2c}{N}\log\frac{1}{|z-a|}$. 

Let us define
\begin{equation}\label{transform y to ytilde}
 \widetilde{Y}(z):=\widetilde{Y}_{n}(z)=Y(z)\begin{bmatrix}
\left(\frac{z-a}{z}\right)^{Nc}\frac{1}{\ee^{Naz}}&0 \\
0&z^{n}
\end{bmatrix},   
\end{equation} where $Y$ satisfies the Riemann-Hilbert problem in \eqref{rhp for y}.  The Riemann-Hilbert problem for $\widetilde{Y}(z)$ is given by
\begin{equation}\nonumber
\begin{cases}
\widetilde{Y}_+(z)=\widetilde{Y}_-(z)\begin{bmatrix}
1&1 \\
0&1
\end{bmatrix} ,&z\in\Gamma,\vspace{0.1cm}\\
\widetilde{Y}_+(z)=\widetilde{Y}_-(z)\begin{bmatrix}
\ee^{2Nc\pi\ii}&0 \\
0&1
\end{bmatrix} ,&z\in (0,a),\\
\widetilde{Y}(z)=\displaystyle\left(I+\mathcal
{O}\left(\frac{1}{z}\right)\right)\begin{bmatrix}
\left(\frac{z-a}{z}\right)^{Nc}\frac{z^n}{\ee^{Naz}}&0 \\
0&1
\end{bmatrix},& z\to\infty,\\
\widetilde{Y}(z) \text{ is holomorphic},& \text{otherwise}.
\end{cases}
\end{equation}
Since $\det{Y}(z)\equiv 1$, the inverse of
${\widetilde{Y}(z)}$ exists in
$\mathbb{C}\setminus(\Gamma\cup(0,a))$,  and we  define
\begin{equation}\label{eq An}
A_n(z):=\frac{\dd\widetilde{Y}_n(z)}{\dd z}{\widetilde{Y}_n(z)}^{-1}.    
\end{equation}
As $\widetilde{Y}_n(z)$ and $\frac{\dd\widetilde{Y}_n(z)}{\dd z}$ have
the same jump matrices, the matrix function $A_n(z)$ is meromorphic and can be determined by
identifying the singularities at $\infty$, $0$ and $a$. For $z\to\infty,$ writing
\begin{equation}\label{Ytilde inf}
\widetilde{Y}_n(z)=\left(I+\frac{1}{z}\begin{bmatrix}
a_n&b_n \\
c_n&d_n
\end{bmatrix}+{\cal O}\Big(\frac{1}{z^2}\Big)\right)\begin{bmatrix}
\left(\frac{z-a}{z}\right)^{Nc}\frac{z^n}{\ee^{Naz}}&0 \\
0&1
\end{bmatrix},     
\end{equation}
we have
\begin{equation}\label{eq anz}
A_n(z)=\begin{bmatrix}
-Na&0 \\
0&0
\end{bmatrix}+\frac{1}{z}\begin{bmatrix}
n&N a b_n \\
-N a c_n&0
\end{bmatrix}+\mathcal
{O}\Big(\frac{1}{z^2}\Big).
\end{equation}

Similarly, for $z\to 0$, writing
\begin{equation}\label{Ytilde 0}
\widetilde{Y}_n(z)=\begin{bmatrix}
\alpha_n&\beta_n \\
\gamma_n&\eta_n
\end{bmatrix}\left(I+\mathcal
{O}(z)\right)\begin{bmatrix}
\left(\frac{z-a}{z}\right)^{Nc}\frac{1}{\ee^{Naz}}&0 \\
0&z^{n}
\end{bmatrix},
\end{equation}
we have
\begin{equation}\nonumber
A_n(z)=\frac{1}{z}\begin{bmatrix}
-Nc-(Nc+n)\beta_n\gamma_n&(Nc+n)\alpha_n\beta_n \\
-(Nc+n)\gamma_n\eta_n&n+(Nc+n)\beta_n\gamma_n
\end{bmatrix}+{\cal O}(1).
\end{equation}
Therefore we have
\begin{equation}\label{def An}
\begin{split}
A_n(z)=\begin{bmatrix}
-Na&0 \\
0&0
\end{bmatrix}&+\frac{1}{z}\begin{bmatrix}
-Nc-(Nc+n)\beta_n\gamma_n&(Nc+n)\alpha_n\beta_n \\
-(Nc+n)\gamma_n\eta_n&n+(Nc+n)\beta_n\gamma_n
\end{bmatrix}\\
& +\frac{1}{z-a}\begin{bmatrix}
(Nc+n)\left(1+\beta_n\gamma_n\right)&Nab_n-(Nc+n)\alpha_n\beta_n \\
-Nac_n+(Nc+n)\gamma_n\eta_n&-n-(Nc+n)\beta_n\gamma_n
\end{bmatrix}.
\end{split}
\end{equation}
As $z$ goes to $\infty$, one can see that the coefficient matrix of $1/z$ in $A_n(z)$ at \eqref{def An}   exactly matches  the one of $A_n(z)$ at \eqref{eq anz}.
Note that the notations $c_n$ and $\gamma_n$ are similar but unrelated to the charge $c$ in the external potential $Q$ in \eqref{def Q} and $\gamma_c$ in \eqref{def gammacl}, respectively.

Defining $M_n(z):=\widetilde{Y}_{n+1}(z){\widetilde{Y}_n(z)}^{-1}$, knowing $\widetilde{Y}_{n+1}(z)$ and $\widetilde{Y}_{n}(z)$ have
the same jump matrices, 
by a similar procedure as above, we obtain that
\begin{equation}\label{def mn}
M_n(z)=\begin{bmatrix}
z+a_{n+1}-a_n&-b_n \\
c_{n+1}&1
\end{bmatrix}.
\end{equation}

Since $\det M_n(z)=z$, we have
\begin{equation}\label{def mn1}
M_n(z)^{-1}=\frac{1}{z}\begin{bmatrix}
1&b_n \\
-c_{n+1}&z+a_{n+1}-a_n
\end{bmatrix},\quad  a_{n+1}-a_n+b_nc_{n+1}=0.
\end{equation}

The compatibility of the Lax pair,
\begin{equation}\label{ytilde lax}
\frac{\dd\widetilde{Y}_n(z)}{\dd z}=A_n(z){\widetilde{Y}_n(z)},\quad 
\widetilde{Y}_{n+1}(z)=M_n(z){\widetilde{Y}_n(z)},
\end{equation}
gives
\begin{equation}\nonumber
A_{n+1}(z)M_n(z)=\frac{\dd M_n(z)}{\dd z}+M_n(z)A_n(z).
\end{equation}
This yields the following recurrence relation:
\begin{equation}\label{recurence1}
\begin{split}
&a_{n+1}=a_n+\frac{b_n \left(1+\beta_n\gamma_n\right)}{\alpha_n
\beta_n},\quad\alpha_{n+1}=\frac{b_n}{\beta_n},\quad
\gamma_{n+1}=-\frac{1}{\beta_n},
\\
 &b_{n+1}=\frac{(1+n+a^2N)b_n}{aN}+\frac{(Nc+n)\alpha_n\beta_n}{N}+\frac{b^2_n\left(1+\beta_n\gamma_n\right)}
{\alpha_n\beta_n},\\
&b_{n-1}=\alpha_n\beta_{n-1}=-\frac{\alpha_n}{\gamma_n},
\\
&\beta_{n+1}=\frac{\tilde{c}_n}{(1+Nc+n)
\left((Nc+n)\alpha_n\beta_n-aNb_n\right)\alpha^2_n\beta_n},
\end{split}
\end{equation}
where
\begin{equation}\nonumber
\begin{split}
\tilde{c}_n&=a^2N-Nc-a(1+2(Nc+n))\alpha_n\beta_n
+\left(a^2N-Nc-a(Nc+n)\alpha_n\beta_n\right)\beta_n\gamma_n\\
&+(Nc+n)(Nc+n+1)\alpha^3_n\beta^3_n+aN^2b^3_n
\left(1+\beta_n\gamma_n\right)^2.
\end{split}
\end{equation}

\begin{prop} \label{prop 2.5}
The following identities hold,
    \begin{align}
    p_{n-1}(z)
    &=\frac{1}{z}\Big(p_{n}(z)-\frac{\alpha_n}{\gamma_n}q_{n}(z)\Big),\label{identity 1}\\
 p_{n+1}(z) &=\bigg(z+\frac{b_n (1+\beta_n\gamma_n)}{\alpha_n
\beta_n}\bigg)p_n(z)-b_nq_n(z),\label{identity 2}\\
 \frac{p_n'(z)}{N}
&=\bigg(\frac{t+(c+t)\beta_n\gamma_n}{z-a}-\frac{(c+t)\beta_n\gamma_n}{z}\bigg)p_n(z)+\bigg(\frac{(c+t)\beta_n\alpha_n}{z}+\frac{ab_n-(c+t)\beta_n\alpha_n}{z-a}\bigg)q_n(z),\label{identity 3}\quad \\
h_n&=\frac{\Gamma(Nc+n+1)}{N^{Nc+n+1}}\pi \beta_n.\label{identity 4}
    \end{align}
\end{prop}
\begin{proof}

Using  $\widetilde{Y}_{n+1}(z)=M_n(z){\widetilde{Y}_n(z)}$ in \eqref{ytilde lax}, the definition of ${\widetilde{Y}_n(z)}$ in \eqref{transform y to ytilde} with \eqref{def mn} and  \eqref{def mn1}, we have
\begin{equation}
    \begin{aligned}
    Y_n(z)
    =M_n(z)^{-1}Y_{n+1}(z)\begin{bmatrix}
1&0 \\
0&z
\end{bmatrix}\quad \mbox{and}\quad
 Y_{n+1}(z)=M_n(z)Y_{n}(z)\begin{bmatrix}
1&0 \\
0&z^{-1}
\end{bmatrix}.
    \end{aligned}
\end{equation} It follows that
\begin{equation}\label{eq pn pn1}
    \begin{aligned}
    p_n(z)
    =\frac{1}{z}\Big(p_{n+1}(z)+b_nq_{n+1}(z)\Big)\quad \mbox{and}\quad
 p_{n+1}(z)=(z+a_{n+1}-a_n)p_n(z)-b_nq_n(z).
    \end{aligned}\end{equation}
Combining \eqref{recurence1} and \eqref{eq pn pn1}, we can see that \eqref{identity 1} and \eqref{identity 2} hold.
    
By \eqref{eq An} and the definition of ${\widetilde{Y}_n(z)}$ in \eqref{transform y to ytilde} we have
    \begin{equation}\nonumber
\frac{\dd{Y}_n(z)}{\dd z}=A_n(z){{Y}_n(z)}-Y_n(z)\frac{\dd}{\dd z}\begin{bmatrix}
\left(\frac{z-a}{z}\right)^{Nc}\frac{1}{\ee^{Naz}}&0 \\
0&z^{n}
\end{bmatrix}\cdot \begin{bmatrix}
\left(\frac{z-a}{z}\right)^{Nc}\frac{1}{\ee^{Naz}}&0 \\
0&z^{n}
\end{bmatrix}^{-1}.    
\end{equation}
    
Using the above equation and the explicit expression of $A_n(z)$ in \eqref{def An} we have
$$\begin{aligned}
p_n'(z)&=\bigg(-N a+\frac{-Nc-(Nc+n)\beta_n\gamma_n}{z}+\frac{Nc+n+(Nc+n)\beta_n\gamma_n}{z-a}\bigg)p_n(z)\\
&\quad +\bigg(\frac{(Nc+n)\beta_n\alpha_n}{z}+\frac{Nab_n-(Nc+n)\beta_n\alpha_n}{z-a}\bigg)q_n(z)-\bigg(\frac{Nc}{z-a}-\frac{Nc}{z}-Na\bigg)p_n(z)\\
&=\bigg(-\frac{(Nc+Nt)\beta_n\gamma_n}{z}+\frac{Nt+(Nc+Nt)\beta_n\gamma_n}{z-a}\bigg)p_n(z)\\
&\quad +\bigg(\frac{(Nc+Nt)\beta_n\alpha_n}{z}+\frac{Nab_n-(Nc+Nt)\beta_n\alpha_n}{z-a}\bigg)q_n(z)\\
&=N\bigg(-\frac{(c+t)\beta_n\gamma_n}{z}+\frac{t+(c+t)\beta_n\gamma_n}{z-a}\bigg)p_n(z)+N\bigg(\frac{(c+t)\beta_n\alpha_n}{z}+\frac{ab_n-(c+t)\beta_n\alpha_n}{z-a}\bigg)q_n(z).
\end{aligned}$$ 

Let us define the norming constant $\widetilde h_n$ of $p_n(z)$ with respect to the weight $\omega_{n,N}(z)\dd z$ by
$$\widetilde h_n:=\int_\Gamma p_{n}^2(z)\omega_{n,N}(z)\dd z.$$ Using Proposition 7.1 in \cite{Ba 2015}
we have, 
    \begin{equation}
    h_n=-\frac{\Gamma(Nc+n+1)}{2\ii N^{Nc+n+1}}\frac{\widetilde{h}_n}{p_{n+1}(0)}=\frac{\Gamma(Nc+n+1)}{N^{Nc+n+1}}\frac{\pi b_n}{\alpha_{n+1}}=\frac{\Gamma(Nc+n+1)}{N^{Nc+n+1}}\pi \beta_n,
\end{equation}
where the 2nd equality is obtained by the fact that $$\widetilde h_n=-\lim_{z\to\infty }2\pi\ii z^{n+1}[Y_n(z)]_{12}=-2\pi\ii b_n,$$ and the last equality is obtained by \eqref{recurence1}. This ends the proof.
\end{proof}

\section{The C-D identity: Proof of Lemma \ref{preker}}\label{section 3.1}
In this section, we prove Lemma \ref{preker} by combining the C-D identity in Theorem \ref{thm 21} with the asymptotics of orthogonal polynomials in Propositions \ref{thm ortho} and \ref{thm ortho1}, and the norming constants and the recurrence relations in Section \ref{section 2.3}.

\begin{lemma}\label{preker}
Let ${\cal K}_{n}(z,\zeta)$ be the pre-kernel defined in \eqref{def prekernel}.  Defining the local scaling coordinates $u$ and $v$ such that
\begin{equation}\label{eq zzeta def}
z=b_c+\frac{\ii \gamma_c}{N^{1/3}} u\quad  \mbox{and} \quad \zeta=b_c+\frac{\ii \gamma_c}{N^{1/3}}v,    
\end{equation}
 we have
\begin{equation}\label{eq lem26}
\begin{aligned}
\frac{1}{N^{5/6}}\overline{\partial}_{\zeta} {\cal K}_{n}(z, \zeta)
&=\exp\bigg(\frac{N^{1/3}\gamma_c^2(u-\bar v)^2}{2}\bigg)\exp \bigg(\ii s (u-\bar v)-\frac{4\ii(b_c+\sqrt c)}{3a}(u^3-\bar v^3)\bigg)\\
&\quad \times \frac{N^{2/3}}{ \sqrt 2 \pi^{3/2}}\bigg(\Psi_{11}(u;s)\overline{\Psi_{11}(v;s)}-\Psi_{21}(u;s)\overline{\Psi_{21}(v;s)}+{\cal O}\Big(\frac{1}{N^{1/3}}\Big)\bigg).
\end{aligned}
\end{equation}
Here the error bound is uniform over $u$ and $v$ in a compact subset of $\CC$.
\end{lemma}
\begin{proof}

 By Theorem \ref{thm 21}, we write
  \begin{equation}\label{main identity}
\overline{\partial}_{\zeta} {\cal K}_{n}(z, \zeta)= {\rm I}_n- {\rm II}_n  ,
\end{equation}
where 
\begin{equation}\label{eq IN}
{\rm I}_n:={\rm I}_n(z,\zeta)={\ee^{ -Nz\overline{\zeta} }}  \frac{1}{ \frac{n+Nc}{N}h_{n-1}-h_n  }
 \overline{ \psi'_{n}(\zeta) }    \Big( \psi_{n}(z)-z \psi_{n-1}(z) \Big),
\end{equation}
and
\begin{equation}\label{eq IIN}
{\rm II}_n:={\rm II}_n(z,\zeta)={\ee^{ -Nz\overline{\zeta} }}\frac{p_{n+1}(a)}{p_n(a)} \frac{Nh_n/h_{n-1} }{ \frac{n+Nc+1}{N} h_n-h_{n+1}    } 
 \overline{ \psi_{n-1}(\zeta) }  \Big(  \psi_{n+1}(z)-z\psi_n(z)  \Big).
\end{equation}

In order to compute \eqref{main identity}, let us prepare several useful identities.

First of all, by \eqref{identity 1}, \eqref{identity 2}, and the definition of $\psi_n$ in \eqref{def psi},  we have    
\begin{align}
\psi_{n}(z)- z\psi_{n-1}(z) 
    &=\frac{\alpha_n}{\gamma_n}(z-a)^{Nc}q_{n}(z),\label{psi differ1}\\
 \psi_{n+1}(z)-z\psi_{n}(z) &=\frac{b_n \left(1+\beta_n\gamma_n\right)}{\alpha_n
\beta_n}\psi_{n}(z)-b_n(z-a)^{Nc}q_n(z).\label{psi differ2}
    \end{align}
Using \eqref{identity 2}, we have the following identity
\begin{equation}\label{eq pn1 and pn11}
\frac{p_{n+1}(a)}{p_n(a)}=a+\frac{b_n \left(1+\beta_n\gamma_n\right)}{\alpha_n
\beta_n}-b_n\frac{q_n(a)}{p_n(a)}.    
\end{equation}

Moreover, using \eqref{identity 3} and the definition of $\psi_n$ in \eqref{def psi}, we have
\begin{equation}\label{psin p}
\begin{aligned}
\psi_n'(z)&=\frac{Nc}{z-a}\psi_n(z)+(z-a)^{Nc}p'_n(z)\\
&=\frac{Nc}{z-a}\psi_n(z)+N\bigg(\frac{t+(c+t)\beta_n\gamma_n}{z-a}-\frac{(c+t)\beta_n\gamma_n}{z}\bigg)\psi_n(z)\\
&\quad +N\bigg(\frac{(c+t)\beta_n\alpha_n}{z}+\frac{ab_n-(c+t)\beta_n\alpha_n}{z-a}\bigg)(z-a)^{Nc}q_n(z)\\
&=N(c+t)\bigg(\frac{1+\beta_n\gamma_n}{z-a}-\frac{\beta_n\gamma_n}{z}\bigg)\psi_n(z)+N\bigg(\frac{(c+t)\beta_n\alpha_n}{z}+\frac{ab_n-(c+t)\beta_n\alpha_n}{z-a}\bigg)(z-a)^{Nc}q_n(z).
\end{aligned}    
\end{equation}

Lastly, using \eqref{identity 4} in Proposition \ref{prop 2.5} and \eqref{recurence1}, we have
\begin{align}
    \frac{1}{ \frac{n+Nc}{N}h_{n-1}-h_n  }&=\frac{N^{Nc+n+1}}{\Gamma(Nc+n+1)}\frac{1}{\pi \big(\beta_{n-1}-\beta_n\big)}=-\frac{N^{Nc+n+1}}{\Gamma(Nc+n+1)}\frac{1}{\pi \big(\frac{1}{\gamma_n}+\beta_n\big)},\label{hn1}\\
    \frac{1}{ \frac{n+Nc+1}{N}h_{n}-h_{n+1}}&=\frac{N^{Nc+n+2}}{\Gamma(Nc+n+2)}\frac{1}{\pi \big(\beta_{n}-\beta_{n+1}\big)},\label{hn2}\\
\frac{Nh_n}{h_{n-1}}&=(Nc+n)\frac{\beta_n}{\beta_{n-1}}=-(Nc+n)\beta_n\gamma_n.\label{hn3}
\end{align}

Using \eqref{hn1}, \eqref{psin p} and \eqref{psi differ1}, we rewrite ${\rm I}_n$ from  \eqref{eq IN} as 
\begin{align}\label{eq in middle}
\begin{split}
&{\rm I}_n=-{\ee^{ -Nz\overline{\zeta} }} \frac{N^{Nc+n+1}}{\pi\Gamma(Nc+n+1)}\frac{N\alpha_n}{ 1+\beta_n\gamma_n}(z-a)^{Nc}q_{n}(z)\\
&\quad\times \Bigg((c+t)\bigg(\frac{1+\beta_n\gamma_n}{\zeta-a}-\frac{\beta_n\gamma_n}{\zeta}\bigg)\psi_n(\zeta)+\bigg(\frac{(c+t)\beta_n\alpha_n}{\zeta}+\frac{ab_n-(c+t)\beta_n\alpha_n}{\zeta-a}\bigg)(\zeta-a)^{Nc}q_n(\zeta)\Bigg)^*.\quad\quad  
\end{split}
\end{align} 
Then we have
\begin{align}\label{eq in}
\begin{split}
{\rm I}_n
&= {\ee^{ -Nz\overline{\zeta} }} \frac{N^{Nc+n+1}}{\pi\Gamma(Nc+n+1)} \ee^{\ii\frac{8}{3}\xi(\beta)^3}(z-a)^{Nc}\ee^{Ntg(z)}\Big(\ee^{\ii\frac{8}{3}\xi(\beta)^3}(\zeta-a)^{Nc}\ee^{Ntg(\zeta)}\Big)^*\\
&\quad\times \frac{N\alpha_n}{1+\beta_n\gamma_n}\frac{ \ee^{\frac{N}{2}\phi(z)}[H\Psi]_{11}^z+{\cal O}(N^{-1})}{\ee^{tN\ell+\ii\frac{16}{3}\xi(\beta)^3}}\Bigg((c+t)\bigg(\frac{1+\beta_n\gamma_n}{\zeta-a}-\frac{\beta_n\gamma_n}{\zeta}\bigg)\Big(\ee^{\frac{N}{2}\phi(\zeta)}[H\Psi]_{21}^\zeta+{\cal O}\big(N^{-1}\big)\Big)\\
&\quad\quad-\bigg(\frac{(c+t)\beta_n\alpha_n}{\zeta}+\frac{ab_n-(c+t)\beta_n\alpha_n}{\zeta-a}\bigg)\frac{\ee^{\frac{N}{2}\phi(\zeta)}[H\Psi]_{11}^\zeta+{\cal O}(N^{-1})}{\ee^{tN\ell+\ii\frac{16}{3}\xi(\beta)^3}}\Bigg)^*\\
&=Q_n  \frac{N\alpha_n}{1+\beta_n\gamma_n}\frac{ \ee^{\frac{N}{2}\phi(z)}[H\Psi]_{11}^z+{\cal O}(N^{-1})}{\ee^{tN\ell+\ii\frac{16}{3}\xi(\beta)^3}\ee^{\frac{N}{2}\phi(z)}}\Bigg((c+t)\bigg(\frac{1+\beta_n\gamma_n}{\zeta-a}-\frac{\beta_n\gamma_n}{\zeta}\bigg)\frac{\ee^{\frac{N}{2}\phi(\zeta)}[H\Psi]_{21}^\zeta+{\cal O}\big(N^{-1}\big)}{\ee^{\frac{N}{2}\phi(\zeta)}}\\
&\quad\quad-\bigg(\frac{(c+t)\beta_n\alpha_n}{\zeta}+\frac{ab_n-(c+t)\beta_n\alpha_n}{\zeta-a}\bigg)\frac{\ee^{\frac{N}{2}\phi(\zeta)}[H\Psi]_{11}^\zeta+{\cal O}(N^{-1})}{\ee^{tN\ell+\ii\frac{16}{3}\xi(\beta)^3}\ee^{\frac{N}{2}\phi(z)}}\Bigg)^*,
\end{split}
\end{align} 
where we defined the common factor
\begin{equation}\label{eq common factor Q}
Q_n ={e^{ -Nz\overline{\zeta}}  }\frac{N^{Nc+n+1}}{\pi\Gamma(Nc+n+1)}\ee^{\ii\frac{8}{3}\xi(\beta)^3}(z-a)^{{Nc}}\ee^{\frac{N(2tg(z)+\phi(z))}{2}}\Big(\ee^{\ii\frac{8}{3}\xi(\beta)^3}(\zeta-a)^{{Nc}}\ee^{\frac{N(2tg(\zeta)+\phi(\zeta))}{2}}\Big)^*,
\end{equation}
that will also appear in the expression for ${\rm II}_n$ in \eqref{eq iin} below.
Here and below the superscript $*$ means the complex conjugation but $b_c^*$ does not represent the complex conjugation of $b_c$, which is defined in \eqref{def bcstar}.  The 1st equality is obtained by \eqref{pn strong}. The error bound is uniform over $D_c$.

Similarly, using \eqref{eq pn1 and pn11}, \eqref{hn2}, \eqref{hn3}, and \eqref{psi differ2}, we rewrite ${\rm II}_n$ from \eqref{eq IIN} as 
\begin{align}\label{eq iin middle}
\begin{split}
{\rm II}_n&={\ee^{ -Nz\overline{\zeta} }} \frac{p_{n+1}(a)}{p_n(a)} \frac{Nh_n/h_{n-1} }{ \frac{n+Nc+1}{N} h_n-h_{n+1}    }\overline{ \psi_{n-1}(\zeta) }    \Big( \psi_{n+1}(z)-z \psi_{n}(z) \Big)\\
&={\ee^{ -Nz\overline{\zeta} }}\frac{N^{Nc+n+2}}{\Gamma(Nc+n+2)}\bigg(a+\frac{b_n \left(1+\beta_n\gamma_n\right)}{\alpha_n
\beta_n}-b_n\frac{q_n(a)}{p_n(a)}\bigg)\frac{(Nc+n)\beta_n\gamma_n}{\pi \big(\beta_{n+1}-\beta_n\big)}\\
&\quad\times \bigg(\frac{b_n \left(1+\beta_n\gamma_n\right)}{\alpha_n
\beta_n}\psi_{n}(z)-b_n(z-a)^{Nc}q_n(z)\bigg)\bigg(\frac{1}{\zeta}\Big(\psi_{n}(\zeta)-\frac{\alpha_n}{\gamma_n}(\zeta-a)^{Nc}q_{n}(\zeta)\Big)\bigg)^*.
\end{split}
\end{align}
Then we have
\begin{align}\label{eq iin}
\begin{split}
{\rm II}_n
&={\ee^{ -Nz\overline{\zeta} }} \frac{N^{Nc+n+1}}{\pi\Gamma(Nc+n+1)} \ee^{\ii\frac{8}{3}\xi(\beta)^3}(z-a)^{Nc}\ee^{Ntg(z)}\Big(\ee^{\ii\frac{8}{3}\xi(\beta)^3}(\zeta-a)^{Nc}\ee^{Ntg(\zeta)}\Big)^*\\
&\quad\times \frac{N(Nc+n)}{Nc+n+1}\bigg(a+\frac{b_n \left(1+\beta_n\gamma_n\right)}{\alpha_n
\beta_n}-b_n\frac{q_n(a)}{p_n(a)}\bigg)\frac{\beta_n\gamma_n}{ \beta_{n+1}-\beta_n} b_n\Bigg(\frac{ \ee^{\frac{N}{2}\phi(z)}[H\Psi]_{21}^z+{\cal O}(N^{-1})}{\alpha_n
\beta_n(1+\beta_n\gamma_n)^{-1}}\\
&\quad +\frac{ \ee^{\frac{N}{2}\phi(z)}[H\Psi]_{11}^z+{\cal O}(N^{-1})}{\ee^{tN\ell+\ii\frac{16}{3}\xi(\beta)^3}}\Bigg) \Bigg(\frac{ \ee^{\frac{N}{2}\phi(\zeta)}[H\Psi]_{21}^\zeta+{\cal O}(N^{-1})}{\zeta}+\frac{\alpha_n}{\gamma_n}\frac{ \ee^{\frac{N}{2}\phi(\zeta)}[H\Psi]_{11}^\zeta+{\cal O}(N^{-1})}{\zeta\ee^{tN\ell+\ii\frac{16}{3}\xi(\beta)^3}}\Bigg)^*\\
&=Q_n \frac{N(Nc+n)}{Nc+n+1}\bigg(a+\frac{b_n \left(1+\beta_n\gamma_n\right)}{\alpha_n
\beta_n}-b_n\frac{q_n(a)}{p_n(a)}\bigg)\frac{\beta_n\gamma_n}{ \beta_{n+1}-\beta_n} b_n\Bigg(\frac{ \ee^{\frac{N}{2}\phi(z)}[H\Psi]_{21}^z+{\cal O}(N^{-1})}{\ee^{\frac{N}{2}\phi(z)}\alpha_n
\beta_n(1+\beta_n\gamma_n)^{-1}}\\
&\quad +\frac{ \ee^{\frac{N}{2}\phi(z)}[H\Psi]_{11}^z+{\cal O}(N^{-1})}{\ee^{tN\ell+\ii\frac{16}{3}\xi(\beta)^3}\ee^{\frac{N}{2}\phi(z)}}\Bigg) \Bigg(\frac{ \ee^{\frac{N}{2}\phi(\zeta)}[H\Psi]_{21}^\zeta+{\cal O}(N^{-1})}{\zeta\ee^{\frac{N}{2}\phi(\zeta)}}+\frac{\alpha_n}{\gamma_n}\frac{ \ee^{\frac{N}{2}\phi(\zeta)}[H\Psi]_{11}^\zeta+{\cal O}(N^{-1})}{\zeta\ee^{tN\ell+\ii\frac{16}{3}\xi(\beta)^3}\ee^{\frac{N}{2}\phi(\zeta)}}\Bigg)^*.
\end{split}
\end{align}
Here the 1st equality is obtained by \eqref{pn strong}. The error bound is again uniform over $D_c$.

We recall that
\begin{equation}\label{eq n and N}
n=Nt,\quad  t=t_c+\frac{2b_cs}{\gamma_cN^{2/3}}.   
\end{equation}

We define the local scaling coordinates $u$ and $v$ by
\begin{equation}\label{def z and zeta}
z=b_c+\frac{\ii \gamma_c u}{N^{1/3}},\quad \zeta=b_c+\frac{\ii \gamma_c v}{N^{1/3}}.    
\end{equation}

We now compute ${\rm I}_n$ and ${\rm II}_n$. Using the expression of $\beta$ in \eqref{def b and beta} with \eqref{eq n and N}  we have 
\begin{equation}\label{beta expansion}
  \beta=b_c-\frac{\gamma_c\sqrt s}{2 N^{1/3}}+{\cal O}\Big(\frac{1}{N^{2/3}}\Big).
\end{equation}
Moreover, using $\ell$ in \eqref{def ell} with \eqref{beta expansion} we have 
\begin{equation} \label{ell expansion}
  t\ell =-ab_c+b_c^2\log (b_c)-\frac{c\log (c)}{2}+\frac{2b_c s\log(b_c)}{\gamma_cN^{2/3}}-\frac{2s^{3/2}}{3N}+{\cal O}\Big(\frac{1}{N^{4/3}}\Big). 
\end{equation}

Using the expressions for $z$ and $\zeta$ in \eqref{def z and zeta}, the expansion for $\ell$ in \eqref{ell expansion}, and \eqref{eq n and N} we have \begin{equation}\label{common1}
\begin{aligned}
\ee^{ -Nz\overline{\zeta}} &=\exp\big( -Nb_c^2-\ii N^{2/3} b_c\gamma_c(u-\bar v) -N^{1/3}\gamma_c^2u\bar v\big) ,\\
\frac{N^{Nc+n+1}}{\pi\Gamma(Nc+n+1)}&=\exp\big(Nb_c^2\big)b_c^{-2Nb_c^2-N^{1/3}\frac{4b_cs}{\gamma_c}}\frac{\sqrt N}{\sqrt 2 \pi^{3/2}b_c}\bigg(1+{\cal O}\Big(\frac{1}{N^{1/3}}\Big)\bigg),\\
\ee^{\frac{Naz+tN\ell}{2}}&=\exp\bigg[ \frac{N^{2/3}}{2}a\gamma_c \ii u-\frac{s^{3/2}}{3}\bigg] b_c^{\frac{N}{2}b_c^2+N^{1/3}\frac{b_cs}{\gamma_c}}c^{-\frac{Nc}{4}}\bigg(1+{\cal O}\Big(\frac{1}{N^{1/3}}\Big)\bigg),\\
z^{\frac{n+Nc}{2}}(z-a)^{\frac{Nc}{2}}&=\exp\bigg[\frac{N^{2/3}}{2}(a+2\sqrt c)\gamma_c \ii u+\frac{N^{1/3}}{2} \gamma_c^2u^2+\ii s u-\frac{\ii (b_c+\sqrt c) \gamma_c^3u^3}{6\sqrt c b_c}\bigg]\\
&\qquad \times b_c^{\frac{N}{2}b_c^2+N^{1/3}\frac{b_cs}{\gamma_c}}c^{\frac{Nc}{4}}\bigg(1+{\cal O}\Big(\frac{1}{N^{1/3}}\Big)\bigg).  
\end{aligned}    
\end{equation}
Here the 2nd identity is obtained by the Stirling approximation formula \begin{equation}\label{eq gamma Stirling}
 \Gamma(z)=\frac{z^z}{\ee^z}\sqrt{\frac{2\pi}{z}}\bigg(1+{\cal O}\Big(\frac{1}{z}\Big)\bigg).   
\end{equation}

Moreover, for $z\in D_c\cap {\rm Ext}{\cal B}$, using the definition of $\xi(z)$ in \eqref{def xi}, the expansion of $\xi(z)$ in \eqref{zeta expanssion} with \eqref{eq-gamma123} and \eqref{def b star}, $\hat s$ in \eqref{def zetabeta}, and $s$  in \eqref{def gammacl}, we have
\begin{equation}\label{eqxi3}
\begin{aligned}
  \exp\Big(\frac{N}{2}\phi(z)\Big)&=\exp\Big[\ii\Big(\frac{4}{3}\xi(z)^3+\hat s\xi(z)+\frac{8}{3}\xi(\beta)^3\Big)\Big]\\
  &=\exp\Big[\ii\Big(\frac{4}{3}u^3+su\Big)\Big]\exp\Big(\ii\frac{8}{3}\xi(\beta)^3\Big)\Big(1-\frac{(4u^2+s)(2a\sqrt c b_cs-u^2t_c\gamma_c^3)}{4a\sqrt c b_c\gamma_c^2N^{1/3}}+{\cal O}\big(\frac{1}{N^{2/3}}\big)\Big).      
\end{aligned}
\end{equation}
Using the definitions of $g(z)$ in \eqref{def gfunction} and $\phi(z)$ in \eqref{def phi}, we also have
\begin{equation}\label{eq expansion of some factor}
(z-a)^{Nc}\ee^{\frac{N(2tg(z)+\phi(z))}{2}}=z^{\frac{n+Nc}{2}}(z-a)^{\frac{Nc}{2}}\ee^{\frac{Naz+tN\ell}{2}}.
\end{equation}

The common factor  of ${\rm I}_n$ and ${\rm II}_n$  from \eqref{eq common factor Q} is given by
\begin{equation}\label{eq common factor}
 \begin{aligned}
Q_n&={e^{ -Nz\overline{\zeta}}  }\frac{N^{Nc+n+1}\ee^{\ii\frac{8}{3}\xi(\beta)^3}}{\pi\Gamma(Nc+n+1)}z^{\frac{n+Nc}{2}}(z-a)^{\frac{Nc}{2}}\ee^{\frac{Naz+tN\ell}{2}}\Big(\zeta^{\frac{n+Nc}{2}}(\zeta-a)^{\frac{Nc}{2}}\ee^{\ii\frac{8}{3}\xi(\beta)^3}\ee^{\frac{Na\zeta+tN\ell}{2}}\Big)^*\\
& =\exp\bigg(\frac{N^{1/3}\gamma_c^2(u-\bar v)^2}{2}\bigg)\exp \bigg( \ii s (u-\bar v)-\frac{\ii(b_c+\sqrt c)}{6\sqrt c b_c}\gamma_c^3(u^3-\bar v^3)\bigg)  \\
&\quad \times \exp\bigg(\frac{16}{3}{\rm Re}(\ii\xi(\beta)^3)-\frac{2{\rm Re}(s^{3/2})}{3}\bigg)\frac{\sqrt N}{\sqrt 2 \pi^{3/2}b_c }\bigg(1+{\cal O}\Big(\frac{1}{N^{1/3}}\Big)\bigg)\\
&=\exp\bigg(\frac{N^{1/3}\gamma_c^2(u-\bar v)^2}{2}\bigg)\exp \bigg(\ii s (u-\bar v)-\frac{4\ii(b_c+\sqrt c)}{3a}(u^3-\bar v^3)\bigg)\frac{\sqrt N}{\sqrt 2 \pi^{3/2}b_c }\bigg(1+{\cal O}\Big(\frac{1}{N^{1/3}}\Big)\bigg),
\end{aligned}   
\end{equation}  where 
 the 1st identity is obtained by \eqref{eq expansion of some factor}, 
the 2nd equality is obtained by \eqref{common1}, the last equality is obtained by the expansion for $\xi(\beta)$ in \eqref{zeta beta} with  $\hat s = s(1+{\cal O}(N^{-2/3}))$ and the definition of $\gamma_c$ in \eqref{def gammacl}.

Now we compute the remaining part  $({\rm I}_n-{\rm II}_n)/Q_n$ of the difference between ${\rm I}_n$ and ${\rm II}_n$, without the common factor $Q_n$.

By \eqref{def Yz}, the relation between $Y_n$ and $\widetilde{Y}_n$ in \eqref{transform y to ytilde}, \eqref{Ytilde 0} and \eqref{pn inside b} in Proposition \ref{thm ortho1},  we have
\begin{equation}\label{eq alpha}
\begin{aligned}
\alpha_n&=p_n(0)=-\ee^{\ii\frac{16}{3}\xi(\beta)^3}\ee^{tN\ell}\bigg(\frac{q}{2r_1N^{1/3}b_c^*}+\frac{r_2p_{21}}{r_1^3N^{2/3}b_c^*}+\frac{p_{21}}{r_1^2N^{2/3}(b_c^*)^2}+{\cal O}\Big(\frac{1}{N}\Big)\bigg),\\
\beta_n&=[Y_n(0)]_{12}=1-\frac{r}{2r_1N^{1/3}b_c^*}-\frac{p_{22}}{r_1^2N^{2/3}(b_c^*)^2}+\frac{C_\beta}{N}+{\cal O}\Big(\frac{1}{N^{4/3}}\Big),\\
\gamma_n&=q_n(0)=-\bigg(1+\frac{r}{2r_1N^{1/3}b_c^*}-\frac{p_{11}}{r_1^2N^{2/3}(b_c^*)^2}+\frac{C_\gamma}{N}+{\cal O}\Big(\frac{1}{N^{4/3}}\Big)\bigg).\end{aligned}    
\end{equation} Here we will not need the explicit values of the coefficients $C_\beta$ and $C_\gamma$ in $\beta_n$ and $\gamma_n$ respectively. The numbers $C_\beta$ and $C_\gamma$ are defined by being the expansion coefficients in the $1/N$-expansion, which will be dropped out of the computation of the correlation kernel at a later stage. Moreover, in this and in the following we note the implicit dependence on $\hat s$, i.e., $q=q(\hat s)$, $r=r(\hat s)$, $p_{11}=p_{11}(\hat s)$ and $q'=q'(\hat s)$.

By \eqref{Ytilde inf}, the relation between $Y_n$ and $\widetilde{Y}_n$ in \eqref{transform y to ytilde},  and \eqref{pn outside b} in  Proposition \ref{thm ortho1}, we have
\begin{equation}\label{eq bn}
b_n=\lim_{z\to\infty } z^{n+1}[Y_n(z)]_{12}=-\ee^{\ii\frac{16}{3}\xi(\beta)^3}\ee^{tN\ell}\bigg(\frac{q}{2r_1N^{1/3}}+\frac{r_2p_{21}}{r_1^3N^{2/3}}+{\cal O}\Big(\frac{1}{N}\Big)\bigg).    \end{equation}

Using this, the expression of $b_c^*$ in \eqref{def b star}, the expressions of $\alpha_n$, $\beta_n$ and $\gamma_n$ in \eqref{eq alpha}, the relations of $p_{11}$, $p_{21}$ and $p_{22}$ in \eqref{relation up} and the expansions of $r_1$ and $r_2$ in \eqref{eq-gamma123} we have the following identities,
\begin{equation}\label{eq alphabngamma}
\begin{aligned}
\alpha_n&=\ee^{\ii\frac{16}{3}\xi(\beta)^3}\ee^{tN\ell}\bigg(\frac{q\gamma_c}{2b_cN^{1/3}}+\frac{a\gamma_c^2(rq+q')}{8b_c^2\sqrt c N^{2/3}}+{\cal O}\Big(\frac{1}{N}\Big)\bigg),\\
b_n&=\ee^{\ii\frac{16}{3}\xi(\beta)^3}\ee^{tN\ell}\bigg(\frac{q\gamma_c}{2N^{1/3}}+\frac{(b_c+\sqrt c )\gamma_c^2(rq+q')}{8b_c\sqrt c N^{2/3}}+{\cal O}\Big(\frac{1}{N}\Big)\bigg),\\
\beta_n\gamma_n&=-1+\frac{q^2\gamma_c^2}{4b_c^2N^{2/3}}-\frac{C_\beta+C_\gamma}{N}+{\cal O}\Big(\frac{1}{N^{4/3}}\Big),\\
\alpha_n\beta_n&=\ee^{\ii\frac{16}{3}\xi(\beta)^3}\ee^{tN\ell}\bigg(\frac{q\gamma_c}{2b_cN^{1/3}}+\frac{(b_c+\sqrt c )\gamma_c^2rq+a\gamma_c^2q'}{8b_c^2\sqrt c N^{2/3}}+{\cal O}\Big(\frac{1}{N}\Big)\bigg).\end{aligned}    \end{equation}
It follows that
\begin{equation}\label{I1}
\frac{N\alpha_n}{ 1+\beta_n\gamma_n}\frac{1}{\ee^{tN\ell+\ii\frac{16}{3}\xi(\beta)^3}}=\frac{2b_cN^{4/3}}{q\gamma_c}+\frac{a(rq+q')N}{2\sqrt c q^2}+{\frac{16b_c^3(C_\beta+C_\gamma)N}{2\gamma_c^3 q^3}}+{\cal O}\big(N^{2/3}\big).    \end{equation}

Similarly, by \eqref{pn inside b} in Proposition \ref{thm ortho1}, the expression of $b_c^*$ in \eqref{def b star}, the relations of $p_{11}$ and $p_{21}$ in \eqref{relation up} and the expansions of $r_1$ and $r_2$ in \eqref{eq-gamma123} we have
\begin{equation}\label{eq pqratio}
    \frac{q_n(a)}{p_n(a)}=\frac{1}{\ee^{\ii\frac{16}{3}\xi(\beta)^3}\ee^{tN\ell}}\bigg(-\frac{2\sqrt c N^{1/3}}{q\gamma_c}+\frac{(b_c+\sqrt c )\gamma_c^2rq-a\gamma_c^2q'}{2b_cq^2 }+{\cal O}\Big(\frac{1}{N^{1/3}}\Big)\bigg).
\end{equation}

Using \eqref{eq h11} and  \eqref{eq-gamma123}, and the boundedness of $\exp\big(\frac{N}{2}\phi(z)\big)$ from \eqref{eqxi3}, we have
\begin{equation}\label{I2} \frac{\ee^{\frac{N}{2}\phi(z)}[H\Psi]_{11}^z+{\cal O}\big(N^{-1}\big)}{\ee^{\frac{N}{2}\phi(z)}}=\Psi_{11}(\xi(z);\hat s)-\frac{(\sqrt c+ b_c)\gamma_c}{8\sqrt c b_cN^{1/3}}\big(r\Psi_{11}(\xi(z);\hat s)+q\Psi_{21}(\xi(z);\hat s)\big) +{\cal O}\Big(\frac{1}{N^{2/3}}\Big).\end{equation}
Similarly, we have
\begin{equation}\label{I21} \frac{\ee^{\frac{N}{2}\phi(z)}[H\Psi]_{21}^z+{\cal O}\big(N^{-1}\big)}{\ee^{\frac{N}{2}\phi(z)}}=\Psi_{21}(\xi(z);\hat s)+\frac{(\sqrt c+ b_c)\gamma_c}{8\sqrt c b_cN^{1/3}}\big(r\Psi_{21}(\xi(z);\hat s)+q\Psi_{11}(\xi(z);\hat s)\big) +{\cal O}\Big(\frac{1}{N^{2/3}}\Big).\end{equation}

Using $\beta_n\gamma_n$, $\beta_n\alpha_n$ and $b_n$ in \eqref{eq alphabngamma}, \eqref{I2} and \eqref{I21} with the expression for $z$ in \eqref{def z and zeta} we have
\begin{equation}\label{I3}
\begin{aligned}
 &\quad (c+t)\bigg(\frac{1+\beta_n\gamma_n}{z-a}-\frac{\beta_n\gamma_n}{z}\bigg)\frac{\ee^{\frac{N}{2}\phi(z)}[H\Psi]_{21}^z+{\cal O}(N^{-1})}{\ee^{\frac{N}{2}\phi(z)}}\\
 &\quad -\bigg(\frac{(c+t)\beta_n\alpha_n}{z}+\frac{ab_n-(c+t)\beta_n\alpha_n}{z-a}\bigg)\frac{ \ee^{\frac{N}{2}\phi(z)}[H\Psi]_{11}^z+{\cal O}(N^{-1})}{\ee^{tN\ell+\ii\frac{16}{3}\xi(\beta)^3}\ee^{\frac{N}{2}\phi(z)}}\\
&=b_c\Psi_{21}(\xi(z);\hat s)+\frac{(\sqrt c+b_c)(q\Psi_{11}(\xi(z);\hat s)+r\Psi_{21}(\xi(z);\hat s))-8\ii \sqrt c u\Psi_{21}(\xi(z);\hat s)}{8\sqrt c\gamma_c^{-1} N^{1/3}}+{\cal O}\Big(\frac{1}{N^{2/3}}\Big).
\end{aligned}\end{equation}

Combing \eqref{I1}, \eqref{I2} and \eqref{I3} we have
\begin{equation}\label{differ from I}
\begin{aligned}
&\frac{{\rm I}_n}{Q_n}
=\frac{N\alpha_n}{1+\beta_n\gamma_n}\frac{ \ee^{\frac{N}{2}\phi(z)}[H\Psi]_{11}^z+{\cal O}(N^{-1})}{\ee^{tN\ell+\ii\frac{16}{3}\xi(\beta)^3}\ee^{\frac{N}{2}\phi(z)}}\Bigg((c+t)\bigg(\frac{1+\beta_n\gamma_n}{\zeta-a}-\frac{\beta_n\gamma_n}{\zeta}\bigg)\frac{\ee^{\frac{N}{2}\phi(\zeta)}[H\Psi]_{21}^\zeta+{\cal O}\big(N^{-1}\big)}{\ee^{\frac{N}{2}\phi(\zeta)}}\\
&\quad\quad-\bigg(\frac{(c+t)\beta_n\alpha_n}{\zeta}+\frac{ab_n-(c+t)\beta_n\alpha_n}{\zeta-a}\bigg)\frac{\ee^{\frac{N}{2}\phi(\zeta)}[H\Psi]_{11}^\zeta+{\cal O}(N^{-1})}{\ee^{tN\ell+\ii\frac{16}{3}\xi(\beta)^3}\ee^{\frac{N}{2}\phi(z)}}\Bigg)^*\\
&=\frac{2b_c^2\Psi_{11}(\xi(z);\hat s)\overline{\Psi_{21}(\xi(\zeta);\hat s)}N^{4/3}}{q \gamma_c}+\frac{b_c\Psi_{11}(\xi(z);\hat s)\overline{\Psi_{21}(\xi(\zeta);\hat s)}\big(a\gamma_c(qr+q')+4\ii\sqrt c\gamma_cq\bar v\big)N}{2\sqrt c q^2\gamma_c}\\
&\quad-\frac{8b_c^4(C_\beta+C_\gamma)\Psi_{11}(\xi(z);\hat s)\overline{\Psi_{21}(\xi(\zeta);\hat s)}N}{q^3\gamma_c^3}\\
&\quad +\frac{b_c(\sqrt c+b_c)\big(\Psi_{11}(\xi(z);\hat s)\overline{\Psi_{11}(\xi(\zeta);\hat s)}-\Psi_{21}(\xi(z);\hat s)\overline{\Psi_{21}(\xi(\zeta);\hat s)}\big)N}{4\sqrt c }+{\cal O}(N^{2/3}).
\end{aligned}    
\end{equation}

Using $\beta_n\gamma_n$, $\beta_n\alpha_n$ and $b_n$ in \eqref{eq alphabngamma}, $\beta_{n+1}$ in \eqref{recurence1},  and the ratio of $q_n(a)$ and $p_n(a)$ in \eqref{eq pqratio}  we have
\begin{equation}\label{II1}
\begin{aligned}
&\quad \frac{N(Nc+n)}{Nc+n+1}\bigg(a+\frac{b_n \left(1+\beta_n\gamma_n\right)}{\alpha_n
\beta_n}-b_n\frac{q_n(a)}{p_n(a)}\bigg)\frac{\beta_n\gamma_n}{ \beta_{n+1}-\beta_n}\\
&=\frac{4b_c^3N^{5/3}}{q^2\gamma_c^2}-2b_c^{2}\frac{qr+q'}{q^3\gamma_c}N^{4/3}+{16b_c^{5}\frac{C_\beta+C_\gamma}{q^4\gamma_c^4}N^{4/3}}+{\cal O}(N).
\end{aligned}\end{equation}

Using $b_n$, $\beta_n\alpha_n$ and $\beta_n\gamma_n$ in \eqref{eq alphabngamma}, \eqref{I2} and \eqref{I21} we have
\begin{equation}\label{II2}
\begin{aligned}& \quad b_n \bigg(\frac{\ee^{\frac{N}{2}\phi(z)}[H\Psi]_{21}^z+{\cal O}(N^{-1})}{\ee^{\frac{N}{2}\phi(z)}\alpha_n
\beta_n(1+\beta_n\gamma_n)^{-1}}+\frac{ \ee^{\frac{N}{2}\phi(z)}[H\Psi]_{11}^z+{\cal O}(N^{-1})}{\ee^{tN\ell+\ii\frac{16}{3}\xi(\beta)^3}\ee^{\frac{N}{2}\phi(z)}}\bigg)\\
&=\frac{q\gamma_c\Psi_{11}(\xi(z);\hat s)}{2N^{1/3}}+\frac{(\sqrt c+b_c)\gamma_c^2(qr+2q')\Psi_{11}(\xi(z);\hat s)-(b_c-3\sqrt c)q^2\gamma_c^2\Psi_{21}(\xi(z);\hat s)}{16\sqrt c b_cN^{2/3}}+{\cal O}\Big(\frac{1}{N}\Big).
\end{aligned}\end{equation}

Using $\beta_n\alpha_n$ and $\beta_n\gamma_n$ in \eqref{eq alphabngamma}, \eqref{I2} and \eqref{I21} with the expression for $z$ in \eqref{def z and zeta} we have
\begin{equation}\label{II3}\begin{aligned}
 &\quad \frac{ \ee^{\frac{N}{2}\phi(z)}[H\Psi]_{21}^z+{\cal O}(N^{-1})}{z\ee^{\frac{N}{2}\phi(z)}}+\frac{\alpha_n}{\gamma_n}\frac{ \ee^{\frac{N}{2}\phi(z)}[H\Psi]_{11}^z+{\cal O}(N^{-1})}{z\ee^{tN\ell+\ii\frac{16}{3}\xi(\beta)^3}\ee^{\frac{N}{2}\phi(z)}}\\
&=\frac{\Psi_{21}(\xi(z);\hat s)}{b_c}+\frac{(b_c-3\sqrt c)q\Psi_{11}(\xi(z);\hat s)-\big(8\ii \sqrt c u-(\sqrt c+b_c)r\big)\Psi_{21}(\xi(z);\hat s)}{8\sqrt c b_c^2\gamma_c^{-1}N^{1/3}}+{\cal O}\Big(\frac{1}{N^{2/3}}\Big).
\end{aligned}\end{equation}

Combing \eqref{II1}, \eqref{II2} and \eqref{II3} we have
\begin{equation}\label{differ from II}
\begin{aligned}
&\frac{{\rm II}_n}{Q_n}
=\frac{N(Nc+n)}{Nc+n+1}\bigg(a+\frac{b_n \left(1+\beta_n\gamma_n\right)}{\alpha_n
\beta_n}-b_n\frac{q_n(a)}{p_n(a)}\bigg)\frac{\beta_n\gamma_n}{ \beta_{n+1}-\beta_n} b_n\Bigg(\frac{ \ee^{\frac{N}{2}\phi(z)}[H\Psi]_{21}^z+{\cal O}(N^{-1})}{\ee^{\frac{N}{2}\phi(z)}\alpha_n
\beta_n(1+\beta_n\gamma_n)^{-1}}\\
&\quad +\frac{ \ee^{\frac{N}{2}\phi(z)}[H\Psi]_{11}^z+{\cal O}(N^{-1})}{\ee^{tN\ell+\ii\frac{16}{3}\xi(\beta)^3}\ee^{\frac{N}{2}\phi(z)}}\Bigg) \Bigg(\frac{ \ee^{\frac{N}{2}\phi(\zeta)}[H\Psi]_{21}^\zeta+{\cal O}(N^{-1})}{\zeta\ee^{\frac{N}{2}\phi(\zeta)}}+\frac{\alpha_n}{\gamma_n}\frac{ \ee^{\frac{N}{2}\phi(\zeta)}[H\Psi]_{11}^\zeta+{\cal O}(N^{-1})}{\zeta\ee^{tN\ell+\ii\frac{16}{3}\xi(\beta)^3}\ee^{\frac{N}{2}\phi(\zeta)}}\Bigg)^*\\
&=\frac{2b_c^2\Psi_{11}(\xi(z);\hat s)\overline{\Psi_{21}(\xi(\zeta);\hat s)}N^{4/3}}{q \gamma_c}+\frac{b_c\Psi_{11}(\xi(z);\hat s)\overline{\Psi_{21}(\xi(\zeta);\hat s)}\big(a\gamma_c(qr+q')+4\ii\sqrt c\gamma_cq\bar v\big)N}{2\sqrt c q^2\gamma_c}\\
&\qquad -\frac{8b_c^4(C_\beta+C_\gamma)\Psi_{11}(\xi(z);\hat s)\overline{\Psi_{21}(\xi(\zeta);\hat s)}N}{q^3\gamma_c^3}\\
&\qquad+\frac{b_c(b_c-3\sqrt c)\big(\Psi_{11}(\xi(z);\hat s)\overline{\Psi_{11}(\xi(\zeta);\hat s)}-\Psi_{21}(\xi(z);\hat s)\overline{\Psi_{21}(\xi(\zeta);\hat s)}\big)N}{4\sqrt c }+{\cal O}(N^{2/3}).
\end{aligned}    
\end{equation}

Moreover, using the definition $\xi(z)$ in \eqref{zeta expanssion}  we have 
\begin{equation}\label{eq xiz app}
  \xi(z)=u+\frac{\ii}{N^{1/3}}\Big(\frac{s}{2\gamma_c^2}-\frac{t_c\gamma_c}{4a\sqrt c b_c}u^2\Big)+{\cal O}\Big(\frac{1}{N^{2/3}}\Big).   
\end{equation}
For $j,k=1,2$, using the expression of $\xi$ in  \eqref{eq xiz app}, and the definitions of $\hat s$ and $s$ in \eqref{def zetabeta} and \eqref{def gammacl} we have
\begin{equation}\label{eq psi approx}
\Psi_{jk}(\xi(z);\hat s)=\Psi_{jk}(u;s)\big(1+{\cal O}(N^{-1/3})\big),\quad 
\Psi_{jk}(\xi(\zeta);\hat s)=\Psi_{jk}(v;s)\big(1+{\cal O}(N^{-1/3})\big).
\end{equation}

Comparing \eqref{differ from I} and \eqref{differ from II}, one can see that the difference comes from the 4th term. Hence, the remaining part of ${\rm I}_n-{\rm II}_n$ is given by
\begin{equation}\label{eq remaining}
\begin{aligned}
\frac{{\rm I}_n}{Q_n}-\frac{{\rm II}_n}{Q_n}
&=b_cN\Big(\Psi_{11}(\xi(z);\hat s)\overline{\Psi_{11}(\xi(\zeta);\hat s)}-\Psi_{21}(\xi(z);\hat s)\overline{\Psi_{21}(\xi(\zeta);\hat s)}+{\cal O}\Big(\frac{1}{N^{1/3}}\Big)\Big)\\
&=b_cN\Big(\Psi_{11}(u;s)\overline{\Psi_{11}(v;s)}-\Psi_{21}(u;s)\overline{\Psi_{21}(v;s)}+{\cal O}\Big(\frac{1}{N^{1/3}}\Big)\Big).
\end{aligned}    
\end{equation}
Here the last equality is obtained by the facts in \eqref{eq psi approx}.

Using the common factor of ${\rm I}_n-{\rm II}_n$ in \eqref{eq common factor} and the remaining part of ${\rm I}_n-{\rm II}_n$ in \eqref{eq remaining}, we have
\begin{equation} \label{eq i and ii}
\begin{aligned}
\frac{1}{N^{5/6}}\overline{\partial}_{\zeta} {\cal K}_{n}(z, \zeta)=\frac{{\rm I}_n-{\rm II}_n}{N^{5/6}}
&=\exp\Big(\frac{N^{1/3}\gamma_c^2(u-\bar v)^2}{2}\Big)\exp \Big(\ii s (u-\bar v)-\frac{4\ii(b_c+\sqrt c)}{3a}(u^3-\bar v^3)\Big)\\
&\quad \times \frac{N^{2/3}}{ \sqrt 2 \pi^{3/2}}\Big(\Psi_{11}(u;s)\overline{\Psi_{11}(v;s)}-\Psi_{21}(u;s)\overline{\Psi_{21}(v;s)}+{\cal O}\Big(\frac{1}{N^{1/3}}\Big)\Big).
\end{aligned}
\end{equation}
 This ends the proof.   
\end{proof}

\section{Integrating the C-D identity: Proof of Theorem \ref{main theorem}}\label{section 3}
In this section, we prove Theorem \ref{main theorem} by using Lemma \ref{preker} and Theorem \ref{Thm 41}.

Let us write the real and imaginary components of the coordinates explicitly by
\begin{equation}\label{def uv}
 u:=u_R-\frac{\ii u_I}{\gamma_c N^{1/6}},\quad v:=v_R-\frac{\ii v_I}{\gamma_c N^{1/6}}.   
\end{equation}
Note that the imaginary component is scaled differently in $N$.

Using \eqref{def z and zeta} and \eqref{def uv} we have
\begin{equation}
    \frac{\ii \gamma_c}{N^{1/3}}\partial_{\zeta}=\partial_v = \frac{1}{2}(\partial_{v_R} + \ii \gamma_c N^{1/6} \partial_{v_I}),\quad 
    -\frac{\ii \gamma_c}{N^{1/3}}\overline\partial_{\zeta}=\overline\partial_v = \frac{1}{2}(\partial_{v_R} - \ii \gamma_c N^{1/6} \partial_{v_I}).
\end{equation}
It follows that
\begin{equation}\label{eq dvi}
    \partial_{v_I}= \frac{1}{N^{1/2}}(\partial_\zeta+\overline\partial_{\zeta}).
\end{equation}

Then \eqref{eq lem26} in Lemma \ref{preker}  can be written as
\begin{equation} \label{eq i and ii2}
\begin{aligned}
\frac{1}{N^{5/6}}\overline{\partial}_{\zeta} {\cal K}_{n}(z, \zeta)
&=\exp\bigg(\frac{N^{1/3}\gamma_c^2(u_R-v_R)^2}{2}-\ii (u_R-v_R)(u_I+v_I)\gamma_c N^{1/6}\bigg)\\
&\quad \times \exp \bigg(\ii s (u_R-v_R)-\frac{4\ii(b_c+\sqrt c)}{3a}(u_R^3-v_R^3)-\frac{(u_I+v_I)^2}{2}\bigg)\Big(1+{\cal O}\big(N^{\delta-1/6}\big)\Big)\\
&\quad \times \frac{N^{2/3}}{ \sqrt 2 \pi^{3/2}}\bigg(\Psi_{11}(u_R;s)\overline{\Psi_{11}(v_R;s)}-\Psi_{21}(u_R;s)\overline{\Psi_{21}(v_R;s)}\\
&\qquad -\frac{\ii u_I\big(\Psi_{11}'(u_R;s)\overline{\Psi_{11}(v_R;s)}-\Psi_{21}'(u_R;s)\overline{\Psi_{21}(v_R;s)}\big)}{\gamma_cN^{1/6}}\\
&\qquad +\frac{\ii v_I\big(\Psi_{11}(u_R;s)\overline{\Psi_{11}'(v_R;s)}-\Psi_{21}(u_R;s)\overline{\Psi_{21}'(v_R;s)}\big)}{\gamma_cN^{1/6}}+{\cal O}\big(N^{2\delta-1/3}\big)\bigg).
\end{aligned}
\end{equation}
Here and below, for an arbitrarily given $0<\delta<1/6$, all the error bounds  will be uniform over $(u_R,N^{-\delta} u_I,v_R,N^{-\delta}v_I)$ in a compact set of ${\mathbb R}^4$, which means that we allow $u_I$ and $v_I$ to grow as $\sim N^{\delta}$.

Since $\sigma_1\widetilde\Psi(-\xi;s)\sigma_1$ and $\overline{\widetilde\Psi(-\overline{\xi};s)}$ satisfy the same Riemann-Hilbert problem as $\widetilde\Psi$ in \eqref{rhp phi}, it follows that
\begin{equation}\label{relation psi and psi tilde}
\widetilde\Psi(\xi;s)=\sigma_1\widetilde\Psi(-\xi;s)\sigma_1,\quad \widetilde\Psi(\xi;s)=\overline{\widetilde\Psi(-\overline{\xi};s)},\end{equation} where $\sigma_1=\begin{pmatrix}
0&1\\1&0
\end{pmatrix}$ is the first Pauli matrix.
Then we have
\begin{equation*}
\overline{\widetilde\Psi_{11}(\overline{\xi};s)}=\widetilde\Psi_{11}(-\xi;s)=\widetilde\Psi_{22}(\xi;s),  \quad 
\overline{\widetilde\Psi_{21}(\overline{\xi};s)}=\widetilde\Psi_{21}(-\xi;s)=\widetilde\Psi_{12}(\xi;s).
\end{equation*}
Combining above relations with the definition of $\Psi$ in \eqref{def of psi0} we have the following relation
\begin{equation}\label{relation psi1121}
  \overline{\Psi_{11}(u_R;s)} =\Psi_{21}(u_R;s), \quad \overline{\Psi_{21}(v_R;s)} =\Psi_{11}(v_R;s)
\end{equation}

Let us assume $u_R\neq v_R$; we skip the case of $u_R=v_R$ which is analogous. Using the relations in  \eqref{relation psi1121},  we rewrite the equation  \eqref{eq i and ii2} as
\begin{equation}
\begin{aligned}
\frac{1}{N^{5/6}}\overline{\partial}_{\zeta} {\cal K}_{n}(z, \zeta)
&=\frac{N^{2/3}}{ \sqrt 2 \pi^{3/2}}\exp\bigg(\frac{N^{1/3}\gamma_c^2(u_R-v_R)^2}{2}-\ii (u_R-v_R)(u_I+v_I)\gamma_c N^{1/6}\bigg)\\
&\quad \times \exp \bigg(\ii s (u_R-v_R)-\frac{4\ii(b_c+\sqrt c)}{3a}(u_R^3-v_R^3)-\frac{(u_I+v_I)^2}{2}\bigg)\\
&\quad \times \Big(\Psi_{11}(u_R;s){\Psi_{21}(v_R;s)}-\Psi_{21}(u_R;s){\Psi_{11}(v_R;s)}+{\cal O}\big(N^{\delta-1/6}\big)\Big).
\end{aligned}
\end{equation}
To integrate the above equation let us define $F_n$ such that
\begin{equation}\label{eq knN}
\begin{aligned}
\frac{1}{N^{5/6}}{\cal K}_{n}(z, \zeta)
&=\exp\bigg(\frac{N^{1/3}\gamma_c^2(u_R-v_R)^2}{2}+\ii s (u_R-v_R)-\frac{4\ii(b_c+\sqrt c)}{3a}(u_R^3-v_R^3)\bigg)\\
&\quad \times \frac{1}{ \sqrt 2 \pi^{3/2}\gamma_c} \bigg(\exp \Big(-\ii (u_R-v_R)(u_I+v_I)\gamma_c N^{1/6}-\frac{(u_I+v_I)^2}{2}\Big)\\
&\quad \times \frac{\Psi_{21}(u_R;s)\Psi_{11}(v_R;s)-\Psi_{11}(u_R;s)\Psi_{21}(v_R;s)}{\ii(u_R-v_R)}+F_n\bigg).
\end{aligned}
\end{equation}

Taking the derivative of the both side by $v_I$ using the following identities from \eqref{eq dvi}:
\begin{equation}\label{eq differ vi1}
\begin{aligned}
\frac{1}{N^{5/6}}{\partial}_{v_I} {\cal K}_{n}(z, \zeta)
&=\frac{1}{N^{5/6}N^{1/2}}(\partial_\zeta+\overline\partial_\zeta){\cal K}_n(z,\zeta) =\frac{1}{N^{5/6}N^{1/2}}\overline\partial_\zeta{\cal K}_n(z,\zeta),
\end{aligned}
\end{equation}
we get that 
\begin{equation}
 \partial_{v_I}F_n =  {\cal O}(N^{\delta-1/6}) \exp \Big(-\frac{(u_I+v_I)^2}{2}\Big).
\end{equation}
Integrating over $v_I$ while allowing $v_I$ and $u_I$ to grow as much as ${\cal O}(N^{\delta})$, we get that 
\begin{equation}
    F_n = {\cal O}(N^{\delta - 1/6}) + F_{n}\big|_{v_I=-N^{\delta}},
\end{equation}
where the last term serves as the integration constant.
Similarly, exchanging the role of $u$ and $v$, we have 
\begin{equation}
    F_n = {\cal O}(N^{\delta - 1/6}) + F_{n}\big|_{u_I=-N^{\delta}}.
\end{equation}
Combining the two equations we have 
\begin{equation}\label{bound of fn}
    F_n = {\cal O}(N^{\delta - 1/6}) + F_{n}\big|_{v_I=u_I=-N^{\delta}}.
\end{equation}
We recall that the error bound is uniform over $(u_R,N^{-\delta} u_I,v_R,N^{-\delta}v_I)$ in a compact set of ${\mathbb R}^4$.

Using the relation between the correlation kernel ${\bf K}_n$ and ${\cal K}_{n}$ in \eqref{relation of kernels} and
\begin{equation}\label{the prefactor}
 \begin{aligned}
&\quad\,\, \frac{\ee^{Nz\overline\zeta}}{\ee^{\frac{N}{2}|z|^2+\frac{N}{2}|\zeta|^2}}\frac{|z-a|^{Nc}|\zeta-a|^{Nc}}{(z-a)^{Nc}(\overline\zeta-a)^{Nc}}\\
&=\exp\bigg(\ii a\gamma_cN^{2/3}\frac{u+\bar u-v-\bar v}{2}+N^{1/3}\gamma_c^2\frac{(u-\bar u)^2+(v-\bar v)^2}{4}-\frac{N^{1/3}\gamma_c^2}{2}(u-\bar v)^2\bigg)\\
&\qquad\times \exp\bigg(\frac{\ii\gamma_c^3(u^3+\bar u^3-v^3-\bar v^3)}{6\sqrt c}\bigg)\Big(1+{\cal O}\big(\frac{1}{N^{1/3}}\big)\Big)\\
&=\exp\bigg(\ii a\gamma_cN^{2/3}(u_R-v_R)-\frac{N^{1/3}\gamma_c^2}{2}(u_R- v_R)^2+\ii (u_R-v_R)(u_I+v_I)\gamma_c N^{1/6}\bigg)\\
&\qquad \times \exp\bigg(-\frac{(u_I-v_I)^2}{2}+\frac{2\ii\gamma_c^3(u_R^3-v_R^3)}{6\sqrt c}\bigg)\Big(1+{\cal O}\big(N^{2\delta-1/3}\big)\Big),
\end{aligned}   
\end{equation}
we get 
\begin{equation}\label{bf knN}
\begin{aligned}
\frac{1}{N^{5/6}}{\bf K}_{n}(z, \zeta)
&=\frac{\ee^{Nz\overline\zeta}}{\ee^{\frac{N}{2}|z|^2+\frac{N}{2}|\zeta|^2}}\frac{|z-a|^{Nc}|\zeta-a|^{Nc}}{(z-a)^{Nc}(\overline\zeta-a)^{Nc}} \frac{1}{N^{5/6}}{\cal K}_{n}(z,\zeta)\\
&=\exp\bigg(\ii a\gamma_cN^{2/3}(u_R-v_R)+\ii s (u_R-v_R)+\frac{4\ii(u_R^3-v_R^3)}{3}\bigg)\Big(1+{\cal O}\big(N^{2\delta-1/3}\big)\Big)\\
&\quad \times \frac{1}{ \sqrt 2 \pi^{3/2}\gamma_c} \bigg(\ee^{ -u_I^2-v_I^2}\frac{\Psi_{21}(u_R;s)\Psi_{11}(v_R;s)-\Psi_{11}(u_R;s)\Psi_{21}(v_R;s)}{\ii(u_R-v_R)}\\
&\qquad +F_n\exp\Big(\ii (u_R-v_R)(u_I+v_I)\gamma_c N^{1/6}\Big)\exp\Big(-\frac{(u_I-v_I)^2}{2}\Big)\bigg).
\end{aligned}
\end{equation}
Plugging in $u_I=v_I=-N^{\delta}$ we get
\begin{equation}\label{eq kn46}
 \begin{aligned}
\Big|\frac{1}{N^{5/6}}{\bf K}_n(z, \zeta) &\big|_{u_I=v_I=-N^{\delta}}\Big|= \Big(1+{\cal O}\big(N^{2\delta-1/3}\big)\Big)\\
&\times \left(\exp\big(-{2}N^{2\delta}\big)\frac{\Psi_{21}(u_R;s)\Psi_{11}(v_R;s)-\Psi_{11}(u_R;s)\Psi_{21}(v_R;s)}{\ii(u_R-v_R)}   +  F_{n}\big|_{v_I=u_I=-N^{\delta}}
\right).    
\end{aligned}
\end{equation}
The following theorem is from Proposition 3.6 in \cite[Section 3.5]{ Ameur 2010}. For the convenience of the readers, we put the proof of the following theorem in Appendix~\ref{appendix c}.
\begin{thm}\label{Thm 41} 
Let $Q(z)$ be given in \eqref{def Q}, and $n=N t$. There exists a constant $C$ that is independent of $z$ and $N$ such that
\begin{equation}
    {\bf K}_n(z,z)\leq CN\ee^{-N {\cal U}_\text{2D}(z)},\quad z\in \partial {\cal S} \cup {\cal S}^c.
\end{equation}
Here ${\cal U}_\text{2D}$ is the effective potential of the 2D equilibrium measure defined by
$${\cal U}_\text{2D}(z):=Q(z)-\frac{2t}{\text{Area}({\cal S})}\int_{\cal S}\log|z-w|\dd A(w)+\ell_{\text{2D}},$$
where $\ell_{\text{2D}}$ is chosen such that   ${\cal U}_\text{2D}=0$ on ${\cal S}\cup\partial{\cal S}$.
\end{thm}
The explicit representation of ${\cal U}_\text{2D}$ is not important but we will remind the known properties of ${\cal U}_\text{2D}$ in our specific case: 
${\cal U}_\text{2D}=0$ on ${\cal S}$ and ${\cal U}_\text{2D}> 0$ on ${\cal S}^c$; ${\cal U}_\text{2D}$ is continuously differentiable. According to Lemma 2.2 in \cite{Ba 2015}, for any $z_0\in \partial{\cal S}$ and $\epsilon>0$ is small enough, we have that ${\cal U}_\text{2D}(z_0+\epsilon(n_x+\ii n_y))=2\epsilon^2+{\cal O}(\epsilon^3)$, where $(n_x,n_y)$ is the unit vector that is outer normal to $\partial{\cal S}$ at $z_0$.

Using Theorem \ref{Thm 41}, the above facts, and $z=b_c+\frac{u_I}{N^{1/2}}+\frac{\ii \gamma_c u_R}{N^{1/3}}$  we have ${\bf K}_n(z,z)= {\cal O}(N\exp[-\kappa u_I^2])$ as $N\to\infty$ for any constant $0<\kappa<2$.  This bound holds as long as $u_I=o(\sqrt N)$ or as long as $z-z_0=o(1)$ for any $z_0\in \partial{\cal S}$.   Using \eqref{eq kn46} and the fact that $|{\bf K}_n(z,\zeta)|\leq \sqrt{{\bf K}_n(z,z){\bf K}_n(\zeta,\zeta)} $ we have that 
$$F_{n}\big|_{v_I=u_I=-N^{\delta}}={\cal O}(N^{1/6}\exp(-\kappa N^{2\delta})).$$    Consequently, for finite $u_I$ and $v_I$, using \eqref{bf knN} we have
$$\begin{aligned}
\frac{1}{N^{5/6}}{\bf K}_{n}(z,\zeta)=\frac{1}{C_N(u_R,v_R)\gamma_c}\frac{\ee^{-u_I^2-v_I^2}}{\sqrt{\pi/2}}  \frac{\Psi_{21}(u_R;s)\Psi_{11}(v_R;s)-\Psi_{11}(u_R;s)\Psi_{21}(v_R;s)}{2\pi\ii(u_R-v_R)}+{\cal O}\Big(\frac{1}{N^{1/6-\delta}}\Big),
\end{aligned}$$
where
 \begin{equation*}
   C_N(u_R,v_R):=
   \ee^{-\ii (\varphi_N(u_R)-\varphi_N(v_R))}\,, \qquad \varphi_N(w):=a\gamma_cN^{2/3} w+s w+4w^3/3. 
      \end{equation*}

Changing the notations from $(u_I,u_R,v_I,v_R)$ to  $(x,y,x',y')$ respectively, we have \eqref{eq limiting kernel}. This ends the proof.

\section{Proof of Theorem \ref{main theorem3}}\label{section 4}
In this section, we prove Theorem \ref{main theorem3} using a similar strategy as in the proof of Theorem \ref{main theorem}.  

\subsection{The C-D identity: Proof of Lemma \ref{preker1}}
\begin{lemma}\label{preker1}
Let ${\cal K}_{n}(z,\zeta)$ be the pre-kernel defined in \eqref{def prekernel}, and ${D}_\tau:=\{z: |z-b_c^*|\leq\frac{1}{N^\tau}\}$ with $0\leq \tau<\frac{1}{3}$.  Defining the local coordinates $u$ and $v$ by
\begin{equation}\label{def zxy}
z:=b_c+u\quad\mbox{and}\quad    \zeta:=b_c+v  
\end{equation} such that $z\in D_c\setminus D_\tau$ and $\zeta\in D_c\setminus D_\tau$.
Then we have
\begin{equation}\label{eq prekernel 2}
\begin{aligned}
\overline{\partial}_{\zeta} {\cal K}_{n}(z, \zeta)
&=\frac{ N^{\frac32}}{\sqrt 2 \pi^{\frac32}  }\ee^{-Nu\overline v-\frac{N}{2}(u^2+\overline v^2)}\bigg[\exp\bigg(N\Big(\frac{u^3+\overline v^3}{3\sqrt c}-\frac{u^4+\overline v^4}{4c}+\frac{u^5+\overline v^5}{5c^{3/2}}\Big)\bigg)\\
&\quad \times \Big(1+\frac{\gamma_c r}{ N^{1/3}}\frac{u+\overline v}{2u\overline v}-\frac{\gamma_c^2(q^2-r^2)}{ N^{2/3}}\frac{(u+\overline v)^2}{8u^2\overline v^2}+{\cal O}\Big(Nu^6,N\overline v^6,\frac{N^{3\tau}}{N}\Big)\Big)\\
&  +\exp\bigg(N\Big(\frac{1}{3}\Big(\frac{\overline v^3}{b_c}+\frac{u^3}{\sqrt c}\Big)-\frac{1}{4}\Big(\frac{\overline v^4}{b_c^2}+\frac{u^4}{c^2}\Big)+\frac{1}{5}\Big(\frac{\overline v^5}{b_c^3}+\frac{u^5}{c^{3/2}}\Big)\Big)+\frac{2sN^{1/3}\overline v}{\gamma_c}\bigg)\\
&\quad \times \Big(\frac{\gamma_cq}{N^{1/3}}\frac{(\overline v+u)}{2u\overline v}+
\frac{\gamma_c^2(qr+q')}{4N^{2/3}}\Big(\frac{1}{u^2}-\frac{1}{\overline v^2}\Big)+{\cal O}\Big(\frac{N^{3\tau}}{N}\Big)\Big)\big(1+{\cal O}\big(Nu^6,N^{1/3}\overline v^2\big)\big)\\
& +\exp\bigg(N\Big(\frac{1}{3}\Big(\frac{u^3}{b_c}+\frac{\overline v^3}{\sqrt c}\Big)-\frac{1}{4}\Big(\frac{u^4}{b_c^2}+\frac{\overline v^4}{c}\Big)+\frac{1}{5}\Big(\frac{u^5}{b_c^3}+\frac{\overline v^5}{c^{3/2}}\Big)\Big)+\frac{2sN^{1/3}u}{\gamma_c}\bigg)\\
&\quad \times \Big(\frac{\gamma_cq}{N^{1/3}}\frac{(\overline v+u)}{2u\overline v}-
\frac{\gamma_c^2(qr+q')}{4N^{2/3}}\Big(\frac{1}{u^2}-\frac{1}{\overline v^2}\Big)+ {\cal O}\Big(\frac{N^{3\tau}}{N}\Big)\Big)\big(1+{\cal O}\big(N^{1/3}u^2,N\overline v^6\big)\big)\\
& -\exp\bigg(N\Big(\frac{u^3+\overline v^3}{3b_c}-\frac{u^4+\overline v^4}{4b_c^2}+\frac{u^5+\overline v^5}{5b_c^{3}}\Big)+\frac{2sN^{1/3}}{\gamma_c}(u+\overline v)\bigg)\\
&\quad \times \Big(1-\frac{\gamma_c r}{ N^{1/3}}\frac{u+\overline v}{2u\overline v}-\frac{\gamma_c^2(q^2-r^2)}{ N^{2/3}}\frac{(u+\overline v)^2}{8u^2\overline v^2}+{\cal O}\Big(N^{1/3}u^2,N^{1/3}\overline v^2,\frac{N^{3\tau}}{N}\Big)\Big)
\bigg] .\end{aligned}   
\end{equation}
\end{lemma}

\begin{proof}
As $z\in D_c\setminus D_\tau$, $\zeta\in D_c\setminus D_\tau$ and 
$$b_c^*=b_c+\frac{s}{2\gamma_cN^{2/3}}+{\cal O}\Big(\frac{1}{N^{4/3}}\Big),$$
we will use Proposition \ref{thm ortho1} to compute the derivative of the pre-kernel ${\cal K}_{n}(z, \zeta)$ in Theorem \ref{thm 21}. 

Since $z\notin D_\tau$ we have $1/u={\cal O}(N^\tau)$. Recall that $$\overline{\partial}_{\zeta} {\cal K}_{n}(z, \zeta)= {\rm I}_n- {\rm II}_n.$$

For convenience, we write $p_n$ and $q_n$ in \eqref{pn omegaplus} as
\begin{equation}
\begin{aligned}
p_n(z)&=\ee^{Naz+tN\ell}\ee^{\ii\frac{16}{3}\xi(\beta)^3}\big(P_{21}^z\big)+\ee^{Naz+tN\ell}\ee^{-N\phi(z)}\big(P_{22}^z\big),\\
q_n(z) &=\ee^{Naz}\ee^{-N\phi(z)}\ee^{-\ii\frac{16}{3}\xi(\beta)^3}\big(P_{12}^z\big)-\ee^{Naz}\big(P_{11}^z\big),
\end{aligned}        
\end{equation}
using the short-hand notations below.
\begin{align}
P_{11}^z:&=  1-\frac{r}{2r_1N^{1/3}(z-b_c^*)}-\frac{p_{11}}{r_1^2N^{2/3}(z-b_c^*)^2}+{\cal O}\Big(\frac{N^{3\tau}}{N}\Big),\\
P_{12}^z:&=\frac{q}{2r_1N^{1/3}(z-b_c^*)}-\frac{r_2p_{12}}{r_1^3N^{2/3}(z-b_c^*)}+\frac{p_{12}}{r_1^2N^{2/3}(z-b_c^*)^2}+{\cal O}\Big(\frac{N^{3\tau}}{N}\Big),\\
P_{21}^z:&=\frac{q}{2r_1N^{1/3}(z-b_c^*)}+\frac{r_2p_{21}}{r_1^3N^{2/3}(z-b_c^*)}-\frac{p_{21}}{r_1^2N^{2/3}(z-b_c^*)^2}+{\cal O}\Big(\frac{N^{3\tau}}{N}\Big),\\
P_{22}^z:&=1+\frac{r}{2r_1N^{1/3}(z-b_c^*)}-\frac{p_{22}}{r_1^2N^{2/3}(z-b_c^*)^2}+{\cal O}\Big(\frac{N^{3\tau}}{N}\Big).
\end{align}  As before, $q$, $r$, $p_{11}$ , $p_{12}$, $p_{21}$ and $p_{22}$ depend on $\hat s$, i.e., $q=q(\hat s)$ etc.

When  $z\in D_c\setminus {D}_\tau$,  using \eqref{pn omegaplus} in  Proposition \ref{thm ortho1},  ${\rm I}_{n}$ from \eqref{eq in middle} can be written as
\begin{align}\label{eq IXY}
\begin{split}
&{\rm I}_n=\widehat Q_n \frac{1}{\ee^{\ii\frac{16}{3}\xi(\beta)^3}\ee^{tN\ell}}\frac{N\alpha_n}{ 1+\beta_n\gamma_n}\Big(P_{11}^z-\ee^{-\ii\frac{16}{3}\xi(\beta)^3}\ee^{-N\phi(z)} P_{12}^z\Big)\\
&\quad\quad\times \Bigg((c+t)\bigg(\frac{1+\beta_n\gamma_n}{\zeta-a}-\frac{\beta_n\gamma_n}{\zeta}\bigg) \Big(P_{21}^\zeta +\ee^{-\ii\frac{16}{3}\xi(\beta)^3}\ee^{-N\phi(\zeta)}P_{22}^\zeta\Big)\\
&\quad\quad \quad -\frac{1}{\ee^{\ii\frac{16}{3}\xi(\beta)^3}\ee^{tN\ell}}\bigg(\frac{(c+t)\beta_n\alpha_n}{\zeta}+\frac{ab_n-(c+t)\beta_n\alpha_n}{\zeta-a}\bigg) \Big(P_{11}^\zeta-\ee^{-\ii\frac{16}{3}\xi(\beta)^3}\ee^{-N\phi(\zeta)}P_{12}^\zeta\Big)\Bigg)^*,
\end{split}
\end{align} 
where we define
\begin{equation}\label{eq common2}
\widehat Q_n={{\ee^{ -Nz\overline\zeta} }}\frac{N^{Nc+n+1}}{\pi\Gamma(Nc+n+1)}(z-a)^{Nc}\ee^{Naz}\ee^{\frac{16}{3}\ii\xi(\beta)^3}\ee^{tN\ell}\Big((\zeta-a)^{Nc}\ee^{Na\zeta}\ee^{\frac{16}{3}\ii\xi(\beta)^3}\ee^{tN\ell}\Big)^*.    
\end{equation}

Similarly, ${\rm II}_{n}$ in \eqref{eq iin middle} can be written as \begin{equation}\label{eq IIXY}
\begin{aligned}
&{\rm II}_{n}=\widehat Q_n \frac{N(Nc+n)}{Nc+n+1}\bigg(a+\frac{b_n \left(1+\beta_n\gamma_n\right)}{\alpha_n
\beta_n}-b_n\frac{q_n(a)}{p_n(a)}\bigg)\frac{\beta_n\gamma_n}{ \big(\beta_{n+1}-\beta_n\big)}\\
&\quad\quad\times \Bigg(\frac{b_n \left(1+\beta_n\gamma_n\right)}{\alpha_n
\beta_n}\Big(P_{21}^z +\ee^{-\ii\frac{16}{3}\xi(\beta)^3}\ee^{-N\phi(z)}P_{22}^z\Big)+\frac{b_n}{\ee^{\ii\frac{16}{3}\xi(\beta)^3}\ee^{tN\ell}} \Big(P_{11}^z-\ee^{-\ii\frac{16}{3}\xi(\beta)^3}\ee^{-N\phi(z)} P_{12}^z\Big)\Bigg)\\
&\quad\quad \times \frac{1}{\overline \zeta}\Bigg(\Big(P_{21}^\zeta +\ee^{-\ii\frac{16}{3}\xi(\beta)^3}\ee^{-N\phi(\zeta)}P_{22}^\zeta\Big)+\frac{\alpha_n}{\gamma_n\ee^{\ii\frac{16}{3}\xi(\beta)^3}\ee^{tN\ell}}\Big(P_{11}^\zeta-\ee^{-\ii\frac{16}{3}\xi(\beta)^3}\ee^{-N\phi(\zeta)} P_{12}^\zeta\Big)\Bigg)^*.
\end{aligned}
\end{equation}

Similar to the proof of Theorem \ref{main theorem}, we now compute ${\rm I}_{n}$ and ${\rm II}_{n}$. 

Using the expanssion of $\ell$ in \eqref{ell expansion} and the definition of $\phi(z)$ in \eqref{def phi}  we have \begin{align}\label{common11}
\ee^{tN\ell}&=\exp\Big(-Nab_c-\frac{2{s^{3/2}}}{3}\Big)b_c^{Nb_c^2+N^{1/3}\frac{2b_cs}{\gamma_c}}c^{-Nc/2}\Big(1+{\cal O}\Big(\frac{1}{N^{1/3}}\Big)\Big),\\\label{common12}
\ee^{-N\phi(z)}\ee^{-\ii\frac{16}{3}\xi(\beta)^3}
    &=\exp (-Na u)\Big(\frac{\sqrt c}{\sqrt c + u}\Big)^{Nc} \Big(1+\frac{u}{b_c}\Big)^{Nb_c^2+N^{1/3}\frac{2b_cs}{\gamma_c}} \Big(1+{\cal O}\Big(\frac{1}{N^{1/3}}\Big)\Big).  \end{align}

Using the definition of $z$ in \eqref{def zxy}, \eqref{common11} and the 2nd identity in \eqref{common1}, $\widehat Q_n$ from \eqref{eq common2}  is written as 
\begin{equation}\label{commoncommon}
\begin{aligned}
\widehat Q_n
&=\exp \big(-N\sqrt c (\overline v+u)-Nu\overline v\big)\Big(1+\frac{u}{\sqrt c}\Big)^{Nc}\Big(1+\frac{\overline v}{\sqrt c}\Big)^{Nc}\\
&\quad \times\exp\bigg(\frac{32}{3}{\rm Re}(\ii\xi(\beta)^3)-\frac{4{\rm Re}(s^{3/2})}{3}\bigg)   \frac{\sqrt N}{\sqrt 2 \pi^{3/2}b_c }\Big(1+{\cal O}\Big(\frac{1}{N^{1/3}}\Big)\Big)\\
&=\exp \big(-N\sqrt c (\overline v+u)-Nu\overline v\big)\Big(1+\frac{u}{\sqrt c}\Big)^{Nc}\Big(1+\frac{\overline v}{\sqrt c}\Big)^{Nc} \frac{\sqrt N}{\sqrt 2 \pi^{3/2}b_c}\Big(1+{\cal O}\Big(\frac{1}{N^{1/3}}\Big)\Big),
\end{aligned}    
\end{equation}  where the 2nd equality is obtained by the expression for $\xi(\beta)$ in \eqref{zeta beta} with  $\hat s = s(1+{\cal O}(N^{-2/3}))$ which comes from \eqref{def zetabeta} and $t-t_c = {\cal O}(N^{-2/3})$. 

We now compute $({\rm I}_n-{\rm II}_n)/\widehat Q_n$. Recall that $u$ and $v$ are chosen such that $z\notin D_\tau$ and $\zeta\notin D_\tau$. Using ${\rm I}_{n}$ in \eqref{eq IXY} and  ${\rm II}_{n}$ in \eqref{eq IIXY} we have
\begin{equation}\label{main difference}
\begin{aligned}
\frac{{\rm I}_{n}}{\widehat Q_n}-\frac{{\rm II}_{n}}{\widehat Q_n}
&={\cal C}_1 \Big(P_{11}^z-\ee^{-\ii\frac{16}{3}\xi(\beta)^3}\ee^{-N\phi(z)} P_{12}^z\Big)\Big(P_{21}^\zeta +\ee^{-\ii\frac{16}{3}\xi(\beta)^3}\ee^{-N\phi(\zeta)}P_{22}^\zeta\Big)^*\\
& \quad - {\cal C}_2 \Big(P_{11}^z-\ee^{-\ii\frac{16}{3}\xi(\beta)^3}\ee^{-N\phi(z)} P_{12}^z\Big)\Big(P_{11}^\zeta-\ee^{-\ii\frac{16}{3}\xi(\beta)^3}\ee^{-N\phi(\zeta)} P_{12}^\zeta\Big)^*\\
& \quad -{\cal C}_3 \Big(P_{21}^z +\ee^{-\ii\frac{16}{3}\xi(\beta)^3}\ee^{-N\phi(z)}P_{22}^z\Big)\Big(P_{21}^\zeta +\ee^{-\ii\frac{16}{3}\xi(\beta)^3}\ee^{-N\phi(\zeta)}P_{22}^\zeta\Big)^* \\
& \quad -{\cal C}_4 \Big(P_{21}^z +\ee^{-\ii\frac{16}{3}\xi(\beta)^3}\ee^{-N\phi(z)}P_{22}^z\Big)\Big(P_{11}^\zeta-\ee^{-\ii\frac{16}{3}\xi(\beta)^3}\ee^{-N\phi(\zeta)} P_{12}^\zeta\Big)^*,
\end{aligned}
\end{equation}
where 
\begin{equation}
    \begin{aligned}
        {\cal C}_1:&=\frac{1}{\ee^{\ii\frac{16}{3}\xi(\beta)^3}\ee^{tN\ell}}\frac{N\alpha_n}{1+\beta_n\gamma_n}\bigg((c+t)\Big(\frac{1+\beta_n\gamma_n}{\zeta-a}-\frac{\beta_n\gamma_n}{\zeta}\Big)  \bigg)^*\\
&\quad -\frac{N(Nc+n)}{Nc+n+1}\bigg(a+\frac{b_n \left(1+\beta_n\gamma_n\right)}{\alpha_n
\beta_n}-b_n\frac{q_n(a)}{p_n(a)}\bigg)\frac{\beta_n\gamma_n}{ \big(\beta_{n+1}-\beta_n\big)}\frac{b_n}{\ee^{\ii\frac{16}{3}\xi(\beta)^3}\ee^{tN\ell}} \frac{1}{\overline \zeta},\\
        {\cal C}_2:&=\frac{N(Nc+n)}{Nc+n+1}\bigg(a+\frac{b_n \left(1+\beta_n\gamma_n\right)}{\alpha_n
\beta_n}-b_n\frac{q_n(a)}{p_n(a)}\bigg)\frac{\beta_n\gamma_n}{ \beta_{n+1}-\beta_n}\frac{1}{\overline \zeta}\frac{b_n}{\ee^{\ii\frac{16}{3}\xi(\beta)^3}\ee^{tN\ell}} \bigg(\frac{\alpha_n}{\gamma_n\ee^{\ii\frac{16}{3}\xi(\beta)^3}\ee^{tN\ell}}\bigg)^*\\
&\quad + \frac{1}{\ee^{\ii\frac{16}{3}\xi(\beta)^3}\ee^{tN\ell}}\frac{N\alpha_n}{ 1+\beta_n\gamma_n}\bigg(\frac{1}{\ee^{\ii\frac{16}{3}\xi(\beta)^3}\ee^{tN\ell}}\Big(\frac{(c+t)\beta_n\alpha_n}{\zeta}+\frac{ab_n-(c+t)\beta_n\alpha_n}{\zeta-a}\Big)\bigg)^*,\\
        {\cal C}_3:&=\frac{N(Nc+n)}{Nc+n+1}\bigg(a+\frac{b_n \left(1+\beta_n\gamma_n\right)}{\alpha_n
\beta_n}-b_n\frac{q_n(a)}{p_n(a)}\bigg)\frac{\beta_n\gamma_n}{ \beta_{n+1}-\beta_n}\frac{1}{\overline \zeta} \frac{b_n \left(1+\beta_n\gamma_n\right)}{\alpha_n
\beta_n},\\
        {\cal C}_4:&={\cal C}_3\bigg(\frac{\alpha_n}{\gamma_n\ee^{\ii\frac{16}{3}\xi(\beta)^3}\ee^{tN\ell}}\bigg)^*.
    \end{aligned}
\end{equation}
We further write \eqref{main difference} as
\begin{equation}\label{main difference 1}
\begin{aligned}
\frac{{\rm I}_{n}}{\widehat Q_n}-\frac{{\rm II}_{n}}{\widehat Q_n}
&={\cal C}_1P_{11}^z(P_{21}^\zeta)^*-{\cal C}_2 P_{11}^z(P_{11}^\zeta)^*-{\cal C}_3P_{21}^z(P_{21}^\zeta)^*-{\cal C}_4P_{21}^z(P_{11}^\zeta)^*\\
&\quad +\big(\ee^{-\ii\frac{16}{3}\xi(\beta)^3}\ee^{-N\phi(\zeta)}\big)^*\Big({\cal C}_1P_{11}^z(P_{22}^\zeta)^*+{\cal C}_2P_{11}^z(P_{12}^\zeta)^*-{\cal C}_3P_{21}^z(P_{22}^\zeta)^*+{\cal C}_4P_{21}^z(P_{12}^\zeta)^*\Big)\\
&\quad +\ee^{-\ii\frac{16}{3}\xi(\beta)^3}\ee^{-N\phi(z)}\Big(-{\cal C}_1 P_{12}^z(P_{21}^\zeta)^* +{\cal C}_2 P_{12}^z(P_{11}^\zeta)^*-{\cal C}_3P_{22}^z(P_{21}^\zeta)^*-{\cal C}_4P_{22}^z(P_{11}^\zeta)^* \Big)\\
&\quad +\ee^{-\ii\frac{16}{3}\xi(\beta)^3}\ee^{-N\phi(z)} \big(\ee^{-\ii\frac{16}{3}\xi(\beta)^3}\ee^{-N\phi(\zeta)}\big)^*\Big(-{\cal C}_1P_{12}^z(P_{22}^\zeta)^*-{\cal C}_2 P_{12}^z(P_{12}^\zeta)^*-{\cal C}_3P_{22}^z(P_{22}^\zeta)^*+{\cal C}_4P_{22}^z(P_{12}^\zeta)^*\Big).
\end{aligned}
\end{equation}

Using \eqref{def zxy} we get 
\begin{equation}\label{eq zandbcstar}
 \frac{1}{z-b_c^*}
 = \frac{1}{u}\bigg(1+ \frac{b_c-b_c^*}{u}\bigg)^{-1} =\frac{1}{u}+{\cal O}\Big(\frac{N^{2\tau}}{ N^{2/3}}\Big).
\end{equation} 
Using the relations of $r$, $p_{11}$ and $p_{22}$ in \eqref{relation up}, the expansion of $r_1$ in \eqref{eq-gamma123}, and \eqref{eq zandbcstar}, yields
\begin{equation}\label{P11expan}
\begin{aligned}
P_{11}^z&=1+\frac{\gamma_c r}{2u N^{1/3}}-\frac{\gamma_c^2(q^2-r^2)}{8u^2N^{2/3}}+{\cal O}\Big(\frac{N^{3\tau}}{N}\Big),\\
P_{22}^z&=1-\frac{\gamma_c r}{2u N^{1/3}}-\frac{\gamma_c^2(q^2-r^2)}{8u^2N^{2/3}}+{\cal O}\Big(\frac{N^{3\tau}}{N}\Big).
\end{aligned}   
\end{equation}

Similarly, with the relations of $r$, $p_{12}$ and $p_{21}$ in \eqref{relation up}, the expansions of $r_1, r_2$ in \eqref{eq-gamma123}, and \eqref{eq zandbcstar}, we find
\begin{equation}\label{P21expan}
\begin{aligned}
P_{12}^z
=-\frac{\gamma_cq}{2uN^{1/3}}+
\frac{\gamma_c^2(qr+q')}{4u^2N^{2/3}}+{\cal O}\Big(\frac{N^{3\tau}}{ N}\Big),\quad 
P_{21}^z
=-\frac{\gamma_cq}{2uN^{1/3}}-
\frac{\gamma_c^2(qr+q')}{4u^2N^{2/3}}+{\cal O}\Big(\frac{N^{3\tau}}{ N}\Big).    
\end{aligned}
\end{equation}

Moreover,  the expansions of $\beta_n\gamma_n$, $\beta_n\alpha_n$ and $b_n$ in \eqref{eq alphabngamma} and \eqref{def zxy}, imply
\begin{equation}\label{zza01}\frac{\alpha_n}{\gamma_n\ee^{\ii\frac{16}{3}\xi(\beta)^3}\ee^{tN\ell}}=-\frac{q\gamma_c}{2b_cN^{1/3}}+{\cal O}\Big(\frac{1}{N^{2/3}}\Big),\end{equation}
\begin{equation}\label{zza02}\frac{1}{z}\frac{b_n \left(1+\beta_n\gamma_n\right)}{\alpha_n
\beta_n}=\frac{q^2\gamma_c^2}{4b_c(b_c+u)N^{2/3}}+{\cal O}\Big(\frac{1}{N}\Big),\end{equation}
\begin{equation}\label{zza1}
 (c+t)\Big(\frac{1+\beta_n\gamma_n}{z-a}-\frac{\beta_n\gamma_n}{z}\Big)=\frac{b_c^2}{b_c+ u}+\frac{8b_cs(\sqrt c + u)+aq^2\gamma_c^3}{4(\sqrt c +u)(b_c+u)\gamma_cN^{2/3}}  +{\cal O}\Big(\frac{1}{N}\Big)   
\end{equation}
and
\begin{equation}\label{zza2}
\begin{aligned}
\frac{1}{\ee^{\ii\frac{16}{3}\xi(\beta)^3}\ee^{tN\ell}}\Big(\frac{(c+t)\beta_n\alpha_n}{z}&+\frac{ab_n-(c+t)\beta_n\alpha_n}{z-a}\Big)=\frac{au   \gamma _c q}{2 (\sqrt c+u)(b_c+u)N^{1/3}}\\
& -\frac{a(b_c+\sqrt c)u\gamma_c^2(qr+q')}{8b_c\sqrt{c}(\sqrt c+u)(b_c+u)N^{2/3}}-\frac{a\gamma_c^2q'}{4(\sqrt c+u)(b_c+u)N^{2/3}}+{\cal O}\Big(\frac{1}{N}\Big).    
\end{aligned}    
\end{equation}

Combining \eqref{I1}, \eqref{zza1}, \eqref{II1}, and $b_n$ in \eqref{eq alphabngamma}, results in
\begin{equation}
\begin{aligned}
 {\cal C}_1={\cal O}(N^{2/3}).
\end{aligned}
\end{equation}

Similarly, combining \eqref{II1}, $b_n$ in \eqref{eq alphabngamma}, \eqref{zza01}, \eqref{I1} and  \eqref{zza2}, we see that 
\begin{equation}
\begin{aligned}
{\cal C}_2=-\frac{Nb_c\sqrt c}{\sqrt c+ \overline v
}+{\cal O}(N^{2/3}),
\end{aligned}
\end{equation}
whereas combining \eqref{II1},  \eqref{zza02}, and \eqref{zza01}, shows that 
\begin{equation}
\begin{aligned}
{\cal C}_3=\frac{Nb_c^2}{b_c+ \overline v
}+{\cal O}(N^{2/3}),\quad {\cal C}_4={\cal O}(N^{2/3}).
\end{aligned}
\end{equation}

Substituting $\{{\cal C}_j\}$ from above expansions, \eqref{P11expan}, and \eqref{P21expan} into \eqref{main difference 1}, reveals
\begin{equation}\label{remainremain}
\begin{aligned}
\frac{({\rm I}_n-{\rm II_n})}{\widehat Q_n}
&=\frac{Nb_c\sqrt c}{\sqrt c+\overline v}\Big(1+\frac{\gamma_c r}{ N^{1/3}}\frac{u+\overline v}{2u\overline v}+\frac{\gamma_c^2r^2}{4u\overline v N^{2/3}}-\frac{\gamma_c^2(q^2-r^2)}{ N^{2/3}}\frac{(u^2+\overline v^2)}{8u^2\overline v^2}+{\cal O}\Big(\frac{N^{3\tau}}{N}\Big)\Big)\\
&-\frac{Nb_c^2}{b_c+\overline v}\Big(\frac{\gamma_c^2q^2}{4u\overline v N^{2/3}}+{\cal O}\Big(\frac{N^{3\tau}}{N}\Big)\Big)+{\cal O}(N^{1/3+\epsilon})\\
&+\ee^{-Na\overline v}\Big(\frac{\sqrt c}{\sqrt c + \overline v}\Big)^{Nc} \Big(1+\frac{\overline v}{b_c}\Big)^{Nb_c^2+\frac{2b_csN^{\frac13}}{\gamma_c}}\bigg(-\frac{Nb_c\sqrt c}{\sqrt c+\overline v}\Big(\frac{-\gamma_cq}{2\overline v N^{1/3}}+
\frac{\gamma_c^2(qr+q')}{4\overline v^2N^{2/3}}-\frac{\gamma_c^2qr}{4u\overline vN^{2/3}}+{\cal O}\Big(\frac{N^{3\tau}}{N}\Big)\Big)\\
&\quad -\frac{Nb_c^2}{b_c+\overline v}\Big(\frac{-\gamma_cq}{2 u N^{1/3}}-
\frac{\gamma_c^2(qr+q')}{4u^2N^{2/3}}+\frac{\gamma_c^2qr}{4u\overline vN^{2/3}}+{\cal O}\Big(\frac{N^{3\tau}}{N}\Big)\Big)+{\cal O}(N^{\frac23})\bigg)\\
&+\ee^{-Nau}\Big(\frac{\sqrt c}{\sqrt c + u}\Big)^{Nc} \Big(1+\frac{u}{b_c}\Big)^{Nb_c^2+\frac{2b_csN^{\frac13}}{\gamma_c}} \bigg(-\frac{Nb_c\sqrt c}{\sqrt c+\overline v}\Big(\frac{-\gamma_cq}{2uN^{1/3}}+
\frac{\gamma_c^2(qr+q')}{4u^2N^{2/3}}-\frac{\gamma_c^2qr}{4u\overline vN^{2/3}}+{\cal O}\Big(
\frac{N^{3\tau}}{ N}\Big)\Big)\\
&\quad -\frac{Nb_c^2}{b_c+\overline v}\Big(\frac{-\gamma_cq}{2\overline v N^{1/3}}-
\frac{\gamma_c^2(qr+q')}{4\overline v^2N^{2/3}}+\frac{\gamma_c^2qr}{4u\overline vN^{2/3}}+{\cal O}\Big(
\frac{N^{3\tau}}{ N}\Big)\Big)+{\cal O}(N^{2/3})\bigg)\\
& +\ee^{-Na(u+\overline v)}\Big(\frac{\sqrt c}{\sqrt c + u}\frac{\sqrt c}{\sqrt c + \overline v}\Big)^{Nc} \Big(\frac{b_c+u}{b_c}\frac{b_c+\overline v}{b_c}\Big)^{Nb_c^2+N^{\frac{1}{3}}\frac{2b_cs}{\gamma_c}}\bigg(\frac{Nb_c\sqrt c}{\sqrt c+\overline v}\Big(\frac{\gamma_c^2q^2}{4u\overline v N^{2/3}}+{\cal O}\Big(\frac{N^{3\tau}}{N}\Big)\Big)\\
&\quad -\frac{Nb_c^2}{b_c+\overline v}\Big(1-\frac{\gamma_c r}{ N^{1/3}}\frac{u+\overline v}{2u\overline v}+\frac{\gamma_c^2r^2}{4u\overline v N^{2/3}}-\frac{\gamma_c^2(q^2-r^2)}{ N^{2/3}}\frac{(u^2+\overline v^2)}{8u^2\overline v^2}+{\cal O}\Big(\frac{N^{3\tau}}{N}\Big)\Big)+{\cal O}(N^{1/3+\epsilon})\bigg)
.\end{aligned}   
\end{equation}
Here the first error bound ${\cal O}(N^{1/3+\epsilon})$ comes from the terms ${\cal C}_1P_{11}^z(P_{21}^\zeta)^*$ and ${\cal C}_4P_{21}^z(P_{11}^\zeta)^*$ in \eqref{main difference 1}, and the second error bound ${\cal O}(N^{1/3+\epsilon})$ comes from the terms ${\cal C}_1P_{12}^z(P_{22}^\zeta)^*$ and ${\cal C}_4P_{22}^z(P_{12}^\zeta)^*$ in \eqref{main difference 1}, the first error bound ${\cal O}(N^{2/3})$ comes from the term ${\cal C}_1P_{11}^z(P_{22}^\zeta)^*$ in \eqref{main difference 1}, and the second error bound ${\cal O}(N^{2/3})$ comes from the term ${\cal C}_4P_{22}^z(P_{11}^\zeta)^*$ in \eqref{main difference 1}. 
And the error bound ${\cal O}(N^{3\tau-1})$ stands for ${\cal O}\big(\frac{1}{Nu^3},\frac{1}{N\overline v^3},\frac{1}{Nu^2\overline v},\frac{1}{Nu\overline v^2}\big)$. 

Using the fact $1/u={\cal O}(N^\tau)$, we have, for any non-zero constant $k$,
\begin{equation}\label{eq unc2}
        \Big(1+\frac{u}{k}\Big)^{Nk^2}=\exp\Big(Nk^2\log\Big(1+\frac{u}{k}\Big)\Big)=\exp\bigg(Nk^2\Big(\frac{u}{k}-\frac{u^2}{2k^2}+\frac{u^3}{3k^{3}}-\frac{u^4}{4k^4}+\frac{u^5}{5k^5}\Big)\bigg)\big(1+{\cal O}\big(Nu^6\big)\big)
\end{equation} and 
\begin{equation}
    \Big(1+\frac{u}{b_c}\Big)^{N^{\frac{1}{3}}\frac{2b_cs}{\gamma_c}}=\exp\bigg(\frac{2suN^{1/3}}{\gamma_c}\bigg)\big(1+{\cal O}\big(N^{1/3}u^2\big)\big).
\end{equation}

It follows that
\begin{equation}\label{eq fact 1}
\begin{aligned}
\ee^{-Nau}\Big(\frac{\sqrt c}{\sqrt c + u}\Big)^{Nc} \Big(1+\frac{u}{b_c}\Big)^{Nb_c^2}
&=\exp\bigg(N\Big(\frac{u^3}{3}\Big(\frac{1}{b_c}-\frac{1}{\sqrt c}\Big)-\frac{u^4}{4}\Big(\frac{1}{b_c^2}-\frac{1}{c}\Big)+\frac{u^5}{5}\Big(\frac{1}{b_c^3}-\frac{1}{c^{3/2}}\Big)\Big)\bigg)\big(1+{\cal O}\big(Nu^6\big)\big),
\end{aligned}\end{equation}
\begin{equation}\label{eq fact 2}
\begin{aligned}
\ee^{-Nk(u+\overline v)}\Big(1+\frac{u}{k}\Big)^{Nk^2}\Big(1+\frac{\overline v}{k}\Big)^{Nk^2}
&=\exp\bigg(-N\Big(\frac{u^2+\overline v^2}{2}-\frac{u^3+\overline v^3}{3k}+\frac{u^4+\overline v^4}{4k^2}-\frac{u^5+\overline v^5}{5k^{3}}\Big)\bigg)\big(1+{\cal O}\big(Nu^6,N\overline v^2\big)\big)
\end{aligned}    
\end{equation} and 
\begin{equation}\label{eq fact 3}
    \Big(1+\frac{u}{b_c}\Big)^{N^{\frac{1}{3}}\frac{2b_cs}{\gamma_c}}\Big(1+\frac{\overline v}{b_c}\Big)^{N^{\frac{1}{3}}\frac{2b_cs}{\gamma_c}}=\exp\bigg(\frac{2s(u+\overline v)N^{1/3}}{\gamma_c}\bigg)\big(1+{\cal O}\big(N^{1/3}u^2,N^{1/3}\overline v^2\big)\big).
\end{equation}

Combining  \eqref{commoncommon} and  \eqref{remainremain} with \eqref{eq fact 1}, \eqref{eq fact 2} and \eqref{eq fact 3}, we have \eqref{eq prekernel 2}. This ends the proof.    
\end{proof}

\subsection{Integrating the C-D identity: Proof of Theorem \ref{main theorem3}}
Let us define
\begin{equation}\label{def further uv} 
\begin{aligned}
z:&=b_c+u,\quad u:=\frac{\ii(2cb_c^2)^{\frac{1}{4}}Y}{N^\tau}+\frac{\nu}{\sqrt N},\\
\zeta:&=b_c+v,\quad v:=\frac{\ii(2cb_c^2)^{\frac{1}{4}}Y}{N^\tau}+\frac{\eta}{\sqrt N}, 
\end{aligned}\end{equation}
where $(\nu,\eta)\in\CC^2$ and $Y\in\RR$.  

The following lemma holds for all positive $\tau$.

\begin{lemma}\label{Lemma 5.2}
Let $u$ and $v$ be given above. Assuming $\nu$ and $\eta$ are all ${\cal O}(N^\epsilon)$ with $\epsilon<1/2-\tau$ and $Y$ is bounded,  we have the following identities:
\begin{equation}\label{eq uv sym}
\begin{aligned}
&\quad \exp\Big(-Nu\overline v-\frac{N}{2}(u^2+\overline v^2)\Big)\exp\bigg(N\Big(\frac{u^3+\overline v^3}{3k}-\frac{u^4+\overline v^4}{4k^2}+\frac{u^5+\overline v^5}{5k^{3}}\Big)\bigg)\\
&=\exp\bigg(-\Big(\frac{\nu+\overline \eta}{\sqrt 2}+\frac{\sqrt c b_c Y^2\sqrt N}{k N^{2\tau}}\Big)^2\bigg)\Big(1+{\cal O}\Big(\frac{N^{2\epsilon}}{N^\tau},\frac{\sqrt N N^\epsilon}{N^{3\tau}}\Big)\Big),\end{aligned}    
\end{equation}
\begin{equation}\label{eq uv nonsym}
    \begin{aligned}
        &\quad \exp\Big(-Nu\overline v-\frac{N}{2}(u^2+\overline v^2)\Big)\exp\bigg(N\Big(\frac{1}{3}\Big(\frac{u^3}{b_c}+\frac{\overline v^3}{\sqrt c}\Big)-\frac{1}{4}\Big(\frac{u^4}{b_c^2}+\frac{\overline v^4}{c}\Big)+\frac{1}{5}\Big(\frac{u^5}{b_c^3}+\frac{\overline v^5}{c^{3/2}}\Big)\Big)+\frac{2sN^{\frac13}u}{\gamma_c}\bigg)\\
       &=\exp\bigg(\frac{\ii N (2cb_c^2)^{\frac{1}{4}}\sqrt 2aY^3}{3N^{3\tau}} -\frac{\ii N}{5}\frac{(2cb_c^2)^{\frac54}Y^5}{N^{5\tau}}\Big(\frac{1}{b_c^3}-\frac{1}{c^{3/2}}\Big)+\frac{\ii 2s(2cb_c^2)^{\frac{1}{4}}Y}{\gamma_cN^{\tau-1/3}}\bigg)\\
       &\quad \times \exp\bigg(-\frac{(\nu+\overline\eta)^2}{2}-\frac{\sqrt 2 Y^2\sqrt N}{ N^{2\tau}}\big(b_c\overline \eta+\sqrt c\nu\big)-\frac{Y^4 N}{N^{4\tau}}\Big(\frac{a^2}{2}  +b_c\sqrt c \Big)\bigg)\Big(1+{\cal O}\Big(\frac{N^{2\epsilon}}{N^{\tau}},\frac{\sqrt N N^{\epsilon}}{N^{3\tau}},\frac{N^\epsilon}{N^{1/6}}\Big)\Big),
    \end{aligned}
\end{equation}
and
\begin{equation}\label{eq uv nonsym2}
    \begin{aligned}
        &\quad \exp\Big(-Nu\overline v-\frac{N}{2}(u^2+\overline v^2)\Big)\exp\bigg(N\Big(\frac{1}{3}\Big(\frac{\overline v^3}{b_c}+\frac{u^3}{\sqrt c}\Big)-\frac{1}{4}\Big(\frac{\overline v^4}{b_c^2}+\frac{u^4}{c}\Big)+\frac{1}{5}\Big(\frac{\overline v^5}{b_c^3}+\frac{u^5}{c^{3/2}}\Big)\Big)+\frac{2sN^{1/3}\overline v}{\gamma_c}\bigg)\\
&=\exp\bigg(-\frac{\ii N (2cb_c^2)^{\frac{1}{4}}\sqrt 2aY^3}{3N^{3\tau}}+\frac{\ii N}{5}\frac{(2cb_c^2)^{\frac54}Y^5}{N^{5\tau}}\Big(\frac{1}{b_c^3}-\frac{1}{c^{3/2}}\Big)-\frac{\ii 2s(2cb_c^2)^{\frac{1}{4}}Y}{\gamma_cN^{\tau-1/3}}\bigg)\\
&\quad \times \exp\bigg(-\frac{(\nu+\overline\eta)^2}{2}-\frac{\sqrt 2 Y^2\sqrt N}{ N^{2\tau}}\big(\sqrt c\overline \eta+b_c\nu\big)-\frac{Y^4 N}{N^{4\tau}}\Big(\frac{a^2}{2}+b_c\sqrt c \Big)\bigg)\Big(1+{\cal O}\Big(\frac{N^{2\epsilon}}{N^{\tau}},\frac{\sqrt N N^{\epsilon}}{N^{3\tau}},\frac{N^\epsilon}{N^{1/6}}\Big)\Big).
    \end{aligned}
\end{equation}

\end{lemma}
\begin{proof}
The assumption that $\nu$ and $\eta$ are all ${\cal O}(N^\epsilon)$ with $\epsilon+\tau<1/2$ simply means that, in the expressions of $u$ and $v$ \eqref{def further uv}, the term containing $Y$ scales bigger than the term containing $\nu$ or $\eta$ in the limit of large $N$.  The former scales $N^{-\tau}$ while the latter scales $N^{\epsilon-1/2}$.

The proof is done by a straightforward computation, but let us explain the error bounds.  The error comes from ignoring the terms containing $Y \nu^2, Y^3 \nu$ and $Y^4 \nu$.  The other homogeneous terms are smaller than one of these; for example, the term containing $Y^2 \nu^2$ is smaller than $Y^3\nu$ because one of $\nu$ is replaced by $Y$ to go from the former to the latter.   Furthermore, since it is also clear that the term containing $Y^3\nu$ dominates over the term containing $Y^4 \nu$, hence the final error bounds in \eqref{eq uv sym} come from the two terms $Y \nu^2$ and $Y^3 \nu$, which give $N{\cal O}(N^{-\tau}N^{2(\epsilon-1/2)})$ and $N{\cal O}(N^{-3\tau}N^{\epsilon-1/2})$ respectively.

\end{proof}

\subsubsection{Case: \texorpdfstring{$1/6<\tau\leq 1/4$}{}}

If $1/6<\tau\leq 1/4$, let us redefine $\nu$ and $\eta$ such that
\begin{equation}
\begin{aligned}
\nu_\text{old}=\nu-\frac{AY^2 \sqrt N}{\sqrt2 N^{2\tau}},\quad 
\eta_\text{old}=\eta-\frac{AY^2 \sqrt N}{\sqrt2 N^{2\tau}}, 
\end{aligned}\end{equation}
where $\nu_\text{old}$ and $\mu_\text{old}$ refer to the variables, $\nu$ and $\mu$, showing up in \eqref{def further uv}.
Since $\nu_\text{old}$ and $\mu_\text{old}$ are ${\cal O}(N^\epsilon)$ in Lemma \ref{Lemma 5.2} this definition requires that the redefined $\nu$ and $\eta$ are also ${\cal O}(N^\epsilon)$ and $\epsilon \geq 1/2-2\tau.$ Here we assume $Y$ is bounded and finitely away from $0$.

Using Lemma \ref{preker1}, Lemma \ref{Lemma 5.2}, and the facts that \begin{equation}\label{facts uv}
\begin{aligned}
\frac{1}{ N^{1/3}}\Big(\frac{1}{u}+\frac{1}{\overline v}\Big)\sim {\cal O}\big(\frac{N^{4\tau}}{N^{4/3}}\big),\quad 
\frac{1}{ N^{1/3}}\Big(\frac{1}{u}-\frac{1}{\overline v}\Big)\sim {\cal O}\big(\frac{N^{\tau}}{N^{1/3}}\big),   \end{aligned} 
\end{equation}  we have
\begin{equation}\label{eq diff preker 2}
\begin{aligned}
&\overline{\partial}_{\zeta} {\cal K}_{n}(z, \zeta)
=\frac{ N^{3/2}}{\sqrt 2 \pi^{3/2}  }\bigg[\exp\bigg(-\Big(\frac{\nu+\overline \eta}{\sqrt 2}+\frac{(b_c-A) Y^2\sqrt N}{N^{2\tau}}\Big)^2\bigg)-\exp\bigg(-\Big(\frac{\nu+\overline \eta}{\sqrt 2}+\frac{(\sqrt c-A) Y^2\sqrt N}{N^{2\tau}}\Big)^2\bigg)\bigg]\\
&\quad +\frac{ N^{3/2}}{\sqrt 2 \pi^{3/2}  }\bigg[\exp\bigg(-\frac{(\nu+\overline\eta)^2}{2}-\frac{\sqrt 2 Y^2\sqrt N}{ N^{2\tau}}\big((b_c-A)\overline \eta-(A-\sqrt c)\nu\big)-\frac{Y^4 N}{N^{4\tau}}\Big(\frac{a^2}{2}  -(b_c-A)(A-\sqrt c) \Big)\bigg)\\
&\quad \quad\times \exp\Big(\frac{\ii N (2cb_c^2)^{\frac{1}{4}}\sqrt 2aY^3}{3N^{3\tau}} -\frac{\ii N}{5}\frac{(2cb_c^2)^{\frac54}Y^5}{N^{5\tau}}\Big(\frac{1}{b_c^3}-\frac{1}{c^{3/2}}\Big)+\frac{\ii 2s(2cb_c^2)^{\frac{1}{4}}Y}{\gamma_cN^{\tau-1/3}}\Big){\cal O}\Big(\frac{N^{3\tau}}{N}\Big)\\
&\quad+\exp\bigg(-\Big(\frac{ \nu+\overline\eta}{\sqrt2}+\frac{(b_c-A) Y^2\sqrt N}{N^{2\tau}}\Big)^2\bigg){\cal O}\Big(\frac{N^{2\epsilon}}{N^\tau},\frac{\sqrt N N^\epsilon}{N^{3\tau}},\frac{N}{N^{6\tau}},\frac{N^{3\tau}}{N}\Big)\\
&\quad+\exp\bigg(-\frac{(\nu+\overline\eta)^2}{2}-\frac{\sqrt 2 Y^2\sqrt N}{ N^{2\tau}}\big((\sqrt c-A)\overline \eta-(A-b_c)\nu\big)-\frac{Y^4 N}{N^{4\tau}}\Big(\frac{a^2}{2}  -(b_c-A)(A-\sqrt c) \Big)\bigg)\\
&\quad\quad \times \exp\Big(-\frac{\ii N (2cb_c^2)^{\frac{1}{4}}\sqrt 2aY^3}{3N^{3\tau}} +\frac{\ii N}{5}\frac{(2cb_c^2)^{\frac54}Y^5}{N^{5\tau}}\Big(\frac{1}{b_c^3}-\frac{1}{c^{3/2}}\Big)-\frac{\ii 2s(2cb_c^2)^{\frac{1}{4}}Y}{\gamma_cN^{\tau-1/3}}\Big){\cal O}\Big(\frac{N^{3\tau}}{N}\Big) \\
&\quad+\exp\bigg(-\Big(\frac{ \nu+\overline\eta}{\sqrt2}+\frac{(\sqrt c-A) Y^2\sqrt N}{N^{2\tau}}\Big)^2\bigg){\cal O}\Big(\frac{N^{2\epsilon}}{N^\tau},\frac{\sqrt N N^\epsilon}{N^{3\tau}},\frac{N^{\frac13}}{N^{2\tau}},\frac{N^{3\tau}}{N},\frac{N^{\epsilon}}{N^{1/6}}\Big)
\bigg].\end{aligned}   
\end{equation}

Let us define $F_n$ such that 
\begin{equation}\label{eq Fdef1} 
\begin{aligned}
\frac{1}{n}{\cal K}_{n}(z, \zeta)
&=\frac{1}{2 \pi t_c}\bigg[{\rm erfc}\bigg(\frac{(\sqrt c -A)Y^2\sqrt N}{N^{2\tau}}+\frac{\nu+\overline{\eta}}{\sqrt 2}\bigg)-{\rm erfc}\bigg(\frac{(b_c-A) Y^2\sqrt N}{N^{2\tau}}+\frac{\nu+\overline{\eta}}{\sqrt 2}\bigg)+F_n(z,\zeta)\bigg].
\end{aligned}
\end{equation}

Taking the derivative with respect to $\overline\eta$ on both sides, we have
\begin{equation} \label{eq diffFdef1} 
\begin{aligned}
\frac{1}{n}\overline{\partial}_{\eta} {\cal K}_{n}(z, \zeta)
&=\frac{1}{\sqrt 2 t_c\pi^{3/2} }\bigg[\exp \bigg(-\bigg(\frac{(b_c-A)Y^2\sqrt N}{N^{2\tau}}+\frac{\nu+\overline{\eta}}{\sqrt 2}\bigg)^2\bigg)-\exp \bigg(-\bigg(\frac{(\sqrt c-A) Y^2\sqrt N}{N^{2\tau}}+\frac{\nu+\overline{\eta}}{\sqrt 2}\bigg)^2\bigg)\\
&\quad +\overline\partial_{\eta} F_n\bigg].
\end{aligned}
\end{equation}
Moreover, using \eqref{eq diff preker 2} and the following relation,
\begin{equation} 
\begin{aligned}
\frac{1}{n}\overline{\partial}_{\eta} {\cal K}_{n}(z, \zeta)&=\frac{1}{\sqrt N}\frac{1}{n}\overline{\partial}_{\zeta} {\cal K}_{n}(z, \zeta),
\end{aligned}
\end{equation}
we have 
\begin{equation}
\begin{aligned}\label{PartialEtaFn}
\overline\partial_{\eta} F_n
&=\exp\bigg(-\frac{(\nu+\overline\eta)^2}{2}-\frac{\sqrt 2 Y^2\sqrt N}{ N^{2\tau}}\big((b_c-A)\overline \eta-(A-\sqrt c)\nu\big)-\frac{Y^4 N}{N^{4\tau}}\Big(\frac{a^2}{2}  -(b_c-A)(A-\sqrt c) \Big)\bigg)\\
&\quad\times \exp\Big(\frac{\ii N (2cb_c^2)^{\frac{1}{4}}\sqrt 2aY^3}{3N^{3\tau}} -\frac{\ii N}{5}\frac{(2cb_c^2)^{\frac54}Y^5}{N^{5\tau}}\Big(\frac{1}{b_c^3}-\frac{1}{c^{3/2}}\Big)+\frac{\ii 2s(2cb_c^2)^{\frac{1}{4}}Y}{\gamma_cN^{\tau-1/3}}\Big){\cal O}\Big(\frac{N^{3\tau}}{N}\Big)\\
&+\exp\bigg(-\Big(\frac{ \nu+\overline\eta}{\sqrt2}+\frac{(b_c-A) Y^2\sqrt N}{N^{2\tau}}\Big)^2\bigg){\cal O}\Big(\frac{N^{2\epsilon}}{N^\tau},\frac{\sqrt N N^\epsilon}{N^{3\tau}},\frac{N}{N^{6\tau}},\frac{N^{3\tau}}{N}\Big)\\
&+\exp\bigg(-\frac{(\nu+\overline\eta)^2}{2}-\frac{\sqrt 2 Y^2\sqrt N}{ N^{2\tau}}\big((\sqrt c-A)\overline \eta-(A-b_c)\nu\big)-\frac{Y^4 N}{N^{4\tau}}\Big(\frac{a^2}{2}  -(b_c-A)(A-\sqrt c) \Big)\bigg)\\
&\quad\times \exp\Big(-\frac{\ii N (2cb_c^2)^{\frac{1}{4}}\sqrt 2aY^3}{3N^{3\tau}} +\frac{\ii N}{5}\frac{(2cb_c^2)^{\frac54}Y^5}{N^{5\tau}}\Big(\frac{1}{b_c^3}-\frac{1}{c^{3/2}}\Big)-\frac{\ii 2s(2cb_c^2)^{\frac{1}{4}}Y}{\gamma_cN^{\tau-1/3}}\Big){\cal O}\Big(\frac{N^{3\tau}}{N}\Big) \\
&+\exp\bigg(-\Big(\frac{ \nu+\overline\eta}{\sqrt2}+\frac{(\sqrt c-A) Y^2\sqrt N}{N^{2\tau}}\Big)^2\bigg){\cal O}\Big(\frac{N^{2\epsilon}}{N^\tau},\frac{\sqrt N N^\epsilon}{N^{3\tau}},\frac{N^{1/3}}{N^{2\tau}},\frac{N^{3\tau}}{N}\Big) 
.\end{aligned}   
\end{equation}

We define $\zeta_0=b_c+\frac{\ii(2cb_c^2)^{\frac{1}{4}}Y}{N^\tau}-\frac{AY^2}{\sqrt2 N^{2\tau}}+\frac{\eta_0+i \im(\eta)}{\sqrt N}$ for some positive $\eta_0$, such that $\im \zeta=\im \zeta_0$ and $\zeta_0$ sits outside the droplet.  Note that the droplet boundary can be at most ${\cal O}(N^{-2\tau}, N^{-2/3})$ from $b_c$ in the horizontal direction.  The first error comes from the second order approximation of the droplet boundary and the second error comes from moving droplet boundary by ${\cal O}(t-t_c)$.  For $\zeta_0$ to be outside the droplet, $|\eta_0/\sqrt N|$ needs to scale greater than ${\cal O}(N^{-2\tau})$ for $\tau\leq 1/4$.  This can be satisfied if $\eta_0 \sim N^\delta$ with
\begin{equation}
    \delta\geq 1/2 -2\tau.
\end{equation}
Such $\delta$ exists within the range of our consideration because $1/2 -2\tau$, the right hand side, is smaller than $1/2-\tau$, the upper bound of $\epsilon$ in Lemma \ref{Lemma 5.2}. 
As a result $F_n(z,\zeta_0)$ is exponentially suppressed for large $N$ by \eqref{eq Fdef1} and Theorem \ref{Thm 41}.  

Moreover, since $\frac{a^2}{2}  -(b_c-A)(A-\sqrt c)\geq a^2/4$, we obtain that
\begin{equation}\label{eq mix1}
\begin{aligned}
-\frac{(\nu+\overline\eta)^2}{2}-\frac{\sqrt 2 Y^2\sqrt N}{ N^{2\tau}}\big((b_c-A)\overline \eta-(A-\sqrt c)\nu\big)-\frac{Y^4 N}{N^{4\tau}}\Big(\frac{a^2}{2}  -(b_c-A)(A-\sqrt c) \Big)
< -\frac{a^2}{8} \frac{Y^4 N}{N^{4\tau}}
\end{aligned} 
\end{equation}
and
\begin{equation}\label{eq mix2}
\begin{aligned}
-\frac{(\nu+\overline\eta)^2}{2}-\frac{\sqrt 2 Y^2\sqrt N}{ N^{2\tau}}\big((\sqrt c-A)\overline \eta-(A-b_c)\nu\big)-\frac{Y^4 N}{N^{4\tau}}\Big(\frac{a^2}{2}  -(b_c-A)(A-\sqrt c) \Big)< -\frac{a^2}{8} \frac{Y^4 N}{N^{4\tau}}
\end{aligned}    
\end{equation} for a sufficiently large $N$ with $\tau<1/4$. This guarantees that the real exponents in \eqref{PartialEtaFn} are negative.

Assuming $\eta={\cal O}(N^\delta)$, we have
$$\begin{aligned}
|F_n(z,\zeta)| &= \Big|\int_{\eta_{0}}^{\overline\eta}\overline\partial_{\eta} F + F_n(z,\zeta_0)\Big|\le  {\cal O}\Big(\frac{N^{2\epsilon}}{N^\tau},\frac{\sqrt N N^\epsilon}{N^{3\tau}},\frac{N^{1/3}}{N^{2\tau}},\frac{N^{3\tau}}{N}\Big) + |F_n(z,\zeta_0)|= {\cal O}\Big(\frac{N^{2\epsilon}}{N^\tau},\frac{\sqrt N N^\epsilon}{N^{3\tau}},\frac{N^{1/3}}{N^{2\tau}},\frac{N^{3\tau}}{N}\Big).    
\end{aligned}$$ 
Here the error bound is uniform for $0<\epsilon<\min\{\tau/2,3\tau-1/2\}$.

Similarly, we exchange the roles of $\eta$ and $\nu$. Taking $\nu_0>0$ such that $z_0=b_c+\frac{\ii(2cb_c^2)^{\frac{1}{4}}Y}{N^\tau}-\frac{AY^2}{\sqrt2 N^{2\tau}} +\frac{\nu_0}{\sqrt N}$  remains outside the droplet, we have
$$\begin{aligned}
|F_n(z,\zeta)| &= {\cal O}\Big(\frac{N^{2\epsilon}}{N^\tau},\frac{\sqrt N N^\epsilon}{N^{3\tau}},\frac{N^{1/3}}{N^{2\tau}},\frac{N^{3\tau}}{N}\Big)+ |F_n(z_0,\zeta)|= {\cal O}\Big(\frac{N^{2\epsilon}}{N^\tau},\frac{\sqrt N N^\epsilon}{N^{3\tau}},\frac{N^{1/3}}{N^{2\tau}},\frac{N^{3\tau}}{N}\Big).    
\end{aligned}$$
It follows that
\begin{equation}\label{prekernel inside}
\begin{aligned}
\frac{1}{n}{\cal K}_{n}(z, \zeta)
&=\frac{1}{2 \pi t_c}\bigg[{\rm erfc}\bigg(\frac{(\sqrt c -A)Y^2\sqrt N}{N^{2\tau}}+\frac{\nu+\overline{\eta}}{\sqrt 2}\bigg)-{\rm erfc}\bigg(\frac{(b_c-A) Y^2\sqrt N}{N^{2\tau}}+\frac{\nu+\overline{\eta}}{\sqrt 2}\bigg)\bigg]\\
&\quad +{\cal O}\Big(\frac{N^{2\epsilon}}{N^\tau},\frac{\sqrt N N^\epsilon}{N^{3\tau}},\frac{N^{1/3}}{N^{2\tau}},\frac{N^{3\tau}}{N}\Big).
\end{aligned}
\end{equation}

Moreover, by using \eqref{eq unc2} we obtain
\begin{equation}\label{eq relation kk}
 \begin{aligned} \frac{\ee^{Nz\overline\zeta}}{\ee^{\frac{N}{2}|z|^2+\frac{N}{2}|\zeta|^2}}\frac{|z-a|^{Nc}|\zeta-a|^{Nc}}{(z-a)^{Nc}(\overline\zeta-a)^{Nc}}
&=\frac{\exp(Nu\overline v)}{\exp\big(\frac{N}{2}|u|^2+\frac{N}{2}|v|^2\big)}\frac{\exp\big(\frac{Nb_c}{2}(u- \overline u)\big)}{\exp\big(\frac{ Nb_c}{2}(v-\overline v)\big)}\frac{(\sqrt c+\overline u)^{\frac{Nc}{2}}}{(\sqrt c+ u)^{\frac{Nc}{2}}}\frac{(\sqrt c+ v)^{\frac{Nc}{2}}}{(\sqrt c+  \overline v)^{\frac{Nc}{2}}}\\
&=\frac{G(\nu,\eta)}{C_{N,\tau}(X,Y,\nu,\eta)},
\end{aligned}    
\end{equation}
where $G(\nu,\eta):=\ee^{\nu\overline\eta-\frac{|\nu|^2}{2}-\frac{|\eta|^2}{2}}$ is the Ginibre kernel, and 
$$\begin{aligned}
C_{N,\tau}(X,Y,\nu,\eta)&:=\ee^{\frac{ (\overline\nu-\nu+\eta-\overline\eta)( X+\sqrt 2\sqrt N b_c)}{2\sqrt 2}}\ee^{\frac{\nu+\overline\nu-\eta-\overline\eta}{2}\frac{\ii(2cb_c^2)^{\frac{1}{4}} Y\sqrt N}{N^{2\tau}} }\frac{(\sqrt c+  u(\nu))^{\frac{Nc}{2}}}{(\sqrt c+\overline {u(\nu)})^{\frac{Nc}{2}}}\frac{(\sqrt c+  \overline {v(\eta)})^{\frac{Nc}{2}}}{(\sqrt c+ v(\eta))^{\frac{Nc}{2}}}.\end{aligned}$$

Using the relation between the correlation kernel ${\bf K}_n$ and ${\cal K}_{n}$ in \eqref{relation of kernels} and \eqref{eq relation kk}, we get 
\begin{equation}\label{eq thm3 limit1}
\begin{aligned}
\frac{1}{n}{\bf K}_{n}(z, \zeta)
&=\frac{\ee^{Nz\overline\zeta}}{\ee^{\frac{N}{2}|z|^2+\frac{N}{2}|\zeta|^2}}\frac{|z-a|^{Nc}|\zeta-a|^{Nc}}{(z-a)^{Nc}(\overline\zeta-a)^{Nc}} \frac{1}{n}{\cal K}_{n}(z,\zeta)\\
&=\frac{G(\nu,\eta)}{C_{N,\tau}(-AY^2\sqrt N/N^{2\tau},Y,\nu,\eta)}\frac{1}{2 \pi t_c}\bigg[{\rm erfc}\bigg(\frac{(\sqrt c -A)Y^2\sqrt N}{N^{2\tau}}+\frac{\nu+\overline{\eta}}{\sqrt 2}\bigg)\\
&\quad  -{\rm erfc}\bigg(\frac{(b_c -A)Y^2\sqrt N}{N^{2\tau}}+\frac{\nu+\overline{\eta}}{\sqrt 2}\bigg)\bigg]+{\cal O}\Big(\frac{N^{2\epsilon}}{N^\tau},\frac{\sqrt N N^\epsilon}{N^{3\tau}},\frac{N^{1/3}}{N^{2\tau}},\frac{N^{2\tau}}{N^{2/3}}\Big).
\end{aligned}
\end{equation} Here the error is uniform for any $1/6<\tau\leq 1/4$ and for any  $0<\epsilon<\min\{\tau/2,3\tau-1/2\}$. Taking $A=-X/Y^2$ and $n,N\to\infty$, we have \eqref{eq main thm3}.

In particular, when $\tau=1/4$ and $A=-X/Y^2$, we have
\begin{equation}\label{eq middle step}
\begin{aligned}
\frac{1}{n}{\bf K}_{n}(z, \zeta)
&=\frac{G(\nu,\eta)}{C_{N,1/4}(X,Y,\nu,\eta)}\frac{1}{2 \pi t_c}\bigg[{\rm erfc}\Big(X+\sqrt c Y^2+\frac{\nu+\overline{\eta}}{\sqrt 2}\Big)-{\rm erfc}\Big(X+b_c Y^2+\frac{\nu+\overline{\eta}}{\sqrt 2}\Big)\bigg] +{\cal O}\Big(\frac{N^{2\epsilon}}{N^{1/4}},\frac{1}{N^{1/6}}\Big).
\end{aligned}
\end{equation}
 Here the error is uniform over bounded regions $\{(X,Y)\in \RR^2\}$ and $\{(N^{-\epsilon}\nu, N^{-\epsilon}\eta)\in \CC^2\}$ that is finitely away from $Y=0$ with $0<\epsilon< 1/8$.

For the case of $0<\tau\leq \frac{1}{6}$, one can see from \eqref{eq thm3 limit1} that ${\cal O}\big(N^{1/3-2\tau}\big)$-term also contributes. It follows that the leading term and the error bar in \eqref{eq thm3 limit1} need to be further modified with additional corrections. We will not pursue such case in the current paper. See Remark \ref{Remark 1}.


\subsubsection{Case: \texorpdfstring{$1/4<\tau<3/10$}{}}
If $1/4<\tau<3/10$, we denote $\nu$ and $\eta$ in \eqref{def further uv} as
\begin{equation}\label{def munu}
    \begin{aligned}
  \nu:&=x+\frac{\ii \sqrt 2\pi y}{aY^2}\frac{N^{2\tau}}{\sqrt N}-\frac{b_c+\sqrt c}{2\sqrt 2}\frac{Y^2\sqrt N}{N^{2\tau}}+\frac{s}{\sqrt 2 \gamma_c N^{1/6}},\\
\eta:&=x'+\frac{\ii \sqrt 2\pi y'}{aY^2}\frac{N^{2\tau}}{\sqrt N}-\frac{b_c+\sqrt c}{2\sqrt 2}\frac{Y^2\sqrt N}{N^{2\tau}}+\frac{s}{\sqrt 2 \gamma_c N^{1/6}}.
    \end{aligned}
\end{equation}

Since $\nu$ and $\mu$ are ${\cal O}(N^\epsilon)$ in Lemma \ref{Lemma 5.2} the above definition implies that $\epsilon=2\tau-1/2$ and we are allowed to have $x\sim {\cal O}(N^{2\tau-1/2})$ and $x'\sim {\cal O}(N^{2\tau-1/2})$ to be able to use the lemma.

In the discussion below, however, we only need $x$ and $x'$ to be ${\cal O}(N^\delta)$ with an arbitrarily small $\delta>0$.  In fact we will assume $\delta=0$ until we need to turn on nonzero $\delta$. We  assume that $(Y,y,y')\in \RR^3$ in a bounded region and away from $Y=0$.  

Using \eqref{eq uv sym} in Lemma \ref{Lemma 5.2} with $k=\sqrt c$ and $b_c$, respectively,  we have
\begin{equation}\label{eq first}
\begin{aligned}
&\quad \exp\Big(-Nu\overline v-\frac{N}{2}(u^2+\overline v^2)\Big)\exp\bigg(N\Big(\frac{u^3+\overline v^3}{3\sqrt c}-\frac{u^4+\overline v^4}{4c}+\frac{u^5+\overline v^5}{5c^{3/2}}\Big)\bigg)\\
&=\exp\bigg(-\Big(\frac{\nu+\overline \eta}{\sqrt 2}+\frac{b_c Y^2\sqrt N}{ N^{2\tau}}\Big)^2\bigg)\Big(1+{\cal O}\Big(\frac{N^{2\epsilon}}{N^\tau},\frac{\sqrt N N^\epsilon}{N^{3\tau}}\Big)\Big)\\
&=\ee^{-\big(\frac{x+x'}{\sqrt 2}+\frac{\ii \pi (y-y')N^{2\tau}}{aY^2\sqrt N}\big)^2}\ee^{-\ii\pi (y-y')\big(1+\frac{ \sqrt 2  sN^{2\tau}}{aY^2\gamma_c N^{2/3}}\big)}\Big(1-\frac{ aY^2\sqrt N(x+x')}{\sqrt 2N^{2\tau}}+{\cal O}\Big(\frac{N^{3\tau}}{N},\frac{1}{N^{\tau}},\frac{N}{N^{4\tau}}\Big)\Big)\\
&=\ee^{-\big(\frac{x+x'}{\sqrt 2}+\frac{\ii \pi (y-y')N^{2\tau}}{aY^2\sqrt N}\big)^2}\ee^{-\ii\pi (y-y')\big(1+\frac{ \sqrt 2  sN^{2\tau}}{aY^2\gamma_c N^{2/3}}\big)}\Big(1-\frac{ aY^2\sqrt N(x+x')}{\sqrt 2N^{2\tau}}+{\cal O}\Big(\frac{N^{3\tau}}{N},\frac{N}{N^{4\tau}}\Big)\Big)
\end{aligned}    
\end{equation} 
and 
\begin{equation}\label{eq second}
\begin{aligned}
&\quad \exp\Big(-Nu\overline v-\frac{N}{2}(u^2+\overline v^2)\Big)\exp\bigg(N\Big(\frac{u^3+\overline v^3}{3b_c}-\frac{u^4+\overline v^4}{4b_c^2}+\frac{u^5+\overline v^5}{5b_c^{3}}\Big)+\frac{2sN^{1/3}(u+\overline v)}{\gamma_c}\bigg)\\
&=\exp\bigg(-\Big(\frac{\nu+\overline \eta}{\sqrt 2}+\frac{\sqrt c Y^2\sqrt N}{N^{2\tau}}\Big)^2+\frac{2s(\nu+\overline\eta)}{\gamma_c N^{1/6}}\bigg)\Big(1+{\cal O}\Big(\frac{N^{2\epsilon}}{N^\tau},\frac{\sqrt N N^\epsilon}{N^{3\tau}}\Big)\Big)\\
&=\ee^{-\big(\frac{x+x'}{\sqrt 2}+\frac{\ii \pi (y-y')N^{2\tau}}{aY^2\sqrt N}\big)^2}\ee^{\ii\pi (y-y')\big(1+\frac{ \sqrt 2  sN^{2\tau}}{aY^2\gamma_c N^{2/3}}\big)}\Big(1+\frac{ aY^2\sqrt N(x+x')}{\sqrt 2N^{2\tau}}+{\cal O}\Big(\frac{N^{3\tau}}{N},\frac{N}{N^{4\tau}}\Big)\Big).
\end{aligned}    
\end{equation}

Here the error ${\cal O}(N^{3\tau-1},N^{-\tau})$ comes from the error  ${\cal O}(N^{2\epsilon-\tau}, N^{\epsilon+1/2-3\tau})$ with $\epsilon=2\tau-1/2$, the error ${\cal O}(N^{1-4\tau})$ comes from the term containing $Y^4$. The final error bound is obtain by the fact that $1/4<\tau<3/10$.

Similarly, using \eqref{eq uv nonsym} in Lemma \ref{Lemma 5.2} we have
\begin{equation}\label{eq third}
    \begin{aligned}
        &\quad \exp\Big(-Nu\overline v-\frac{N}{2}(u^2+\overline v^2)\Big)\exp\bigg(N\Big(\frac{1}{3}\Big(\frac{u^3}{b_c}+\frac{\overline v^3}{\sqrt c}\Big)-\frac{1}{4}\Big(\frac{u^4}{b_c^2}+\frac{\overline v^4}{c}\Big)+\frac{1}{5}\Big(\frac{u^5}{b_c^3}+\frac{\overline v^5}{c^{3/2}}\Big)\Big)+\frac{2sN^{1/3}u}{\gamma_c}\bigg)\\
&=\ee^{\frac{\ii Y(2cb_c^2)^{1/4}}{N^{\tau-1/3}}\big(\frac{\sqrt2 a Y^2}{3N^{2\tau-2/3}} + \frac{2s}{\gamma_c}\big)}\ee^{-\big(\frac{x+x'}{\sqrt 2}+\frac{\ii \pi (y-y')N^{2\tau}}{aY^2\sqrt N}\big)^2}\ee^{\ii\pi (y+y')\big(1+\frac{ \sqrt 2  sN^{2\tau}}{aY^2\gamma_c N^{2/3}}\big)} \Big(1+\frac{ aY^2\sqrt N(x-x')}{\sqrt 2N^{2\tau}}+{\cal O}\Big(\frac{N}{N^{4\tau}},\frac{N^{3\tau}}{N}\Big)\Big)
    \end{aligned}
\end{equation} 
and using  \eqref{eq uv nonsym2} in Lemma \ref{Lemma 5.2}, we have
\begin{equation}\label{eq forth}
    \begin{aligned}
        &\quad \exp\Big(-Nu\overline v-\frac{N}{2}(u^2+\overline v^2)\Big)\exp\bigg(N\Big(\frac{1}{3}\Big(\frac{\overline v^3}{b_c}+\frac{u^3}{\sqrt c}\Big)-\frac{1}{4}\Big(\frac{\overline v^4}{b_c^2}+\frac{u^4}{c}\Big)+\frac{1}{5}\Big(\frac{\overline v^5}{b_c^3}+\frac{u^5}{c^{3/2}}\Big)\Big)+\frac{2sN^{1/3}\overline v}{\gamma_c}\bigg)\\
&=\ee^{-\frac{\ii Y(2cb_c^2)^{1/4}}{N^{\tau-1/3}}\big(\frac{\sqrt2 a Y^2}{3N^{2\tau-2/3}} + \frac{2s}{\gamma_c}\big)}\ee^{-\big(\frac{x+x'}{\sqrt 2}+\frac{\ii \pi (y-y')N^{2\tau}}{aY^2\sqrt N}\big)^2}\ee^{-\ii\pi (y+y')\big(1+\frac{ \sqrt 2  sN^{2\tau}}{aY^2\gamma_c N^{2/3}}\big)}\Big(1-\frac{ aY^2\sqrt N(x-x')}{\sqrt 2N^{2\tau}}+{\cal O}\Big(\frac{N}{N^{4\tau}},\frac{N^{3\tau}}{N}\Big)\Big).
    \end{aligned}
\end{equation}

Using Lemma \ref{preker1}, the facts in \eqref{facts uv} and above equations, we have
\begin{equation}
\begin{aligned}
\overline{\partial}_{\zeta} {\cal K}_{n}(z, \zeta)
&=\frac{ N^{3/2}\ee^{-\big(\frac{x+x'}{\sqrt 2}+\frac{\ii \pi (y-y')N^{2\tau}}{aY^2\sqrt N}\big)^2}}{\sqrt 2 \pi^{3/2}  }\bigg[\ee^{-\ii\pi (y-y')\big(1+\frac{ \sqrt 2  sN^{2\tau}}{aY^2\gamma_c N^{2/3}}\big)}\Big(1-\frac{ aY^2\sqrt N(x+x')}{\sqrt 2N^{2\tau}}+{\cal O}\Big(\frac{N^{3\tau}}{N},\frac{N}{N^{4\tau}}\Big)\Big)\\
&\quad -\ee^{\ii\pi (y-y')\big(1+\frac{ \sqrt 2  sN^{2\tau}}{aY^2\gamma_c N^{2/3}}\big)}\Big(1+\frac{ aY^2\sqrt N(x+x')}{\sqrt 2N^{2\tau}}+{\cal O}\Big(\frac{N^{3\tau}}{N},\frac{N}{N^{4\tau}}\Big)\Big)+{\cal O}\Big(\frac{N^{3\tau}}{N}\Big)
\bigg] .\end{aligned}   
\end{equation}

Moreover, using the following relation, we have
\begin{equation} 
\begin{aligned}
\frac{1}{n}\partial_{x'} {\cal K}_{n}(z, \zeta)&=\frac{1}{n}\overline{\partial}_{\eta} {\cal K}_{n}(z, \zeta)=\frac{1}{n}\frac{1}{\sqrt N}\overline{\partial}_{\zeta} {\cal K}_{n}(z, \zeta)\\
&=\frac{\ee^{-\big(\frac{x+x'}{\sqrt 2}+\frac{\ii \pi (y-y')N^{2\tau}}{aY^2\sqrt N}\big)^2}}{ \sqrt{2} t_c \pi^{3/2}  }\bigg[\ee^{-\ii\pi (y-y')\big(1+\frac{ \sqrt 2  sN^{2\tau}}{aY^2\gamma_c N^{2/3}}\big)}\Big(1-\frac{ aY^2\sqrt N(x+x')}{\sqrt 2N^{2\tau}}+{\cal O}\Big(\frac{N^{3\tau}}{N},\frac{N}{N^{4\tau}}\Big)\Big)\\
&\quad -\ee^{\ii\pi (y-y')\big(1+\frac{ \sqrt 2  sN^{2\tau}}{aY^2\gamma_c N^{2/3}}\big)}\Big(1+\frac{ aY^2\sqrt N(x+x')}{\sqrt 2N^{2\tau}}+{\cal O}\Big(\frac{N^{3\tau}}{N},\frac{N}{N^{4\tau}}\Big)\Big) +{\cal O}\Big(\frac{N^{3\tau}}{N}\Big)
\bigg].
\end{aligned}
\end{equation}

Let us define $F_n$ such that 
$$\frac{1}{n}{\cal K}_{n}(z, \zeta)=\frac{\ee^{-\big(\frac{x+x'}{\sqrt 2}+\frac{\ii \pi (y-y')N^{2\tau}}{aY^2\sqrt N}\big)^2}}{2\pi t_c}\frac{\ii aY^2\sqrt N}{\pi^{3/2}N^{2\tau}}\bigg(\frac{\ee^{-\ii  \pi(y-y')\big(1+\frac{ \sqrt 2  sN^{2\tau}}{aY^2\gamma_c N^{2/3}}\big)}-\ee^{\ii  \pi(y-y')\big(1+\frac{ \sqrt 2  sN^{2\tau}}{aY^2\gamma_c N^{2/3}}\big)}}{y-y'}+F_n(z,\zeta)\bigg).$$

Taking the derivative with respect to $x'$ on both sides, we have
$$\begin{aligned}
\frac{1}{n}\partial_{x'}{\cal K}_{n}(z, \zeta)
&=\frac{\ee^{-\big(\frac{x+x'}{\sqrt 2}+\frac{\ii \pi (y-y')N^{2\tau}}{aY^2\sqrt N}\big)^2}}{\sqrt 2 t_c \pi^{3/2} } \bigg(\ee^{-\ii  \pi(y-y')\big(1+\frac{ \sqrt 2  sN^{2\tau}}{aY^2\gamma_c N^{2/3}}\big)}-\ee^{\ii  \pi(y-y')\big(1+\frac{ \sqrt 2  sN^{2\tau}}{aY^2\gamma_c N^{2/3}}\big)}\\
&\quad -\frac{1}{\sqrt 2\pi }\frac{\ii aY^2\sqrt N\big(\ee^{-\ii  \pi(y-y')\big(1+\frac{ \sqrt 2  sN^{2\tau}}{aY^2\gamma_c N^{2/3}}\big)}-\ee^{\ii  \pi(y-y')\big(1+\frac{ \sqrt 2  sN^{2\tau}}{aY^2\gamma_c N^{2/3}}\big)}\big)(x+x')}{(y-y')N^{2\tau}}\bigg)\\
&\quad + \partial_{x'}\Big(\ee^{-\big(\frac{x+x'}{\sqrt 2}+\frac{\ii \pi (y-y')N^{2\tau}}{aY^2\sqrt N}\big)^2}\frac{1}{2\pi t_c}\frac{\ii aY^2\sqrt N}{\pi^{3/2}N^{2\tau}} F_n(z,\zeta)\Big)
\end{aligned}$$

It follows that 
\begin{equation}\label{F derivative}
\begin{aligned}
&\quad \partial_{x'}\Big(\ee^{-\big(\frac{x+x'}{\sqrt 2}+\frac{\ii \pi (y-y')N^{2\tau}}{aY^2\sqrt N}\big)^2}\frac{1}{2\pi t_c}\frac{\ii aY^2\sqrt N}{\pi^{3/2}N^{2\tau}} F_n(z,\zeta)\Big)\\
&=\frac{\ee^{-\big(\frac{x+x'}{\sqrt 2}+\frac{\ii \pi (y-y')N^{2\tau}}{aY^2\sqrt N}\big)^2}}{\sqrt 2 t_c \pi^{3/2} } \bigg[\frac{aY^2\sqrt N (x+x')}{\sqrt 2 N^{2\tau}}\bigg(\frac{\ii \big(\ee^{-\ii  \pi(y-y')\big(1+\frac{ \sqrt 2  sN^{2\tau}}{aY^2\gamma_c N^{2/3}}\big)}-\ee^{\ii  \pi(y-y')\big(1+\frac{ \sqrt 2  sN^{2\tau}}{aY^2\gamma_c N^{2/3}}\big)}\big)}{\pi(y-y')}\\
&\quad -\ee^{-\ii\pi (y-y')\big(1+\frac{ \sqrt 2  sN^{2\tau}}{aY^2\gamma_c N^{2/3}}\big)}-\ee^{\ii\pi (y-y')\big(1+\frac{ \sqrt 2  sN^{2\tau}}{aY^2\gamma_c N^{2/3}}\big)}\bigg)+{\cal O}\Big(\frac{N^{3\tau}}{N},\frac{N}{N^{4\tau}}\Big)\bigg]\\
&=\frac{\ee^{-\big(\frac{x+x'}{\sqrt 2}+\frac{\ii \pi (y-y')N^{2\tau}}{aY^2\sqrt N}\big)^2}}{ \sqrt{2} t_c \pi^{3/2}  }{\cal O}\Big(\frac{\sqrt N(y-y')}{N^{2\tau}},\frac{N^{3\tau}}{N}\Big).
\end{aligned}    
\end{equation}

Setting $\widehat{X}:=\frac{\sqrt 2 \pi (y-y')N^{2\tau}}{aY^2 \sqrt N}$, then \eqref{F derivative} reads 
\begin{equation} \label{derivative of Fn}
\frac{\sqrt{N}}{N^{2 \tau}}\partial_{x'} \Big(\ee^{-\frac{1}{2}(x+x'+\ii \widehat{X})^2}F_n\Big) = \ee^{-\frac{1}{2}(x+x'+\ii \widehat{X})^2}\bigg[ \frac{\sqrt{N}}{N^{2\tau}}(x+x')(y-y')\hat g(y-y')+{\cal O}\Big(\frac{N^{3\tau}}{N}, \frac{N}{N^{4\tau}}\Big) \bigg],
\end{equation}
where $\hat g: \CC \to \CC$ is a bounded analytic function. So far we have assumed $x$ and $x'$ in \eqref{F derivative} are bounded. Now we will let them be ${\cal O}( N^\delta)$ for an arbitrarily small  nonzero $\delta>0$, to be able to integrate both sides of the equation \eqref{F derivative} in the $x'$ variable from $x'$ to $N^\delta$.

We estimate the integrals over the first and second summand in \eqref{derivative of Fn} separately. For the second summand we use its size and the Gaussian decay in $x'$ to see that 
\[
\int_{x'}^{N^\delta} \ee^{-\frac{1}{2}(x+t')^2 - \ii (x+t')\widehat{X}}
{\cal O}\Big(\frac{N^{3\tau}}{N}, \frac{N}{N^{4\tau}}\Big)
\dd t' = {\cal O}\Big(\frac{N^{3\tau}}{N}, \frac{N}{N^{4\tau}}\Big)\,.
\]
The second summand does not contribute when $y=y'$. In case $y \ne y'$
 we use a standard stationary phase approximation and that $\widehat{X}\gg 1$ to see that 
 \[
\int_{x'}^{N^\delta} \ee^{-\frac{1}{2}(x+t')^2 - \ii (x+t')\widehat{X}}(x+t') \dd t' = {\cal O}\Big(\frac{N^\delta}{\widehat{X}}\Big)\,.
 \]
 Together we have
$$F_n ={\cal O}\bigg(\frac{|y-y'|N^{\delta}}{\widehat{X}},\frac{N^{5\tau}}{N^{3/2}}, \frac{\sqrt{N}}{N^{2\tau}}\bigg)= {\cal O}\Big(\frac{N^\delta\sqrt N}{N^{2\tau}},\frac{N^{5\tau}}{N^{3/2}}\Big)$$
where $\delta>0$ is arbitrarily small.

Consequently, we have 
$$\frac{1}{n}{\cal K}_{n}(z, \zeta)=\frac{\ee^{-\big(\frac{x+x'}{\sqrt 2}+\frac{\ii \pi (y-y')N^{2\tau}}{aY^2\sqrt N}\big)^2}}{2\pi t_c}\frac{\ii aY^2\sqrt N}{\pi^{3/2}N^{2\tau}}\bigg(\frac{\ee^{-\ii  \pi(y-y')\big(1+\frac{ \sqrt 2  sN^{2\tau}}{aY^2\gamma_c N^{2/3}}\big)}-\ee^{\ii  \pi(y-y')\big(1+\frac{ \sqrt 2  sN^{2\tau}}{aY^2\gamma_c N^{2/3}}\big)}}{y-y'}+
N^\delta
{\cal O}\Big(\frac{\sqrt N}{N^{2\tau}},\frac{N^{5\tau}}{N^{3/2}}\Big)\bigg).$$

Using the relation between the correlation kernel ${\bf K}_n$ and ${\cal K}_{n}$ in \eqref{relation of kernels} and \eqref{eq relation kk},
we get 
\begin{equation}\label{eq thm3 limit2}
\begin{aligned}
\frac{1}{n}{\bf K}_{n}(z, \zeta)
&=\frac{\ee^{Nz\overline\zeta}}{\ee^{\frac{N}{2}|z|^2+\frac{N}{2}|\zeta|^2}}\frac{|z-a|^{Nc}|\zeta-a|^{Nc}}{(z-a)^{Nc}(\overline\zeta-a)^{Nc}} \frac{1}{n}{\cal K}_{n,N}(z,\zeta)\\
&=\frac{1}{\widehat C_{N}(Y,x,y,x',y')}\ee^{-\frac{(x-x')^2}{2}+\frac{\ii\sqrt 2\pi(x+x')(y-y')N^{2\tau}}{aY^2\sqrt N}-\frac{\pi^2(y-y')^2N^{4\tau}}{a^2Y^4N}}\ee^{-\big(\frac{x+x'}{\sqrt 2}+\frac{\ii \pi (y-y')N^{2\tau}}{aY^2\sqrt N}\big)^2}\\
&\times \frac{1}{2\pi t_c}\frac{\ii aY^2\sqrt N}{\pi^{3/2}N^{2\tau}}\bigg(\frac{\ee^{-\ii  \pi(y-y')\big(1+\frac{ \sqrt 2  sN^{2\tau}}{aY^2\gamma_c N^{2/3}}\big)}-\ee^{\ii  \pi(y-y')\big(1+\frac{ \sqrt 2  sN^{2\tau}}{aY^2\gamma_c N^{2/3}}\big)}}{y-y'}+
N^\delta
{\cal O}\Big(\frac{\sqrt N}{N^{2\tau}},\frac{N^{5\tau}}{N^{3/2}}\Big)\bigg)\\
&=\frac{1}{\widehat C_{N}(Y,x,y,x',y')}\frac{a Y^2\sqrt N\ee^{-x^2-(x')^2}}{ \pi^{3/2}t_cN^{2\tau}}\bigg(\frac{\sin \big(\pi(y-y')\big(1+\frac{ \sqrt 2  sN^{2\tau}}{aY^2\gamma_c N^{2/3}}\big)\big) }{\pi(y-y')}+
N^\delta
{\cal O}\Big(\frac{\sqrt N}{N^{2\tau}},\frac{N^{5\tau}}{N^{3/2}}\Big)\bigg),
\end{aligned}
\end{equation} 
where $$\widehat C_{N}(Y,x,y,x',y'):=C_{N,\tau}\Big(-\frac{b_c+\sqrt c}{2}\frac{Y^2\sqrt N}{N^{2\tau}}+\frac{s}{\sqrt 2 \gamma_c N^{1/6}},Y,x+\frac{\ii \sqrt 2\pi y}{aY^2}\frac{N^{2\tau}}{\sqrt N},x'+\frac{\ii \sqrt 2\pi y'}{aY^2}\frac{N^{2\tau}}{\sqrt N}\Big),$$ and the last error bar comes from the fact that $1/4<\tau<3/10$. 
It follows that
\begin{equation}
 \begin{aligned}
\frac{\widehat C_{N}(Y,x,y,x',y')N^{2\tau}}{n\sqrt N}{\bf K}_{n}(z, \zeta)
&=\frac{a Y^2\ee^{-x^2-(x')^2}}{ \pi^{3/2}t_c} \bigg(\frac{\sin \big(\pi(y-y')\big(1+\frac{ \sqrt 2  sN^{2\tau}}{aY^2\gamma_c N^{2/3}}\big)\big) }{\pi(y-y')}+
N^\delta
{\cal O}\Big(\frac{\sqrt N}{N^{2\tau}},\frac{N^{5\tau}}{N^{3/2}}\Big)\bigg).     \end{aligned}
\end{equation} 
The error is uniform over a bounded region $\{Y\in \RR\}$ that is finitely away from $Y=0$ for any $1/4<\tau<3/10$ and  an arbitrarily small $\delta>0$. 

This ends the proof of Theorem \ref{main theorem3}.

For $3/10\leq\tau<1/3$, one can see that ${\cal O}(N^{5\tau+\delta-3/2})$-term in the error becomes larger as $N$ grows. To handle such term we need additional corrections. However we still believe that  the limiting kernel in \eqref{eq main thm32} holds. See Remark \ref{Remark 1}.

\appendix
\section{Relations in the Painlev\'e II Riemann-Hilbert problem}\label{appendix a}

Let $\widetilde\Psi:=\widetilde\Psi(\xi;\hat s)$. Since the Riemann-Hilbert problem \eqref{rhp phi} has constant jump conditions, we write the following lax pair,
\begin{equation}\label{differ 1}
    \frac{\dd}{\dd \xi}\widetilde\Psi=A\widetilde\Psi,\quad  \frac{\dd}{\dd \hat s}\widetilde\Psi=B\widetilde\Psi, 
\end{equation}
where
\begin{equation}
A={A}(\xi,\hat s):=\begin{bmatrix}-4\ii\xi^2-\ii(\hat s+2q^2)&4\xi q+2\ii q'\\
4\xi q-2\ii q'&4\ii\xi^2+\ii(\hat s+2q^2)\end{bmatrix}    
\end{equation} and 
\begin{equation}
{B}={B}(\xi,\hat s):=\begin{bmatrix}-\ii\xi&q\\
q&\ii\xi\end{bmatrix}.    \end{equation} 

The compatibility condition of the linear system \eqref{differ 1} gives
\begin{equation*}
    A_{\hat s}-B_{\xi}+[A,B]=0.\end{equation*} 
It follows that $q(\hat s)$ is the Hastings-McLeod solution of the Painlev\'e II equation $q''=\hat sq+2q^3$.

Let us denote
\begin{equation}\label{def pi xi}
\Pi(\xi;\hat s):=I+\frac{\Pi_1(\hat s)}{2\ii\xi}+\frac{\Pi_2(\hat s)}{\xi^2}+\frac{\Pi_3(\hat s)}{(\ii\xi)^3}+{\cal O}\left(\frac{1}{\xi^4}\right),\end{equation} where
\begin{equation}\label{def pi1 pi2 pi3}
    \Pi_1(\hat s)= \begin{pmatrix}
r(\hat s)&q(\hat s)\\-q(\hat s)&-r(\hat s)
\end{pmatrix},\quad \Pi_2(\hat s)= \begin{pmatrix}
p_{11}(\hat s)&p_{12}(\hat s)\\p_{21}(\hat s)&p_{22}(\hat s)
\end{pmatrix}, \quad \Pi_3(\hat s)= \begin{pmatrix}
 q_{11}(\hat s)&q_{12}(\hat s)\\q_{21}(\hat s)&q_{22}(\hat s)
\end{pmatrix}.
\end{equation}
Since the Painlev\'e II Riemann-Hilbert problem for $\Psi$ in \eqref{rhp phi} only has constant jumps, we have
\begin{equation}\label{def psi differ}
    \frac{\dd \Psi}{\dd \xi}=A(\xi;\hat s)\Psi, 
\end{equation} where the matrix function $A(\xi;\hat s)$ is meromorphic and can be determined by
identifying the singularities. Using $\Pi(\xi;\hat s)$ in \eqref{def pi xi} with \eqref{def pi1 pi2 pi3}, and \eqref{def psi differ}, we have
\begin{equation}\begin{aligned}
A(\xi;\hat s)&=\begin{bmatrix}-4\ii\xi^2-\ii(\hat s+2q^2)&4\xi q-2\ii q r+8\ii p_{12}\\
4\xi q+2\ii q r-8\ii p_{21}&4\ii\xi^2+\ii(\hat s+2q^2)\end{bmatrix} +\frac{A_1(\hat s)}{\xi}+\frac{A_2(\hat s)}{\xi^2}+\frac{A_3(\hat s)}{\xi^3},
\end{aligned} 
\end{equation}
where
\begin{equation}\begin{aligned}
A_1(\hat s)=\begin{bmatrix}4q(p_{12}-p_{21})&q^3-8q_{12}+4p_{12}r+q(\hat s-4p_{22}-r^2)\\
q^3+8q_{21}+4p_{21}r+q(\hat s-4p_{11}-r^2)&-4q(p_{12}-p_{21})\end{bmatrix},
\end{aligned} 
\end{equation}
and the corresponding entries of $A_2$ and $A_3$ are given by
\begin{equation}\begin{aligned}
[A_2(\hat s)]_{11}&=\frac{\ii}{2}(-16p_{21}^2+q^4+r+8p_{21}qr+q^2(\hat s-r^2)),\\
[A_2(\hat s)]_{22}&=-\frac{\ii}{2}(-16p_{21}^2+q^4+r+8p_{21}qr+q^2(\hat s-r^2)),\\
[A_2(\hat s)]_{12}&=\frac{\ii}{2q}(q^4r-16p_{21}(p_{21}r-2q_{12})+8qr(p_{21}r-q_{12})+q^2(1-8q_{22}+\hat sr-r^3)),\\
[A_2(\hat s)]_{21}&=\frac{-\ii}{2q}(q^4r-16p_{21}(p_{21}r+2q_{21})+8qr(p_{21}r+q_{21})+q^2(1+8 q_{11}+\hat sr-r^3)).
\end{aligned} 
\end{equation}
\begin{equation}\begin{aligned}
[A_3(\hat s)]_{11}&=\hat sp_{12}q-4p_{12}p_{22}q+\frac{q^2}{4}+2p_{12}q^3-2q^2\hat q_{11}+8p_{12}q_{21}-2q^2q_{22}+p_{22}q^2r\\
&-\frac{(1+4p_{12}q)r^2}{4}-p_{11}(2+q(4p_{12}-4p_{21}+qr))+p_{21}(-2q^3+8q_{12}+q(-\hat s+4p_{22}+r^2)),\\
[A_3(\hat s)]_{22}&=-(\hat s-4p_{11})(p_{12}-p_{21})q+2(p_{21}-p_{12})q^3-8p_{21}q_{12}-8p_{12}q_{21}-\frac{r^2}{4}\\
&+(p_{12}-p_{21})qr^2+q^2(\frac{1}{4}+2q_{11}+2q_{22}+p_{11}r)+p_{22}(-2-q(-4p_{12}+4p_{21}+qr)).
\end{aligned} 
\end{equation}

By the fact that $A_1(\hat s)=0$, we have
\begin{equation}\label{final 1}
\begin{aligned}
    p_{12}&=p_{21},\\
    p_{11}&=\frac{q^3+8q_{21}+4p_{21}r+q(\hat s-r^2)}{4q},\\
    p_{22}&=\frac{q^3-8q_{12}+4p_{21}r+q(\hat s-r^2)}{4q}.
\end{aligned}\end{equation}

Furthermore, by the fact that $A_2(\hat s)=0$, we have
\begin{equation}\label{final 12}
\begin{aligned}
    r&=-\hat sq^2-q^4+(q')^2,\\
    (q_{12}+q_{21})q'&=q^2( q_{11}+q_{22}).
\end{aligned}\end{equation}

Defining
\begin{equation}
   B(\xi;\hat s):= \frac{\dd \Psi}{\dd \hat s} \cdot \Psi^{-1},
\end{equation} by the similar procedure as above, we have
\begin{equation}\begin{aligned}
B(\xi;\hat s)&=\begin{bmatrix}-\ii\xi&q\\
q&\ii\xi\end{bmatrix} +\frac{B_1(\hat s)}{\xi},
\end{aligned} 
\end{equation}
where
\begin{equation}\begin{aligned}
B_1(\hat s)=\begin{bmatrix}-\frac{\ii(q^2+r')}{2}&\frac{\ii(4p_{12}-qr-q')}{2}\\
\frac{\ii(4p_{21}-qr-q')}{2}&\frac{\ii(q^2+r')}{2}\end{bmatrix}.
\end{aligned} 
\end{equation}

By the fact that $B_1(\hat s)=0$, we have
\begin{align}\label{final 2}
    p_{12}=\frac{qr+q'}{4},\quad r'=-q^2.
\end{align}

Using the fact that $A_3(\hat s)=0$ and the relation $(q_{12}+q_{21})q'=q^2( q_{11}+q_{22})$ from \eqref{final 12}, we have
\begin{align}\label{final 3}
    q_{12}=-q_{21}=\frac{1}{16}(2\hat sq+q^3+\hat s^2q^5+2\hat sq^7+q^9-2q^2(\hat s+q^2)q'-2q^3(\hat s+q^2)(q')^2+2(q')^3+q(q')^4).
\end{align}

Therefore,
\begin{align}
p_{12}(\hat s)=p_{21}(\hat s)&=\frac{q(\hat s)r(\hat s)+q'(\hat s)}{4}, \label{eq p12}\\
 p_{11}(\hat s)=p_{22}(\hat s)&=\frac{2q(\hat s)q'(\hat s)(\hat s+q(\hat s)^2)+q(\hat s)^2-r(\hat s)^2}{8}+\frac{2q'(\hat s)(r(\hat s)-q'(\hat s)^2)}{8q(\hat s)}\\
 &=\frac{q(\hat s)^2-r(\hat s)^2}{8}.\label{eq p11}
\end{align}

\section{Proof of Theorem \ref{thm 21}}\label{appendix b}

Let $p_j$ be the monic orthogonal polynomials satisfy the following orthogonality conditions
\begin{align} \label{Pj monic}
	\int_{ \CC }  p_j(z) \overline{ p_k (z) } |z-a|^{2Nc} \ee^{-N |z|^2}\,\dd A(z)= h_j \, \delta_{jk},
\end{align}
where $h_j$ is the norming constant. 
Let us denote
\begin{equation}
\psi_j(z):= (z-a)^{Nc}p_j(z), \qquad \phi_j (z):=(z-a)^{Nc} \frac{p_j(z)}{h_j}. 
\end{equation}

We define the inner product
\begin{equation}\label{inner product}
	\langle U | V \rangle := \int_\CC \overline{U(z)}\,V(z)\,\ee^{-N|z|^2}\,\dd A(z)=\overline{ \langle V| U \rangle  }
\end{equation}
and denote $\Psi:=[ \psi_0,\psi_1,\cdots ]^{\rm T}$, $\Phi:=[ \phi_0,\phi_1,\cdots ]^{\rm T}$, where ${\rm T}$ is the transpose of a matrix. Let
\begin{equation} \label{projection op}
	\Pi_n:=\textup{diag}(\underbrace{1,\cdots,1}_n,0,\cdots )
\end{equation}
be the projection operator, we write
\begin{equation}
	{\cal K}_n(z,\zeta)=\ee^{-N z\bar{\zeta}} \Phi^*(\zeta) \Pi_n \Psi(z),
\end{equation} where the superscript $*$ means the complex conjugation.
Then we have 
\begin{equation} \label{bfK bp eta}
\overline\partial_\zeta {\cal K}_{n}(z,\zeta) = -N \ee^{ -N z \bar{\zeta} }  \Big(z-\frac{1}{N}\overline\partial_\zeta\Big) 	 \Phi^*(\zeta) \Pi_n \Psi(z).
\end{equation}

Let us denote 
\begin{equation}
L_{j,j-1}=-\frac{\langle z\psi_j | \phi_0 \rangle}{\langle z\psi_{j-1} | \phi_0 \rangle}\quad \mbox{for each $j$}.
\end{equation}
Notice that the denominator does not vanish due to the fact that $p_{n+1}(0)\neq 0$ and $p_n(z)\neq z p_{n-1}(z)$ for all $n$, which is obtained by the following relation  \begin{equation*}
    h_n=-\frac{\Gamma(Nc+n+1)}{2\ii N^{Nc+n+1}}\frac{\widetilde{h}_n}{p_{n+1}(0)}
\end{equation*} with $\widetilde h_n\neq 0$ and $h_n\neq 0$ for all $n$ from \cite[Proposition 3.3, Proposition 7.1]{Ba 2015}.
Note also that 
$$	
z \big( \psi_j(z)+L_{j,j-1} \psi_{j-1}(z)  \big) \,  \bot \, \phi_0
$$
with respect to the inner product \eqref{inner product}.
The numbers $L_{j,j-1}$ are building blocks to define the lower diagonal matrix 
$$
L:=\begin{bmatrix}0&0&0&0&\dots\\
L_{2,1}&0&0&0&\dots\\
0&L_{3,2}&0&0&\dots\\
0&0&L_{4,3}&0&\dots\\
\vdots&\vdots&\vdots&\ddots&\vdots
\end{bmatrix}.
$$

We write 
\begin{equation}\label{def psitilde}
\widetilde{ \psi}_j :=\psi_j+L_{j,j-1} \psi_{j-1}. 
\end{equation}
Then if 
$$
\phi(z)= (\text{polynomials of deg} \le j-2) \cdot (z-a) \cdot (z-a)^{Nc}, 
$$
we have
\begin{equation}
\langle \phi \, |  \, z \widetilde{\psi}_j \rangle =\langle \partial \phi \, | \, \widetilde{\psi}_j \rangle=0.
\end{equation}
It follows that
$$ 
\mathrm{span}\{ \phi_0, \phi_1,\cdots, \phi_{j-1}  \} \, \bot  \, z \widetilde{\psi}_j,
$$
which leads to 
\begin{equation}\label{z psi1}
z \widetilde{\psi}_j(z) 
= \psi_{j+1}(z)+B_{j,j} \, \psi_j(z)
\end{equation}
for some $B_{j,j}$.
Thus we obtain 
\begin{equation} \label{B act}
 z( I+L ) \Psi=B \, \Psi,  \qquad 
B:=\begin{bmatrix}
B_{1,1}&1&0&0&\dots\\
0&B_{2,2}&1&0&\dots\\
0&0&B_{3,3}&1&\dots\\
0 & 0 &0& B_{4,4}&\ddots \\
\vdots & \vdots & \vdots &\ddots&\ddots
\end{bmatrix}.
\end{equation}

Let us also write
\begin{equation} \label{Ujj+1}
 U_{j,j+1}=-\frac{h_{j+1}}{h_j}\frac{p_j(a)}{p_{j+1}(a)}, 
\end{equation}
which is well-defined due to the fact that $h_j\neq 0$ for all $j$ from \cite[Proposition 3.3, Proposition 7.1]{Ba 2015}, and define the upper diagonal matrix 
$$
U:=\begin{bmatrix}
0& U_{1,2}&0&0&\dots\\
0&0& U_{2,3} &0&\dots\\
0&0&0& U_{3,4}&\dots\\
0 & 0 & 0 &0&\ddots\\
\vdots & \vdots & \vdots &\vdots&\ddots
\end{bmatrix}.
$$
Then the function 
\begin{align}
\begin{split}
\label{whphi}
\widehat{\phi}_j(z)&:=\phi_j(z)+U_{j,j+1} \phi_{j+1}(z)
=(\text{polynomials of deg}\le j) \cdot (z-a)\cdot (z-a)^{Nc}
\end{split}
\end{align}
satisfies
\begin{equation} 
\langle \partial \widehat{\phi}_j\, | \, \psi_k \rangle = \langle \widehat{\phi}_j\,| \, z \, \psi_k \rangle=0 \quad \text{if }k\le j-2.  
\end{equation}
Thus we have
\begin{equation}
	\partial \widehat{\phi}_j=A_{j,j} \phi_j+A_{j,j-1} \phi_{j-1}
\end{equation}
for some $A_{j,k}$, equivalently, 
\begin{equation} \label{A act}
	\partial(I+U) \Phi=A\, \Phi, \qquad 
A:=\begin{bmatrix}
A_{1,1}&0&0&0&\dots\\
A_{2,1}&A_{2,2}&0&0&\dots\\
0&A_{3,2}&A_{3,3}&0&\dots\\
0&0&A_{4,3}&A_{4,4}&\dots\\
\vdots&\vdots&\vdots&\ddots&\ddots
\end{bmatrix}.
\end{equation}

We now determine $A_{j,j-1}$ and $A_{j,j}$. Note that integration by parts gives
\begin{align} \label{AB rel intbypart}
\begin{split}
	\overline{B} (I+U)^t&=\overline{B} \langle \Psi \, | \, \Phi^t \rangle (I+U)^t = \langle B \Psi \, | \, \Phi^t (I+U)^t \rangle 
	\\
	&=\langle  z( I+L ) \Psi\,| \, \Phi^t(I+U)^t \rangle 
=\frac{1}{N}\langle ( I+L ) \Psi\,| \, \partial \Phi^t(I+U)^t \rangle  	\\
	&=\frac{1}{N}\langle ( I+L ) \Psi\,| \, \Phi^tA^t \rangle=\frac{1}{N}(I+\overline{L}) A^t.
\end{split}
\end{align}
Thus we obtain the relation
\begin{equation}\label{AB relation}
	\tfrac{1}{N}A(I+L^*)=(I+U)B^*, \qquad 	B=\tfrac{1}{N}(I+L) A^* (I+U^*)^{-1}.
\end{equation}
Comparing the terms involving $A_{j,j-1}$, one can observe that
\begin{equation} \label{Ajj-1}
A_{j,j-1}=N.
\end{equation}
To determine $A_{j,j}$, note that
\begin{align*}
\partial \widehat{\phi}_j(z)&= \partial \Big(  \phi_j+U_{j,j+1} \phi_{j+1} \Big)
=\frac{1}{h_j}\partial \Big( (z-a)^{Nc}p_j \Big)+\frac{U_{j,j+1} }{h_{j+1}}\partial  \Big( (z-a)^{Nc}p_{j+1} \Big)
\\
&=(z-a)^{Nc-1} \frac{1}{h_j} \bigg( \Big( Ncp_j+(z-a)p_j' \Big)-\frac{ p_j(a) }{ p_{j+1}(a) } \Big( Ncp_{j+1}+(z-a)p_{j+1}' \Big) \bigg)
\\
&=(z-a)^{Nc} \frac{1}{h_j} \bigg( \frac{Nc}{z-a} \Big(p_j-\frac{ p_j(a) }{ p_{j+1}(a) } p_{j+1} \Big) + p_j' -\frac{ p_j(a) }{ p_{j+1}(a) } p_{j+1}'  \bigg).
\end{align*}
This gives 
$$
A_{j,j} p_j+N\frac{ h_j }{ h_{j-1} } p_{j-1}= \frac{Nc}{z-a} \Big(p_j-\frac{ p_j(a) }{ p_{j+1}(a) } p_{j+1} \Big) + p_j' -\frac{ p_j(a) }{ p_{j+1}(a) } p_{j+1}' . 
$$
Comparing the coefficient of $z^{j}$ term of this identity, we obtain 
\begin{equation}\label{Ajj}
 A_{j,j}=-  \frac{p_j(a)}{p_{j+1}(a)}  (Nc+j+1) .
\end{equation}
Notice in particular that $A_{j,k}$'s are real. 

Now let us consider the decomposition
\begin{equation}\label{A decomposition}
	A=N\, T_-+A_0,
\end{equation} 
where 
$$T_-:=\begin{bmatrix}0&0&0&0&\dots\\
1&0&0&0&\dots\\
0&1&0&0&\dots\\
0&0&1&0&\dots\\
\vdots&\vdots&\vdots&\ddots&\vdots
\end{bmatrix}, \qquad A_0:=\textup{diag}(A_{1,1},A_{2,2},\dots)
$$
are the translation and the diagonal part respectively. Write
\begin{equation}
A^*=T_++A_0^*, \qquad T_+:=T_-^*.
\end{equation}

Note also that we have
\begin{align*}
&\quad (T_+-\zeta) \Psi=\big(T_+-(I+L)^{-1}B \big) \Psi
=\Big(T_+-(I+L)^{-1}\frac{1}{N}(I+L) A^* (I+U^*)^{-1} \Big) \Psi
\\
&=\Big(T_+-\frac{1}{N}A^* (I+U^*)^{-1} \Big) \Psi=(T_+(I+U^*)-\frac{1}{N}A^* ) (I+U^*)^{-1} \Psi=\Big(T_+U^*-\tfrac{1}{N}A_0^* \Big) (I+U^*)^{-1} \Psi,
\end{align*}
where the second and the fourth identity follow from \eqref{AB relation} and \eqref{A decomposition} respectively.

Let us claim that $(T_+U^*-\tfrac{1}{N}A_0^* )$ is invertible. 
Suppose that this is not the case. Then there exists some $k$ such that
$U^*_{k-1,k}-\frac{1}{N}A_{k-1,k-1}=0$. Consequently, we have $\psi_{k}(z)=z\psi_{k-1}(z)$. This contradicts the assumption that $a\neq 0$.
Therefore, $(T_+U^*-\tfrac{1}{N}A_0^* )$ is invertible, which is also equivalent to 
\begin{equation} \label{h j+1 j}
h_{j+1} \not= \frac{Nc+j+1}{N}h_j
\end{equation}
by \eqref{Ujj+1} and \eqref{Ajj}. 

By letting 
\begin{equation}
	\widehat{\Psi}:=\Big(T_+U^*-\frac{1}{N}A_0^*\Big)^{-1}  (T_+-z)\Psi,
\end{equation}
we have 
\begin{equation}
	(I+U^*) 	\widehat{\Psi}=\Psi. 
\end{equation}

Note that 
\begin{align*}
	\begin{split}
	\widehat{\Psi}&:=[ \widehat{\psi}_0,\widehat{\psi}_1,\cdots ]^t
	=\Big(T_+U^*-\frac{1}{N}A_0^*\Big)^{-1}  (T_+-z)  [ \psi_0,\psi_1,\cdots ]^t
	\\
	&=\textup{diag}\Big(  U^*_{1,2}-\frac{1}{N}A_{1,1},U^*_{2,3}-\frac{1}{N}A_{2,2},\cdots  \Big)^{-1} [ \psi_1-z \psi_0,\psi_2-z \psi_1,\cdots ]^t.
\end{split}
\end{align*}
Thus we have 
\begin{equation} \label{psi hat}
\widehat{\psi}_j=\frac{  \psi_{j+1}-z \psi_j }{ U_{j,j+1}^*-\frac{1}{N} A_{j,j} }.
\end{equation}
Here the denominator again does not vanish due to \eqref{Ujj+1}, \eqref{Ajj} and \eqref{h j+1 j}.
Then by \eqref{B act} and \eqref{AB relation}, we have
\begin{align*}
(I+L) z 	(I+U^*) 	\widehat{\Psi}&=(I+L) z  \Psi = B\Psi 
=B(I+U^*) \widehat{\Psi} =\frac{1}{N}(I+L) A^* \widehat{\Psi},
\end{align*}
which leads to 
\begin{equation} \label{UA Psihat}
  z 	(I+U^*) 	\widehat{\Psi}=\tfrac{1}{N}A^* \widehat{\Psi}. 
\end{equation}

Combining \eqref{UA Psihat}, \eqref{A act} and \eqref{AB relation}, we obtain  
\begin{align*}
\Big(z-\frac{1}{N}\partial_\zeta\Big) 	 \Phi^*(\zeta) \Pi_n \Psi(z)
&=\Big(z-\frac{1}{N}\partial_\zeta\Big) 	 \Phi^*(\zeta) \Pi_n (I+U^*)\widehat{\Psi}(z)
\\
&=\frac{1}{N} \Phi^*(\zeta) [\Pi_n,A^*] \widehat{\Psi}(z)-\frac{1}{N}\partial_\zeta \Phi^*(\zeta)[\Pi_n,I+U^*] \widehat{\Psi}(z).
\end{align*}
Moreover by \eqref{Ajj-1} and \eqref{psi hat}, we have 
\begin{align*}
	\begin{split}
\Phi^*(\zeta) [\Pi_n,A^*] \widehat{\Psi}(z)&=\overline{ \phi_{n-1}(\zeta) } A_{n,n-1} \widehat{\psi}_n(z)-\overline{ \phi_{n}(\zeta) } A_{n-1,n} \widehat{\psi}_{n-1}(z)
\\
&=\frac{N}{  U_{n,n+1}^*-\frac{1}{N} A_{n,n} } \overline{ \phi_{n-1}(\zeta) }  (  \psi_{n+1}(z)-z \psi_n(z)  ).
\end{split}
\end{align*}
Similarly, we obtain 
\begin{align*} 
	\begin{split}
	\partial_\zeta \Phi^*(\zeta)[\Pi_n,I+U^*] \widehat{\Psi}(z)
	&=	\partial_\zeta \overline{ \phi_{n-1}(\zeta) } U_{n,n-1}^* \widehat{\psi}_n(z)-	\partial_\zeta \overline{ \phi_{n}(\zeta) } U_{n-1,n}^* \widehat{\psi}_{n-1}(z)
		\\
		&=    -\frac{U_{n-1,n}^*}{ U_{n-1,n}^*-\frac{1}{N} A_{n-1,n-1} }	\partial_\zeta \overline{ \phi_{n}(\zeta) }( \psi_{n}(z)-z \psi_{n-1}(z) ).
	\end{split}
\end{align*}
Combining all of the above identities with \eqref{bfK bp eta}, the proof is complete.

\section{Proof of Theorem \ref{Thm 41}}\label{appendix c}

Since $\Delta Q(z) =(\partial_x^2+\partial_y^2)Q(z) = 4$ away from the singularity, the function, $2\log|u(z)|+4N |z-z_0|^2-N Q(z)$, is a subharmonic function (and harmonic away from the singularities) for any polynomial $u(z)$.  Exponentiating the subharmonic function we get that
$|u(z)|^2 \ee^{4N|z-z_0|^2-NQ(z)}$ is also a subharmonic function, using Jensen's inequality.  Integrating around a small disk of radius $1/\sqrt N$ centered at $z_0$, that we denote by $D(z_0;1/\sqrt N)$ below, we get
\begin{align}
|u(z_0)|^2\ee^{-NQ(z_0)}&\leq \frac{N}{\pi}\int_{D(z_0;1/\sqrt N)}|u(z)|^2 \ee^{4N|z-z_0|^2-NQ(z)}\dd A(z)
\\&<\frac{N\ee^{4}}{\pi}\int_{D(z_0;1/\sqrt N)}|u(z)|^2 \ee^{-NQ(z)}\dd A(z)
\\
&<\frac{N\ee^{4}}{\pi}\int_{\CC}|u(z)|^2 \ee^{-NQ(z)}\dd A(z).
\end{align}
The above inequality holds for arbitrary $z_0\in\CC$.  Taking the log of the above inequality, we get
$$2\log|u(z)| - N Q(z) \leq \log \frac{N\ee^4}{\pi}\|u\|_{L^2(NQ)}, \qquad z\in\CC,
$$ 
where we define $$\|u\|_{L^2(NQ)}=\int_{\CC}|u(z)|^2 \ee^{-NQ(z)}\dd A(z).$$
Let us define $\widehat Q$, adopting the notation of \cite{Ameur 2010}, by
$$\widehat Q(z)= \frac{2t}{\text{Area}({\cal S})}\int_{\cal S}\log|z-w|\dd A(w)-\ell_{\text{2D}},  $$
where $\ell_{\text{2D}}$ is chosen such that $Q=\widehat Q$ on the boundary of ${\cal S}$.  Since $\widehat Q$ is harmonic in ${\cal S}^c$, and $\widehat Q$ satisfies the growth condition $\widehat Q(z)\sim 2t\log|z|$ as $|z|\to\infty$, we have
$$2\log|u(z)| - N \widehat Q(z) \leq \log \frac{N\ee^4}{\pi}\|u\|_{L^2(NQ)}
$$
on the boundary of ${\cal S}$, and the left hand side goes to $-\infty$ if $u(z)$ is a polynomial of degree less than $tN$. Since the left hand side is harmonic on ${\cal S}^c$ the above inequality holds for all $z\in{\cal S}^c\cup\partial {\cal S}$.

Let us take $n=tN$ and $u(z)=\ee^{\frac{N}{2}Q(z)}{\bf K}_n(z,w)/\sqrt{{\bf K}_n(w,w)}$ which is a polynomial of degree $n-1$.  We have $\|u\|_{L^2(NQ)}=1$ and, by exponentiating the above inequality, we get
$$
\ee^{N(Q(z)-\widehat Q(z))}\frac{|{\bf K}_n(z,w)|^2}{{\bf K}_n(w,w)}\leq \frac{N\ee^4}{\pi}, \quad \text{for}\quad z\in {\cal S}^c\cup\partial{\cal S}\quad\text{and}\quad w\in\CC.
$$
Taking $w=z$ we have
$$
{\bf K}_n(z,z)\leq \frac{N\ee^4}{\pi}\ee^{-N(Q(z)-\widehat Q(z))}, \quad \text{for}\quad z\in {\cal S}^c\cup\partial{\cal S}.
$$ Setting ${\cal U}_\text{2D}(z)=Q(z)-\widehat Q(z)$, we finish the proof Theorem \ref{Thm 41}.

{\bf Acknowledgements:} We thank Yacin Ameur, Sung-Soo Byun, Tom Claeys, and Arno Kuijlaars for helpful discussions.


\begin{thebibliography}{0}
\bibitem{Akemann 2018} G. Akemann, M. Cikovic and M. Venker, Universality at Weak and Strong Non-Hermiticity Beyond the Elliptic Ginibre Ensemble, Commun. Math. Phys., Vol. 362, 1111–1141, (2018).

\bibitem{Akemann 2003} G. Akemann and G. Vernizzi, Characteristic polynomials of complex random matrix models, Nuclear Phys. B,
660(3):532–556, (2003).


\bibitem{Ameur 2011} Y. Ameur, H. Hedenmalm and N. Makarov, Fluctuations of eigenvalues of random normal matrices, Duke Math. J., 159, 31-81, (2011).

 \bibitem{Ameur 2010} Y. Ameur, H. Hedenmalm and N. Makarov, Berezin transform in polynomial
 Bergman spaces, Commum. Pure Appl. Math., 63, 1533-1584, (2010).

\bibitem{Ameur 2020} Y. Ameur, N.G.Kang, N. Makarov and A.Wennman, Scaling limits of random normal matrix processes at singular boundary points, J. Funct. Anal. {278}, 3, 108340, (2020).

 \bibitem{Ameur 2023} Y. Ameur and S.-S. Byun, Almost-Hermitian random matrices and bandlimited point processes. Anal. Math. Phys. 13, 52, (2023).

\bibitem{Arovas 1984} D. Arovas, J. R. Schrieffer, and F. Wilczek, Fractional statistics and the quantum hall effect,
Phys. Rev. Lett. 53 (7), 722, (1984).

\bibitem{Ba 2015} F. Balogh, M. Bertola, S.-Y. Lee, and K.T.-R. Mclaughlin, Strong asymptotics of the orthogonal polynomials
with respect to a measure supported on the plane, Commun. Pure Appl. Math., 68, 112--172, (2015).

\bibitem{Bauerschmidt 2017} R. Bauerschmidt, P. Bourgade, M. Nikula and H.T. Yau, Local Density for Two-Dimensional One-Component Plasma, Commun. Math. Phys. 356, 189–230, (2017).

\bibitem{Bertola 2008} M. Bertola and S.-Y. Lee, First colonization of a spectral outpost in random matrix theory, Constr. Approx., 30, 225-263, (2008).

\bibitem{Bettelheim 2005} E. Bettelheim, O. Agam, A. Zabrodin, and P. Wiegmann, Singularities of the Hele-Shaw flow and shock waves in dispersive media, Phys. Rev. Lett. 95, 244504, (2005).

\bibitem{Bleher 1999} P. Bleher and A. Its, Semiclassical asymptotics of orthogonal polynomials, Riemann-Hilbert problem, and universality in the matrix model,
Ann. Math., 150, 185–266, (1999).


\bibitem{Bleher 2003} P. Bleher and A. Its, Double scaling limit in the random matrix model:
the Riemann-Hilbert approach, Commun. Pure Appl. Math., 56,
433–516, (2003).

\bibitem{ku94 2015} P.M. Bleher and A. Kuijlaars, Orthogonal polynomials in the normal matrix model
with a cubic potential, Adv. Math., 230, 1272-1321, (2012).

\bibitem{Hikami 1998} E. Br\'ezin and S. Hikami, Level spacing of random matrices in an external source, Phys. Rev. E, (3) 58, 7176–7185, (1998).

\bibitem{Hikami 19982} E. Br\'ezin and S. Hikami, Universal singularity at the closure of a gap in a random matrix theory, Phys. Rev. E,  (3) 57, 4140–4149, (1998).

\bibitem{byun 2025} S.-S. Byun, P.J. Forrester,  Progress on the study of the Ginibre ensembles I: GinUE, KIAS Springer
 Series in Mathematics 3 (2025).

\bibitem{byun 2021} S.-S. Byun, S.-Y. Lee, and M. Yang, Lemniscate ensembles with spectral singularity, (arXiv:2107.07221v2), (2023). 

\bibitem{byun 2022} S.-S. Byun and M. Yang, Determinantal Coulomb gas ensembles with a class of discrete rotational symmetric potentials, SIMA, Vol. 55, Iss. 6, (2023). 

\bibitem{Campbell 2024} A. Campbell, G. Cipolloni, L. Erdős, H.-C. Ji, On the spectral edge of non-Hermitian random matrices, (arXiv:2404.17512), (2024). 


\bibitem{Chau 1998} L.-L. Chau and O. Zaboronsky, On the Structure of Correlation Functions in the Normal Matrix Model, Comm. Math. Phys. 196, 203–247, (1998).


 \bibitem{Chau 1992} L.-L. Chau and Y. Yu,  Unitary polynomials in normal matrix models
and wave functions for the fractional quantum Hall effects, Phys. Lett. A, Vol. 167, 5–6, 452-458, (1992).

 \bibitem{Erdos 2021} G. Cipolloni, L. Erd\H{o}s and D. Schr\"oder, Edge universality for non-Hermitian random matrices,
Probab. Theory Relat. Fields, 179, 1–28, (2021).


\bibitem{Claeys 2011} T. Claeys and T. Grava, The KdV hierarchy: universality and a Painlev\'e transcendent, 
 Int. Math. Res. Not., rnr220, (2011).

\bibitem{Claeys 2010} T. Claeys, A. Its, and I. Krasovsky, Higher order analogues of the Tracy-Widom distribution and the Painlev\'e II hierarchy, Commun. Pure Appl. Math. 63, 362-412, (2010).

\bibitem{Claeys 2006}
T. Claeys and A. Kuijlaars,
Universality of the double scaling limit in
random matrix models,
Commun. Pure Appl. Math., Vol. LIX, 1573–1603, (2006).

\bibitem{Claeys 2007} T. Claeys and M. Vanlessen, Universality of a Double Scaling Limit near Singular Edge Points in Random Matrix Models. Commun. Math. Phys. 273, 499–532, (2007)

\bibitem{Deift 1999} P. Deift, Orthogonal polynomials and random matrices: a Riemann-Hilbert approach, Volume 3
of Courant Lecture Notes in Mathematics, New York University Courant Institute of Mathematical
Sciences, 1999.

\bibitem{Deift 1998} P. Deift, T. Kriecherbauer, and K.T-R McLaughlin, New results on the
equilibrium measure for logarithmic potentials in the presence of an
external field, J. Approx. Theory, 95, 388–475, (1998).


\bibitem{DKMVZ 1999} P. Deift, T. Kriecherbauer, K.T-R McLaughlin, S. Venakides, and X. Zhou, Strong asymptotics of orthogonal
polynomials with respect to exponential weights, Commun. Pure and Appl. Math., 52, 1491-1552, (1999).

\bibitem{DKMVZ 19992}  P. Deift, T. Kriechterbauer, K.T-R McLaughlin, S. Venakides, and X.
Zhou, Uniform asymptotics for polynomials orthogonal with respect to
varying exponential weights and applications to universality questions in
random matrix theory, Commun. Pure Appl. Math., 52, 1335–1425, (1999).

\bibitem{deift 1995} P. Deift and X. Zhou, Asymptotics for the Painlev\'e II equation, Commun. Pure and Appl. Math., 48, 277–337, (1995).

\bibitem{Erdos 2020} L. Erd\H{o}s, T. Kr\"uger and D. Schr\"oder, Cusp universality for random matrices I: local Law
and the complex Hermitian case, Commun. Math. Phys., 378, 1203–1278, (2020).

\bibitem{Faddeyeva 1961} V.N. Faddeeva and N.M. Terent’ev, Tables of Values of the Function $w(z)=\ee^{-z^2}\big(1+2\ii\pi^{-1/2}\int_0^z\ee^{t^2}\dd t\big)$ for Complex Argument. Mathematical Tables Series, 11. Pergamon Press, Oxford-London-New York-Paris, 1961.

\bibitem{Flaschka 1980} H. Flaschka and A.C. Newell, Monodromy and spectrum-preserving deformations I, Commun. Math. Phys., 76, 65–116, (1980).

\bibitem{Forrester 2010} P.J. Forrester, Log-Gases and Random Matrices (LMS-34), Princeton University Press, 2010. 

\bibitem{Frohlich 1982} J. Fr\"{o}hlich and D. Ruelle, Statistical mechanics of vortices in an inviscid two-dimensional fluid,
Commun. Math. Phys. 87, 1-36, (1982).

\bibitem{Fyodorov 1982} Y. V. Fyodorov,
H.-J. Sommers and
B. A. Khoruzhenko, Universality in the random matrix spectra in the
regime of weak non-hermiticity, Ann. Henri Poincar\'e, Vol. 68, No. 4, 449-489 (1998).

\bibitem{Ginibre 1965} J. Ginibre, Statistical ensembles of complex, quaternion, and real matrices. J. Math. Phys., 6(3):440–449, (1965).



\bibitem{HM 2013} H. Hedenmalm and N. Makarov,
Coulomb gas ensembles and Laplacian growth, Proc. London Math. Soc.,
106, 859-907, (2013).

\bibitem{Hedenmalm 2017} H. Hedenmalm and A. Wennman,
Planar orthogonal polynomials and boundary universality in the random normal matrix model, Acta Math., Vol. 227, no. 2, 309 – 406, (2021).

\bibitem{Laughlin 1983} R. B. Laughlin, Anomalous quantum hall effect: an incompressible quantum fluid with fractionally charged excitations, Phys. Rev. Lett. 50, (18) 1395, (1983).

\bibitem{Leble 2018} T. Lebl\'e and S. Serfaty, Fluctuations of Two Dimensional Coulomb Gases, Geom. Funct. Anal. 28, 443–508, (2018). 

\bibitem{Lee 2006} S.-Y. Lee, E. Bettelheim and P. Wiegmann, Bubble break-off in Hele–Shaw flows—singularities and integrable structures, Physica D,
219(1), 22-34, (2006). 

\bibitem{Lee 2017} S.-Y. Lee and M. Yang, Discontinuity in the asymptotic behavior
of planar orthogonal polynomials under a perturbation of the
Gaussian weight, Commun. Math. Phys., 355, 303-338, (2017).


\bibitem{Marchioro 1994} C. Marchioro and M. Pulvirenti, Mathematical theory of incompressible nonviscous fluids, Applied Mathematical Sciences (AMS, Vol. 96), (1994).

\bibitem{Minnhagen 1987} P. Minnhagen, The two-dimensional Coulomb gas, vortex unbinding, and superfluid-superconducting films, Rev. Mod. Phys. 59, 1001, (1987).

\bibitem{Onsager 1949} L. Onsager, Statistical hydrodynamics, Nuovo Cim 6 (Suppl 2), 279–287, (1949).

\bibitem{Pastur 1997} L. Pastur and M. Shcherbina, Universality of the local eigenvalue statistics for a class of unitary invariant random matrix ensembles, J. Stat. Phys. 86, 109–147, (1997).

\bibitem{Pastur 2003} L. Pastur and M. Shcherbina, On the edge universality of the local eigenvalue statistics of matrix models, Mat. Fiz. Anal. Geom. 10(3), 335–365, (2003).

\bibitem{Pearcey 1946} T. Pearcey, The structure of an electromagnetic field in the neighborhood of a cusp
of a caustic, Philos. Mag. 37, 311–317, (1946).

\bibitem{Teodorescu 2005} R. Teodorescu, P. Wiegmann, and A. Zabrodin, Unstable fingering patterns of Hele-Shaw flows as a dispersionless limit
of the Kortweg–de Vries hierarchy, Phys. Rev. Lett. 95, 044502, (2005).

\bibitem{Tracy 2006} C. A. Tracy and H. Widom, The Pearcey process, Commun. Math. Phys., 263, 381–400, (2006).



 


\end{thebibliography}
\end{document}